\def\lncs{0}

\documentclass[11pt]{article}
\usepackage[colorlinks]{hyperref}
\hypersetup{linkcolor=cyan,filecolor=cyan,citecolor=cyan,urlcolor=cyan}
\usepackage{pgf,tikz}
\usetikzlibrary{trees, arrows, shapes, snakes}
\usepackage{xspace}
\usepackage{fullpage}
\usepackage{boxedminipage}
\usepackage{epigraph}
\usepackage{framed}
\usepackage[framemethod=tikz]{mdframed}
\usepackage{subfig}
\usepackage{titlesec}
\usepackage{algorithm}
\usepackage{algpseudocode}
\usepackage[sc]{mathpazo}
\usepackage{amsmath}
\usepackage{accents}
\usepackage{caption}
\usepackage{tikz}
\usepackage{boldline}
\usepackage{multirow}
\usetikzlibrary{shapes.geometric}
\usetikzlibrary{arrows}
\usetikzlibrary{arrows.meta}
\usetikzlibrary{patterns}
\usetikzlibrary{shapes.misc}

\newcommand{\local}{${\mathsf{LOCAL}}$}
\newcommand{\congest}{${\mathsf{CONGEST}}$}

\newcommand{\dist}{\mbox{\rm dist}}
\newcommand{\cdist}{\mbox{\rm c-dist}}

\newcommand{\TZ}{\mathsf{TZ}}

\newcommand{\Cluster}{\mathsf{TruncatedBS}}
\newcommand{\Ball}{\mathsf{\textbf{B}}}

\newcommand{\ClusterHop}{\mathsf{TruncatedTZ}}

\newcommand{\BasicSpanner}{\mathsf{BasicSpanner}}
\newcommand{\ImprovedSpannerI}{\mathsf{SpannerShortDist}}
\newcommand{\ImprovedSpannerII}{\mathsf{SpannerLongDist}}
\newcommand{\ImprovedHopsetII}{\mathsf{HopsetsSmallStretch}}
\newcommand{\ImprovedHopsetI}{\mathsf{HopsetsSmallHops}}

\newcommand{\ClusterAndAugment}{\mathsf{SuperClusterAugment}}
\newcommand{\ClusterAndAugmentHop}{\mathsf{ClusterAndAugmentHop}}

\newcounter{note}[section]

\titlespacing*{\section}{0pt}{1.1\baselineskip}{\baselineskip}
\titlespacing*{\subsection}{-5pt}{\baselineskip}{\baselineskip}

\ifnum\lncs=0
\usepackage{amsthm}
\fi
\usepackage{amsmath,amssymb}

\newtheorem{theorem}{Theorem}
\newtheorem{definition}{Definition}

\newtheorem{lemma}{Lemma}

\newtheorem{claim}{Claim}
\newtheorem{observation}{Observation}

\newtheorem{corollary}{Corollary}
\newtheorem{fact}{Fact}

\usepackage{tikz}
\usepackage{relsize}
\usepackage{ctable}
\usepackage{cleveref,aliascnt}

\crefname{claim}{Claim}{Claims}
\crefname{observation}{Observation}{Observations}
\usepackage{tikz}
\usepackage{relsize}
\usepackage{ctable}

\newcommand{\poly}{\mathsf{poly}}

\newcommand{\ignore}[1]{}

\newenvironment{boxfig}[2]{\begin{figure}[#1]\fbox{\begin{minipage}{\linewidth}
				\vspace{0.2em}
				\makebox[0.025\linewidth]{}
				\begin{minipage}{0.95\linewidth}
					{{
							#2 }}
				\end{minipage}
				\vspace{0.2em}
	\end{minipage}}}{\end{figure}}

\begin{document}
	\title{New $(\alpha,\beta)$ Spanners and Hopsets}
	
	\author{Uri Ben-Levy \thanks{The Weizmann Institute of Science, Israel. Email: {\tt uri.benlevy@weizmann.ac.il}.}\and Merav Parter\thanks{The Weizmann Institute of Science, Israel. Email: {\tt merav.parter@weizmann.ac.il}. Supported in part by an ISF grant (no. 2084/18).}}
	\date{}
	
\begin{titlepage}
\maketitle
\thispagestyle{empty}
\begin{abstract}
An $f(d)$-spanner of an unweighted $n$-vertex graph $G=(V,E)$ is a subgraph $H$ satisfying that $\dist_H(u, v)$ is at most $f(\dist_G(u, v))$ for every $u,v \in V$. A simple girth argument implies that any $f(d)$-spanner with $O(n^{1+1/k})$ edges must satisfy that $f(d)/d=\Omega(k/d+1)$. A matching upper bound (even up to constants) for super-constant values of $d$ is currently known only for $d=\Omega((\log k)^{\log k})$ as given by the well known $(1+\epsilon,\beta)$ spanners of Elkin and Peleg, and its recent improvements by [Elkin-Neiman, SODA'17], and [Abboud-Bodwin-Pettie, SODA'18]. 

%
%
%
We present new spanner constructions that achieve a nearly optimal stretch of $O(k/d+1 \rceil)$ for any distance value $d \in [1,k^{1-o(1)}]$ and $d \geq k^{1+o(1)}$. 
We also show more optimized spanner constructions with nearly linear number of edges. Specifically, for every $\epsilon \in (0,1)$ and integer $k\geq 1$, we show the construction of 
$(3+\epsilon, \beta)$ spanners for $\beta=O_{\epsilon}(k^{\log(3+8/\epsilon)})$ and 
$\widetilde{O}_{\epsilon}(n^{1+1/k})$ edges. 

In addition, we consider the related graph concept of \emph{hopsets} introduced by [Cohen, J. ACM '00]. Informally, an hopset $H$ is a weighted edge set that, when added to the graph $G$, allows one to get a path from each node $u$ to a node $v$ with at most $\beta$ hops (i.e., edges) and length at most $\alpha \cdot \dist_G(u,v)$. We present a new family of $(\alpha,\beta)$ hopsets with $\widetilde{O}(k \cdot n^{1+1/k})$ edges and $\alpha \cdot \beta=O(k)$. Turning to nearly linear-size hopsets, we show a construction of $(3+\epsilon,\beta)$ hopset with $\widetilde{O}_{\epsilon}(n^{1+1/k})$ edges and hop-bound of $\beta=O_{\epsilon}(k^{\log(3+9/\epsilon)})$, improving upon 
the state-of-the-art hop-bound of $\beta=O(\log k /\epsilon)^{\log k}$. 
%
%

\end{abstract}
\end{titlepage}

	\pagenumbering{gobble}
	\tableofcontents
	\newpage 
	\pagenumbering{arabic}
	
\section{Introduction}
Compressing the distance metric of an undirected input graph $G=(V,E)$ up to a small approximation, or \emph{stretch} has been subject to an extensive research over the years. An $f(d)$-spanner of a graph $G$ is a subgraph $H \subseteq G$ satisfying that $\dist_H(u,v)=f(\dist_G(u,v))$. Letting $f(d)=k\cdot d$ for some fixed integer $k$ gives the standard \emph{multiplicative spanners} \cite{PelegU87,althofer1993sparse}. More generally, $f(d)=\alpha\cdot d+\beta$ corresponds to $(\alpha,\beta)$ spanners \cite{peleg1989graph,baswana2005new}. 

%
%

Alth{\"o}fer et al. \cite{althofer1993sparse} provided the first tight construction of $(2k-1)$ multiplicative spanners with $O(n^{1+1/k})$ edges. These spanners are believed to provide the optimal size-stretch tradeoff assuming the girth conjecture of Erd\H{o}s \cite{erdos1970extremal}. It has been widely noted, however, that this optimality notion, has some caveats as the girth argument by itself provides a stretch lower bound only for \emph{adjacent} vertex pairs. 
The first indication that one can provide improved stretch for distant vertex pairs was given by the notion of 
$(1+\epsilon, \beta)$ spanners of Elkin and Peleg \cite{ElkinP04}. In their seminal work, they showed that one can compute an $(1+\epsilon, \beta)$ spanner with $O_{\epsilon,k}(n^{1+1/k})$ edges and $\beta=O(\log k/\epsilon)^{\log k}$, for every integer $k$ and $\epsilon \in (0,1)$. This, in particular implies that one can compute $f(d)$-spanners where $f(d)=O(1)$ for $d=\Omega(\log k)^{\log k}$ and with $O_{\epsilon,k}(n^{1+1/k})$ edges.
Recently, Abboud, Bodwin and Pettie \cite{abboud2018hierarchy} showed that this tradeoff is nearly optimal at least for constant values of $k$, ruling out the possibility for obtaining $(1+\epsilon)$ stretch value for considerably closer vertex pairs while keeping the same bound on the number of edges.

Another approach for obtaining improved stretch for non-adjacent pairs was suggested by the \emph{hybrid spanners} of Parter \cite{Parter14}. For every integer $k$, these spanners have $O_k(n^{1+1/k})$ edges and provide non-adjacent pairs a stretch of $k$ rather than $2k-1$. This stretch value is optimal for vertex pairs at distance $2$, assuming the girth conjecture, but does not provide a significant improvement for pairs at distance $d=\omega(1)$. For instance, for pairs at distance $d=\sqrt{k}$, current spanner constructions still provide a stretch of $k$, rather than a stretch of $O(\sqrt{k})$ as might be attainable, or else be proven otherwise.

To summarize, the existing $f(d)$-spanner constructions currently provide a nearly optimal stretch in two extreme regimes: short distances $d=O(1)$ and large distances $d\geq (\log k)^{\log k}$. Our paper zooms into the missing intermediate regime of distances, aiming at providing the ultimate $O(k/d+1)$ stretch for any value of $d$. Since our stretch values are optimal up to constants, these constructions are useful when the stretch parameter $k$ is super-constant, i.e., $k=g(n)$ for some function $g$ of the number of nodes $n$ (e.g., $k=O(\log\log n)$). Also note that a lower bound of $\Omega(k/d+1)$ is unconditional in the girth conjecture, and holds by a simple girth argument.
While obtaining $(O(1), O(k))$ spanners will provide the holy grail stretch of $O(\lceil k /d \rceil)$ for the entire range of distances, our results are asymptotically very close to this goal. That is, our spanner constructions achieve the optimal stretch values (up to constants) for \emph{almost} the entire range of distances. See Fig. \ref{fig:plot-spanner} for a pictorial illustration for the current $f(d)$-spanner constructions with $\widetilde{O}(n^{1+1/k})$ edges. 

\paragraph{Hopsets.} Hopsets are fundamental graph structures introduced by Cohen \cite{cohen2000polylog}. Since their introduction, they have been receiving considerably more attention recently \cite{ElkinN16b,abboud2018hierarchy,HuangP19}, due to their applications to shortest path computation 
in many computational settings e.g., parallel computing \cite{klein1997randomized,cohen2000polylog,miller2015improved,friedrichs2018parallel,ElkinN19new}, dynamic graph algorithms \cite{henzinger2018decremental}, streaming and distributed algorithms \cite{nanongkai2014distributed,henzinger2016deterministic,elkin2016efficient,elkin2017distributed}. 

For an $n$-vertex undirected \emph{weighted} graph $G = (V, E, w)$, a subset of weighted\footnote{The weight of each edge $(u,v) \in H$ is $\dist_G(u,v)$.} edges $H \subset {V\choose 2}$ (not in $G$) is called $(\alpha,\beta)$ hopset, if for any $u, v \in V$, it holds that
\[
\dist_{G}(u, v) \le \dist_{G'}^{(\beta)}(u, v) \le \alpha \cdot \dist_G(u, v), 
\]
where $G' = (V, E\cup H,w')$, and the weight function $w'$ is defined as follows: for every $e\in E$, $w'(e)=w(e)$ and for every $e=(x,y) \in H$, $w'(e)=\dist_G(x,y)$. The distance $\dist^{(\beta)}_{G'}(u, v)$ is the length of the shortest path from $u$ to $v$ that uses at most $\beta$ edges in $G'$. The first $(1+\epsilon,\beta)$ hopset construction by Cohen \cite{cohen2000polylog} had $\widetilde{O}(n^{1+1/k})$ edges and hop-bound of $\beta=O(\log n/\epsilon)^{\log k}$. Elkin and Neiman \cite{ElkinN16b} presented an improved
construction of $(1+\epsilon,\beta_{EN})$ hopsets with $\beta_{EN} = O\left( \log k/\epsilon \right)^{\log k}$ and $\widetilde{O}(n^{1+1/k})$ edges. The state of the art result is by Huang and Pettie \cite{HuangP19} who proved that the emulators by Thorup and Zwick are in fact also $(1+\epsilon,\beta_{EN})$ hopsets with 
$O(n^{1+1/k})$ edges. A similar construction with slightly worse bounds has been independently shown by Elkin and Neiman \cite{elkin2017linear}. 

Klein and Sairam \cite{klein1997randomized} and Shi and Spencer \cite{shi1999time} gave an efficient PRAM algorithm for computing \emph{exact} hopset with hop-bound 
$\beta=O(\sqrt{n}\log n)$ and linear number of edges. Abboud, Bodwin and Pettie \cite{abboud2018hierarchy} showed that any hopset with less than $n^{1 + 1/k - \delta}$ edges for any $\delta>0$ must have $\beta = \Omega\left(1/(k\epsilon)\right)^{k}$. This implies that the $(1+\epsilon,\beta)$ of Elkin and Neiman \cite{ElkinN16b} and Huang and Pettie \cite{HuangP19} are nearly optimal for $k=O(1)$. 

At the other extreme with respect to parameters, Huang and Pettie \cite{HuangP19} observed that the distance oracle of Thorup and Zwick \cite{ThorupZ05} immediately implies $(\alpha,\beta)$ hopsets with stretch $\alpha=2k-1$, hop-bound $\beta=2$ and $O_k(n^{1+1/k})$ edges. In their paper, Huang and Pettie raised the following question concerning the existence of additional hopsets and specifically asked:
\begin{quote}
\emph{Are there other tradeoffs available when $\beta$ is a fixed constant (say 3 or 4), independent of $k$?}
\end{quote}
We answer this question in the affirmative, by presenting a new family of $(\alpha,\beta)$ hopsets. 
For any $\epsilon \in [1/2,1)$ and integer $k$, we give the construction of $(9 k^{\epsilon}, 5 \cdot k^{1-\epsilon})$ hopsets with $O_k(n^{1+1/k})$ edges, in expectation. Thus taking $\epsilon = 1-4/\log k$ gives constant hop-bound of $\beta=80$, and an improved stretch of $\alpha=0.56 k$. We also show a construction of $(k^{\epsilon}, k^{1-\epsilon})$ hopsets for the complementary range of $\epsilon \in (0,1/2)$.
It is important to note that whereas for $(\alpha,\beta)$ spanners with $O(n^{1+1/k})$ edges it must hold that $\alpha+\beta=\Omega(k)$ by a girth argument, this lower bound does not hold for hopsets. For any constant value of $\epsilon$, our new family of $(\alpha,\beta)$ hopsets in fact satisfies that $\alpha \cdot \beta=O(k)$. 
\paragraph{Application to Shortest Path Computation.} 
The efficient computation of $(\alpha,\beta)$ spanners and hopsets lead to some immediate applications for fast shortest path computation, see e.g., \cite{elkin2006efficient} and \cite{ElkinN16b}. 
Interestingly, the $\beta$ parameter of these structures effects not only the \emph{quality} (or approximation) of the solution, but might also determine the \emph{time} complexity. For example, in the streaming model, the number of passes for computing $(1+\epsilon,\beta)$ APSP approximation is linear in $\beta$. This further motivates the study of $(\alpha,\beta)$ spanners with a considerably improved $\beta$ on the expense of having a constant approximation rather than $(1+\epsilon)$ approximation. 
For example, our new $(3+\epsilon, \beta)$ spanners with $\beta=\poly(k)$ leads to $3+\epsilon$ approximation for the APSP problem using only $O_{\epsilon}(\beta)$ passes. This should be compared against the current $(1+\epsilon)$ approximation, but with $O(\log k/\epsilon)^{\log k}$ passes. 

Turning to the distributed computing models, the parameter $\beta$ also determines the \emph{locality} of the $(\alpha,\beta)$ spanner computation. Specifically, an immediate outcome of our constructions is an algorithm that computes an $(O(1), k^{1+o(1)})$ spanner with $k^{1+o(1)}$ \local\ rounds, almost matching the (tight) round complexity of the standard multiplicative $(2k-1)$ spanners (in the latter, all pairs suffer from a multiplicative stretch of $2k-1$).
\vspace{-5pt}
\subsection{Our Contribution.}\vspace{-5pt}
In this paper we provide improved $f(d)$-spanner constructions with nearly optimal stretch value, up to constants, for almost the entire range of distances. Our key result shows:
\begin{mdframed}[hidealllines=true,backgroundcolor=gray!25]
\vspace{-5pt}
\begin{theorem}[Almost Optimal $f(d)$-Spanners]
\label{thm:spanner-all}
For any integer $k\geq 1$, and an unweighted $n$-vertex graph $G=(V,E)$, one can compute an $f(d)$-spanner $H \subseteq G$ with $\widetilde{O}(n^{1+1/k})$ edges such that $f(d)/d=O(k/d+1)$ for any $d \in [1, k^{1-o(1)}]\cup [k^{1+o(1)},n]$. 
\end{theorem} 
\end{mdframed}
This $f(d)$-spanner is almost optimal in the following sense. The stretch of $O(\lceil k/d \rceil)$ is the best possible up to a constant factor based on the girth argument, the size of the spanner is optimal up to logarithmic terms, and the bounded stretch is provided \emph{almost} for the entire distance range, i.e., excluding $d \in [k^{1-o(1)}, k^{1+o(1)}]$. We note that for this ``problematic range" of $[k^{1-o(1)}, k^{1+o(1)}]$ our spanners still provide a considerably improved stretch over previous constructions. 
The spanner of Theorem \ref{thm:spanner-all} is obtained by two separate constructions. 
The first construction, which is also simpler, considers the range of distances $d \in  [1,\sqrt{k}/2]$.
For this distance range we show an $f(d)$-spanner with $f(d)=7k/d$. Using the terminology of 
$(\alpha,\beta)$ spanner, we can say that our spanner is an $(\alpha,\beta)$ spanner with $\alpha=O(\sqrt{k})$ and $\beta=O(k)$. 
\begin{mdframed}[hidealllines=true,backgroundcolor=gray!25]
\vspace{-5pt}
\begin{theorem}[Spanners for Pairs at Dist. $O(\sqrt{k})$]
\label{thm:sqrt}
For any $n$-vertex unweighted graph $G=(V,E)$ and integers $k \geq 1$, and $d \in [1,\sqrt{k}/2]$, there is a subgraph $H\subseteq G$ of expected size $|E(H)|=(\sqrt{k}\cdot n^{1+1/k})$ such that for every pair of vertices $u$ and $v$ at distance $d$ in $G$, it holds that $\dist_H(u,v)\leq 7\cdot k$. Hence, providing a stretch of $7k/d$.
\end{theorem} 
\end{mdframed}
The algorithm of Theorem \ref{thm:sqrt} already provides a stretch of $O(\sqrt{k})$ for all remaining distance values $d \geq \sqrt{k}/2$. Obtaining an improved stretch of  $o(\sqrt{k})$ for $d=\omega(\sqrt{k})$ is considerably more challenging, and requires additional ideas and techniques. This has led to the construction of new $(\alpha, \beta)$ spanners, that p the desired stretch of $O(\lceil k /d \rceil)$ for almost the entire range of distances. 
\paragraph{New $(\alpha,\beta)$ Spanners.}
Our key contribution is in providing a new $(\alpha,\beta)$ spanner that provides a constant stretch already for vertices at distance at least $k^{1+o(1)}$. Up to the extra factor of $o(1)$ in the exponent, this is the best that one can hope for based on a girth argument. In addition, these spanners also settle down the desired stretch of $O( k /d )$ for any $d < k^{1-o(1)}$. 
\begin{mdframed}[hidealllines=true,backgroundcolor=gray!25]
\begin{theorem}
\label{thm:secondspanner}
For any $n$-vertex unweighted graph $G=(V,E)$, any $0<\epsilon<1/2$, and $k\ge 16^{1/\epsilon}$,\footnote{The statement can work for any $k$ upon suffering from larger constants.} there is a  $(8\cdot k^{\epsilon}, 64^{1/\epsilon}\cdot k)$-spanner of $G$ of expected size $O(64^{1/\epsilon}\cdot k\cdot n^{1+1/k}+64^{2/\epsilon}\cdot k^{2}\cdot n)$.
\end{theorem} 
\end{mdframed}
By setting $\epsilon=\Theta(1/\log k)$ in the above, we get an $(\alpha,\beta)$-spanner with $\alpha=O(1)$ and $\beta=k^{1+o(1)}$, hence providing a constant stretch for every distance $d \geq k^{1+o(1)}$. On the other hand, for every constant value of $\epsilon$, we can get $\alpha=O(k^{\epsilon})$ and $\beta=O(k)$, thus providing a stretch of $O(k/d)$ for $d=k^{1-\epsilon}$. Prior to that construction, the known $(\alpha,\beta=O(k))$ spanners are given by the $(k,k-1)$ spanners of Baswana-Kavitha-Kurt-Pettie \cite{baswana2005new}.
Note that the $(1+\epsilon,\beta)$ spanners of Elkin and Neiman \cite{ElkinN19} also provide\footnote{This is implicit in their analysis.} a stretch of $O(\log k)$ for vertices at distance $d\geq k^{\log 10}$. Our construction provides a \emph{constant} stretch for this distance range, while keeping almost the same bound on the number of edges\footnote{In this paper we did not optimize for secondary order factors in the size bound. All our solutions have $O(k^2 \cdot \log n \cdot  n^{1+1/k})=\widetilde{O}(n^{1+1/k})$ edges w.h.p.}. 

While the spanners of Theorem \ref{thm:secondspanner} provide a constant stretch for $d\geq k^{1+o(1)}$, this constant might be large. For that purpose, we also consider spanners that provide a small as possible stretch $\alpha$, while keeping $\beta$ at most polynomial in $k$. For instance, a simplification of the algorithm from Theorem \ref{thm:secondspanner} can also give $(4, \beta)$-spanners with $\beta=k^{\log 11}$. 
\begin{mdframed}[hidealllines=true,backgroundcolor=gray!25]
\begin{lemma}[New $(3+\epsilon,\beta)$ Spanners]\label{lem:spanner-best}
For any $n$-vertex unweighted graph $G=(V,E)$, integer $k$ and constant $\epsilon>0$, one can compute a $(3+\epsilon,\beta)$ spanner $H$ of $G$ for $\beta=O((3+ 8/\epsilon)\cdot k^{\log(3+8/\epsilon)})$ and expected size $O(n^{1+1/k}+(3+ 8/\epsilon)\cdot k^{\log(3+8/\epsilon)}\cdot n)$.
\end{lemma}
\end{mdframed}
Pettie \cite{Pettie09} showed a construction of nearly linear-size spanner that provides a constant stretch of $17$ for vertices at distance $O(\log^4 n)$. The spanners of Lemma \ref{lem:spanner-best} provides a stretch of $4$ for this range of distances, and also work for any $k=O(\log n)$.

\paragraph{New $(\alpha,\beta)$ Hopsets.}
Hopsets, the cousins of $(\alpha,\beta)$ spanners and emulators, have received quite a lot of attention recently from the graph theoretical and the algorithmic prescriptive. Currently, hopsets constructions are known only for a narrow regime of $\alpha,\beta$ values. A particular setting that attracted a lot of attention is where $\alpha=(1+\epsilon)$. Since all existing $(1+\epsilon,\beta)$ constructions provide a fairly large hop-bound, in this paper we resort to the constant stretch of $\alpha=O(1)$. We show that this relaxation can significantly reduce the hop-bound. Specifically, we discover a new family of hopsets that is technically related to the spanner constructions described above. 
\begin{mdframed}[hidealllines=true,backgroundcolor=gray!25]
\begin{theorem}[New $(k^{\epsilon}, k^{1-\epsilon})$ Hopsets]
\label{thm:new-hopset} 
For any $n$-vertex weighted graph $G=(V,E,w)$, integer $k\geq 1$ and $\epsilon \in (0,1)$ such that $k^{\epsilon}\ge 16$, one can compute an $(\alpha,\beta)$ hopset $H$ for $\alpha=O(k^{\epsilon})$ and $\beta=(c^{1/\epsilon} \cdot k^{1-\epsilon})$ for some constant $c>1$. The number of edges is bounded by $|E(H)|=O((k^{\epsilon}\cdot n^{1+1/k}+k^{\epsilon}/\epsilon\cdot n) \log \Lambda)$ edges in expectation, where $\Lambda$ is the aspect ratio of the graph\footnote{The ratio between the largest and smallest distances between vertex pairs in $G$.}.
\end{theorem}
\end{mdframed}
For example, by setting $\epsilon=\Theta(1/\log k)$, we get an $(\alpha,\beta)$ hopset with constant stretch of $\alpha=O(1)$, and an almost linear hop-bound $O(k^{1+o(1)})$. This brings us very close to the ultimate holy-grail construction of $(O(1),k)$ hopsets. 

We also consider the other direction of minimizing the stretch $\alpha$ as much as possible, while keeping the hop-bound $\beta$ to be polynomial in $k$. As with spanners, this setting is considerably simpler (compared to that of Thm. \ref{thm:new-hopset}), and we have the following:
\begin{mdframed}[hidealllines=true,backgroundcolor=gray!25]
\begin{lemma}[New $(3+\epsilon,\beta)$ Hopsets]\label{lem:hopset-best}
For any $n$-vertex weighted graph $G=(V,E,w)$, integer $k$ and $\epsilon>0$, one can compute a $(3+\epsilon,\beta)$ hopset $H$ where $\beta=16\cdot (3+ 9/\epsilon)\cdot k^{\log(3+9/\epsilon)}$ with expected size $|E(H)|=O((n^{1+1/k}+\log{k}\cdot n) \cdot \log{\Lambda})$.
\end{lemma}
\end{mdframed}
An interesting reference point is for $k=O(\log n)$, i.e., where the hopset has a nearly linear size of $\widetilde{O}(n)$ edges. In this setting, Lemma \ref{lem:hopset-best} gives for example stretch $\alpha=4$ and $\beta=O((\log n)^{3.6})$. Lemma \ref{lem:hopset-best} gives a constant stretch but with hop-bound $\beta=(\log n)^{1+o(1)}$. This should be compared with the $(1+\epsilon,\beta)$ hopsets of \cite{HuangP19} and \cite{elkin2016efficient} that provide a hop-bound of $O_{\epsilon}(\log\log n)^{\log\log n}$. 
See Figure \ref{fig:hopplot} for comparison with existing work.

\textbf{Remark.} We note that Gitlitz, Elkin and Neiman \cite{GENPRAMDO19} independently provided different constructions for $(3+\epsilon,\beta)$ spanners and hopsets with slightly larger $\beta$ vlues than those obtained in Theorems \ref{lem:spanner-best} and \ref{lem:hopset-best}.
%

\paragraph{Applications to Shortest Paths.}
We also show the efficient computation of our simplified $(O(1), \beta)$ spanners and hopsets in various computational settings. This has direct implications to APSP computation.  Elkin and Neiman \cite{ElkinN19,ElkinN16b} specified the implementation details of their hopsets and spanners, along with some immediate applications to shortest paths computation in several computational settings. In the non-centralized settings (e.g., distributed, streaming, etc.), the value of $\beta$ effects not only the approximation quality of the solution, but rather also determines the \emph{locality} of the problem. Our simplified $(3+\epsilon,\poly(k))$ spanners and hopsets is similar, implementation-wise\footnote{I.e., the steps that determine the computational cost are quite similar.}, to the $(1+\epsilon, O(\log k)^{\log k})$ spanners and hopsets of \cite{ElkinN19} and \cite{ElkinN16b}. Therefore we can get these applications, almost for free, while enjoying an improved running time due to our improved $\beta$, upon suffering from a slightly larger stretch of $(3+\epsilon)$ in the centralized regime and $(4+\epsilon)$ in the distributed and streaming regimes rather than $(1+\epsilon)$.  
For example, we can have the following:
\begin{mdframed}[hidealllines=true,backgroundcolor=gray!25]
\begin{lemma}[Approx. APSP]\label{lem:apsp-stream}
For every $n$-vertex unweighted graph and any parameters  $\epsilon>0,\rho \in (0,1)$, there exists a streaming algorithm that computes a $(4+\epsilon,\beta)$ approximation for the APSP in the multi-pass streaming model for $\beta = O((5+16/\epsilon)^{\log{\rho}+2/\rho}\cdot k^{\log(5+16/\epsilon)})$ in either (1) $O(n^{1+\rho}\cdot \log{n})$ space with high probability and $O(\beta)$ passes, or (2) with $O(n^{1+1/k}+(\beta+\log{n}) \cdot n)$ space in expectation and $O(n^{\rho}/\rho \cdot \log{n} \cdot \beta)$ passes with high probability.
\end{lemma}
\end{mdframed}
This is the analogue of Cor. 21 in \cite{ElkinN19} only that they have $(1+\epsilon)$ approximation with $\beta_{EN}=O(\log k)^{\log k}$ passes, rather than $(4+\epsilon)$ approximation with $O_{\epsilon}(k^{\log(5+16/\epsilon)})$ passes. 
%
\subsection{Technical Overview.}\vspace{-5pt}
Throughout, we consider a fixed stretch parameter of $k$, and restrict the number of edges in the output spanners and hopsets to $O_{k}(n^{1+1/k})$ edges in expectation. 
\vspace{-10pt}
\subsubsection{New Spanners}
The starting point for our algorithms is the observation that existing multiplicative spanner constructions (e.g., Baswana-Sen \cite{BaswanaS07}, Thorup-Zwick \cite{ThorupZ05}) provide a considerably improved multiplicative stretch, i.e., of $o(k)$, for edges incident on \emph{sparse} vertices. By sparse, we mean vertices whose $o(k)$-ball contains a \emph{small} number of vertices. To be more concrete, we start by describing some useful properties of the Baswana-Sen Algorithm.   
\paragraph{Useful Properties of the Baswana-Sen Algorithm \cite{BaswanaS07}.}
The Baswana-Sen algorithm consists of $k$ steps of clustering. 
A \emph{clustering} $\mathcal{C}=\{C_1,\ldots, C_\ell\}$ is a collection of vertex disjoint sets which we call \emph{clusters}. Every cluster has some special vertex which we call the \emph{cluster center}. The set of clustered vertices is $V(\mathcal{C})=\bigcup_{i=1}^{\ell} C_i$. 
In the high level, the Baswana-Sen algorithm computes $k$ levels of clustering $\mathcal{C}_0, \ldots, \mathcal{C}_{k-1}$. In each clustering step $i$, given the clustering $\mathcal{C}_{i-1}$, the algorithm computes a clustering $\mathcal{C}_i$, along with a subset of edges $H_i$ that ``takes care" of the newly unclustered vetices (i.e., those that belong to a cluster in $\mathcal{C}_{i-1}$, but do not belong to the clusters of $\mathcal{C}_i$). The clustering $\mathcal{C}_i$ and the output sugbraph $H_i$ have the following useful properties:
\begin{itemize}
\item{(1)}  $|\mathcal{C}_i|=O(n^{1-i/k})$ and $|E(H_i)|=O(n^{1+1/k})$ in expectation. 
\item{(2)}  The radius of each cluster $C \in \mathcal{C}_i$ is at most $i$. Specifically, for each cluster $C \in \mathcal{C}_i$, the subgraph $H_i$ contains a tree $T_C\subseteq G[C]$ rooted at the cluster center of $C$, spanning all the vertices in $C$ and has depth at most $i$. 
\item{(3)} For every unclustered vertex $u \notin V(\mathcal{C})$, $\dist_{H_i}(u,v)\leq 2i-1$ for every $v \in N(u)$. 
\end{itemize}
\paragraph{Warming Up: $(\alpha,\beta)$ Spanners with $\alpha=O(\sqrt{k})$ and $\beta=k$.} 
To illustrate the essence of our constructions, we start by showing an algorithm for computing an $(\alpha,\beta)$ spanner with $\alpha=O(\sqrt{k})$ and $\beta=k$. As we will see, with this approach, a stretch of $\alpha=O(\sqrt{k})$ is the best that one can get for $\beta=k$. Obtaining  the ultimate stretch of $O(\lceil k /d \rceil)$ for every $d=\omega(\sqrt{k})$ will require additional ideas, and a considerably more delicate analysis.

The \textbf{first phase} of the algorithm applies a truncated variant of the Baswana-Sen algorithm in which only the first $\sqrt{k}$ clustering steps (out of the $k$ many steps) are applied. As a result, we get a cluster collection $\mathcal{C}=\mathcal{C}_{\sqrt{k}}$ containing $O(n^{1-1/\sqrt{k}})$ clusters (in expectation), and a subgraph $H'=\bigcup_i^{\sqrt{k}} H_i$ that \emph{takes care} of all edges incident to the non-clustered vertices (i.e., vertices not in $\mathcal{C}$) by Property (3).

In the \textbf{second phase}, the algorithm computes a cluster-graph $\widehat{G}$ whose nodes correspond to the clusters of $\mathcal{C}$. Every two clusters $C, C' \in \mathcal{C}$ are connected by an edge in $\widehat{G}$ iff $\dist_G(r(C),r(C'))\leq 3\sqrt{k}$ where $r(C),r(C')$ are the centers of the clusters $C,C'$ respectively.
The algorithm then computes a $(2\sqrt{k}-1)$ multiplicative spanner $\widehat{H}$ on this cluster graph. The edges of this spanner are translated to $G$-edges as follows. For each edge $(C,C')$ in the spanner $\widehat{H}$, the shortest path in $G$ between $r(C)$ and $r(C')$ is added to the spanner $H$. By property (3), $\widehat{G}$ contains $N=O(n^{1-1/\sqrt{k}})$ nodes, and thus $\widehat{H}$ contains $O(N^{1+1/\sqrt{k}})=O(n)$ edges in expectation. Overall, this step adds $O(\sqrt{k} n)$ edges to the spanner. 

\textbf{The stretch argument:} For the sake of this intuitive explanation, we will show that the spanner provides a stretch of $O(\sqrt{k})$ for vertex pairs at distance $\sqrt{k}$. 
Fix such pair $u,v$ and let $P$ be their shortest path in $G$. 
If $P$ has at most one clustered vertex (i.e., vertex appearing in the clusters of $\mathcal{C}$), then all the edges in $P$ are incident to non-clustered vertices, and thus by 
property (3), for each edge $(x,y) \in P$, $\dist_{H}(x,y)\leq 2\sqrt{k}-1$.

Otherwise, let $u'$ and $v'$ be the far-most clustered vertices on the path $P$, where $u'$ (respectively, $v'$) is the closest clustered vertex to $u$ (respectively, $v$). By property (3) again, all edges on the path segments $P[u,u']$ and $P[v',v]$ enjoy a stretch of at most  $(2\sqrt{k}-1)$ in the spanner. It remains to consider the segment $P[u',v']$. Let $C_{u'}, C_{v'}$ be the clusters of $u',v'$ respectively in $ \mathcal{C}$. Since the radius of these clusters is at most $\sqrt{k}$ (by property (2)), we have that $\dist_G(r(C_{u'}),r(C_{v'}))\leq 3\sqrt{k}$ and thus $C_{u'}$ and $C_{v'}$ are neighbors in $\widehat{G}$. Since $\widehat{H}$ is a $(2\sqrt{k}-1)$ spanner of $\widehat{G}$, we have $\dist_{\widehat{H}}(C_{u'},C_{v'})\leq 2\sqrt{k}-1$. Finally, as each edge $(C,C')$ in $\widehat{H}$ is translated into a path of length $\leq 3\sqrt{k}$ in $G$, we have that 
$u'$ and $v'$ are connected in $H$ by a path of length $O(k)$ in $G$, concluding that $\dist_H(u,v)=O(k)$ as desired. The complete algorithm appears in Sec. \ref{sec:spanner-close}.

\paragraph{The challenge in obtaining a multiplicative stretch $\alpha=o(\sqrt{k})$.}
We note that this algorithm, as is, cannot be extended to provide an improved stretch of $o(\sqrt{k})$ for distances $d\geq \sqrt{k}$ for the following reason. Let $i$ be a parameter that determines the number of the Baswana-Sen steps applied in the first phase of the algorithm. Then, after applying $i$ steps of the Basawna-Sen algorithm, we end with $n^{1-i/k}$ clusters in $\mathcal{C}_i$. By property (3), the stretch obtained for all unclustered vertices (i.e., not in $\mathcal{C}_i$) is bounded by $2i-1$. 
In the second phase, a cluster graph $\widehat{G}$ with $|\mathcal{C}_i|$ nodes is defined. Each node of  $\widehat{G}$ corresponds to a cluster in $\mathcal{C}_i$, and two clusters $C,C' \in \mathcal{C}_i$ are connected in $\widehat{G}$, if their distance in $G$ is at most $3i$. To keep the number of edges in the final spanner small\footnote{Any standard $(2k-1)$-spanner on $N$ vertices contains $O(N^{1+1/k})$ edges.} the algorithm can only afford the computation of  
$(2(k/i)-1)$-multiplicative spanner $\widehat{H} \subseteq \widehat{G}$. Each edge in this spanner $\widehat{H}$ is translated to a path of length $\leq 3i$ in the final spanner. Thus the algorithm adds at most $O(i \cdot n)$ edges in this phase. Overall, we get a stretch of $2i-1$ on all the edges incident to the non-clustered vertices (those that do not appear in $\mathcal{C}$), and a stretch of $\Theta(k/i)$ between every pair of \emph{clustered} vertices at distance at most $i$ in $G$. The optimal stretch is therefore achieved for $i=\Theta(\sqrt{k})$. 

In the next paragraph, we explain how to bypass this obstacle by adding a crucial intermediate phase to the algorithm.  

\paragraph{A New Three Stage Approach for $(\alpha,\beta)$ Spanners.} We next explain the high level ideas to obtain an $(\alpha,\beta)$ spanner with a multiplicative stretch $\alpha=O(k^{\epsilon})$ and an additive stretch $\beta=c^{1/\epsilon} \cdot k$ for some constant $c>1$ and any $\epsilon \in (0,1/2)$. 
By taking $\epsilon=\Theta(1/\log k)$, it provides a \emph{constant} multiplicative stretch for all pairs at distance $d \geq k^{1+o(1)}$. By setting $\epsilon=o(1)$, we get a multiplicative stretch of $O(k/d)$ for all pairs at distance $d < k^{1-o(1)}$. 
In the following, we zoom into a fixed distance value $d \in [1,k^{1-\epsilon}]$ and describe the high-level construction of an $f(d)$-spanner $H$ with $f(d)=O(k/d)$. The same procedure will be repeated for every $d \in [1,k^{1-\epsilon}]$ (in fact, it will be sufficient to repeat it for every class of distances $[d,2d]$). 
The algorithm has three phases. 
The \textbf{initial clustering stage} applies a truncated Baswana-Sen algorithm, running only its first $\lceil k^{\epsilon} \rceil$ steps.
This results in a clustering $\mathcal{C}_0$ containing $n^{1-1/k^{1-\epsilon}}$ clusters in expectation of radius $\lceil k^{\epsilon} \rceil$, and a subset of edges $H_0$. Property (3) of the Baswana-Sen algorithm guarantees a stretch of $O(k^{\epsilon})$ on all edges incident to the unclustered vertices (i.e., vertices not in $\mathcal{C}_0$). 
Note that at this point the number of clusters is too large to be able to compute a $k^{\epsilon}$-spanner on the cluster graph, and terminate. The purpose of the next stage is to rapidly reduce the number of clusters to the ultimate number of $n^{1-1/k^{\epsilon}}$ clusters while keeping the radius of this clustering bounded by $O_{\epsilon}(k^{1-\epsilon})$. 

The \textbf{intermediate superclustering stage} is the most delicate part of the algorithm. It consists of $T=O(1/\epsilon)$ phases of superclustering. This step is similar in flavor to the $(1+\epsilon,\beta)$ spanner construction by Elkin-Neiman \cite{ElkinN19}, with several key differences that we explicitly state.

A \emph{supercluster} $SC=\{C_1,\ldots, C_\ell\}$ is a collection of vertex-disjoint clusters. Every supercluster $SC$ has a special vertex $r(SC)$, that is denoted as the supercluster's center. We denote be $V(SC)$ the set of vertices in the supercluster, that is $V(SC)=\bigcup_{C\in SC}V(C)$. A \emph{superclustering} $\mathcal{SC}=\{SC_1,\ldots, SC_{\ell'}\}$ is a collection of vertex-disjoint superclusters.  
The radius of a supercluster $SC_\ell$ is defined by $\max_{u\in V(SC_{\ell})}\dist_{G}(u,r(SC_{\ell}))$.
Each phase $i \in \{1,\ldots, T\}$ of the superclustering procedure starts with a clustering $\mathcal{C}_{i-1}$ with $N_{i-1}=O(n^{1-k^{\epsilon\cdot (i-1)}/k^{1-\epsilon}})$ clusters of radius $r_{i-1}=O((2k^{\epsilon})^{i})$. The output of the phase is a clustering $\mathcal{C}_{i}$ with $N_i$ clusters and radius $r_i$, as well as a collection of edges $H_i$ added to the spanner that \emph{takes care} of the vertices that stopped being clustered at that phase (i.e., appearing in $\mathcal{C}_{i-1}$ but not in $\mathcal{C}_{i}$). We now describe the high level structure of this $i^{th}$ phase. 

The phase has $t=O(k^{\epsilon})$ steps of superclustering. Starting with the $0^{th}$ superclustering 
$\mathcal{SC}_{i-1,0}=\{C ~\mid~ C \in \mathcal{C}_{i-1}\}$, each step $j \in \{1,\ldots, t\}$ gets as input a superclustering $\mathcal{SC}_{i-1,j-1}$, where the radius of each supercluster is bounded by $r_{i-1,j-1}=O(k^{\epsilon}\cdot e^{4j/k^{\epsilon}}\cdot r_{i-1,0})$. Initially, $r_{i-1,0}$ is simply the radius of the clusters in $\mathcal{C}_{0}$. 
It then outputs a new superclustering $\mathcal{SC}_{i-1,j}$ by applying the following sequence of operations:
\begin{itemize}
\item{\textbf{Augmentation:}} We set an augmentation parameter $\alpha_{i-1,j}=O(r_{i-1,0} \cdot e^{4j/k^{\epsilon}})$. Each \emph{vertex} at distance at most $r_{i-1,j-1}+\alpha_{i-1,j}=O(e^{4j/k^{\epsilon}}\cdot k^{\epsilon}\cdot r_{i-1,0})$ from a center of some supercluster in the current superclustering, adds its shortest path to its closest center to the spanner $H$ (without being added to that supercluster). 

\item{\textbf{Sub-sampling:}} Each supercluster $SC \in \mathcal{SC}_{i-1,j-1}$ gets sampled into $\mathcal{SC}'$ with probability of $N_{i-1}/n$ where $N_{i-1}$ is the number of clusters in $\mathcal{C}_{i-1}$.
\item{\textbf{New Superclustering:}} Each cluster (belonging to any supercluster $SC \in \mathcal{SC}_{i-1,j-1}$) at (center) distance at most $r_{i-1,0}+2\alpha_{i-1,j}+r_{i-1,j-1}$ from a center of a sampled supercluster, joins its closest sampled supercluster and adds the shortest path between their centers to the spanner. The center of the sampled supercluster maintains its role.  
\item{\textbf{Handling Lost Clusters:}} Each other cluster $C$ (that is too far from the sampled superclusters), adds to the spanner $H$ a shortest path from its center to the center of every supercluster $SC \in \mathcal{SC}_{i-1,j-1}$ provided that their distance at most $r_{i-1,0}+2\alpha_{i-1,j}+r_{i-1,j-1}$. This cluster will no-longer appear in the superclustering. 
\end{itemize}
The augmentation parameters $\alpha_{i,j}$, as well as the precise number of external phases $T$, and internal superclustering steps $t$ are all set in a very delicate manner. We note that this additional \emph{augmentation} step is not applied in the Elkin-Neiman's algorithm \cite{ElkinN16b}, and we find it to be quite useful in our stretch analysis. The output clustering $\mathcal{C}_i$ of the $i^{th}$ phase consists of a cluster $C_\ell=V(SC_\ell)$ for every supercluster $SC_\ell \in \mathcal{SC}_{i-1,t}$ where $\mathcal{SC}_{i-1,t}$ is the last superclustering of that phase. 

At the end of all these $T$ phases, we are left with a clustering $\mathcal{C}_T$ with $N_T=O(n^{1-1/k^{\epsilon}})$ clusters. The augmentation parameters $\alpha_{i,j}$ are set in a way that guarantees that the radius this clustering is bounded $r_T=O_{\epsilon}(k^{1-\epsilon})$. Before starting the final phase, 
every unclustered vertex at distance at most $r_{T}+d$ from some cluster center, adds to the spanner its shortest path to its closest center.

The \textbf{final clustering-graph phase} computes a cluster graph $\widehat{G}$ where each cluster in $\mathcal{C}_T$ corresponds to a node in that graph. Two clusters $C,C' \in \mathcal{C}_T$ are connected in $\widehat{G}$ if their center-distance is at most $2r_{T}+2d$. The algorithm then computes a $(2k^{\epsilon}-1)$ spanner $\widehat{H} \subseteq \widehat{G}$ containing at most $O(N^{1+1/k^{\epsilon}})=O(n)$ edges. Finally, these spanner edges are translated into edges in $G$, by adding to the final spanner $H$ the shortest path between the centers $r(C)$ and $r(C')$ for each edge $(C,C') \in \widehat{H}$. It is easy to see that this step adds $O_k(n)$ edges to the spanner. 

The analysis is based on providing distinct stretch guarantees depending on the precise step\footnote{The $(i,j)$ step is the $j$'th step of the $i$'th phase.} $(i,j)$ in which the vertex stopped being clustered. 
Formally, a vertex is said to be $(i-1,j)$-unclustered if it belongs to the superclusters of $\mathcal{SC}_{i-1,j-1}$ but does not belong to the superclusters of $\mathcal{SC}_{i-1,j}$. The analysis shows that the later a vertex $u$ stopped being cluster the stronger is its stretch guarantee in the following sense. For every $(i-1,j)$-unclustered vertex $u$, the edges added to the spanner in phase $i$ guarantee that $\dist_H(u,v) \leq k^{\epsilon}\cdot \dist_G(u,v)$ for every vertex $v$ at distance of $\alpha_{i-1,j}$ from $u$, where $\alpha_{i-1,j}$ grows with both $i$ and $j$. The key point is that this stretch bound holds even if $v$ is $(i-1,j-1)$ unclustered. We complete the argument by considering any $u$-$v$ shortest path $P$ (of any length), and dividing it into consecutive disjoint segments $P[x,y]$ of possibly varying lengths. The length of each segment depends on the step $(i,j)$ in which certain vertices along the path $P$ stopped being clustered. 
We then show that for each segment (except perhaps the last one), the spanner provides a multiplicative stretch of $k^{\epsilon}$ between the endpoints of this segment. A detailed algorithm description appears in Sec. \ref{sec:second-regime-spanner}.

\subsubsection{New Hopsets} Our hopsets constructions bare similarities to the spanner algorithms, but include several modifications. First, our hopset is defined for weighted graphs whereas in the $(\alpha,\beta)$ spanners the graph is required to be unweighted. The key difference will be in the way that we handle the sparse vertices. In the spanners above, we applied a truncated variant of the Baswana-Sen algorithm. Here, we will use the classic $(2k-1,2)$ hopsets that followed by the distance oracle of Thorup and Zwick \cite{ThorupZ05}. To explain the ideas in their cleanest and most simplified form, in the below high-level description we restrict attention to unweighted graph and mainly explain the first non-trivial construction of $(\sqrt{k},\sqrt{k})$ hopsets. We start by over-viewing the construction of $(2k-1,2)$ hopsets, and highlight their key properties. 

\paragraph{A Short Exposition of $(2k-1,2)$ Hopsets.} The construction of the distance oracle by Thorup and Zwick \cite{ThorupZ05} is based on an hierarchical collection of \emph{centers}  $A_{k-1}\subset...\subset A_{0}=V$. Each $A_{i}$ is obtained by sampling each $v \in A_{i-1}$ independently with probability of $n^{-1/k}$. The $i^{th}$ \emph{pivot} of every vertex $v$ denoted by $p_i(v)$ is the closest vertex in $A_{i}$ to $v$. For every $v$, its $i^{th}$ \emph{bunch} $B_i(v)$ contains all vertices in $A_i$ that are closer to $v$ than $p_{i+1}(v)$.
Thorup and Zwick showed that for every $v$, $|B(v)|=O(k \cdot n^{1/k})$ in expectation. As observed by Huang and Pettie \cite{HuangP19}, the collection of bunches translates into a $(2k-1,2)$ hopset $H$ as follows: for every $v$ and $u\in B(v)$, add to $H$ an edge $(u,v)$ of weight $\dist_G(u,v)$. 
Since each bunch $B(v)$ is of size $O(k n^{1/k})$, the output hopset has $O(k n^{1+1/k})$ edges.
The key property that will be used by our algorithms is as follows. Fix a pair of vertices $u,v$ and define $i^*$ as the minimal index such that $p_{i^*}(v)$ is in the bunch of $u$ or vice-versa. By construction, $i^* \leq k-1$. By the stretch argument of the distance oracle of Thorup and Zwick \cite{ThorupZ05}, one can show that 
the hopset $H$ contains a two-hop path between $u$ and $v$ that goes through the common $(i^*)$th pivot, and the path has length at most $(2i^*+1)\cdot \dist(u,v)$. 
We are next explain the construction of the $(\sqrt{k}, \sqrt{k})$ hopset. 

\paragraph{Warming Up: $(\alpha,\beta)$ Hopset with $\alpha,\beta=O(\sqrt{k})$.} 
As common in hopset constructions, we fix a distance range $[d,2d]$ and describe a construction that provides the desired stretch and hop-bound for all pairs $u,v$ at distance $d' \in [d,2d]$. The same procedure is then applied for each of the $\log{\Lambda}$ distance ranges, where $\Lambda$ is the aspect ratio of the graph. 

The hopset algorithm has two phases. First it computes the $(2k-1,2)$ hopsets based on the Thorup-Zwick distance oracle. In fact, for our purposes it will be sufficient to apply only the first $\sqrt{k}$ steps of the construction, i.e., computing the center sets $A_0=V, \ldots, A_{\sqrt{k}}$, and adding hops to the hopset $H'$ based on the bunches $B_i(v)$ for every $v$ and every $i \in \{1,\ldots,\sqrt{k}\}$. 

In the second phase, the algorithm computes a clustering $\mathcal{C}_0$ centered at the vertices of $A_{\sqrt{k}}$ as follows. Each vertex $v \in G$ at distance $2d$ from $A_{\sqrt{k}}$ joins the cluster of its closest center in $A_{\sqrt{k}}$. This defines a cluster collection $\mathcal{C}_0$ of $|A_{\sqrt{k}}|=n^{1-1/\sqrt{k}}$ (in expectation) of vertex-disjoint clusters of radius at most $2d$. In the hopset, we add a hop $(v,c_v)$ between each clustered vertex $v$ to its cluster center $c_v \in A_{\sqrt{k}}$. 

Now, the algorithm computes the cluster-graph $\widehat{G}$ in which each cluster in $\mathcal{C}_0$ corresponds to a node, and two clusters $C,C' \in \mathcal{C}_0$ are adjacent iff the distance between their centers is at most $6d$. We note that the graph $\widehat{G}$ will be unweighted even when the graph $G$ is weighted, as each edge in $\widehat{G}$ corresponds to a path of length at most $6d$ in $G$. 
Letting $\widehat{H}$ be a $(2\sqrt{k}-1)$ multiplicative spanner of $\widehat{G}$, for every edge $(C,C')\in \widehat{H}$, the algorithm adds the hop $(r(C), r(C'))$ to the hopset, where $r(C),r(C')$ are the centers of $C,C'$ respectively. 

\textbf{Why it works?} First we bound the size of the hopset by $O_k(n^{1+1/k})$. The first phase adds a TZ-hopset with $O(k \cdot n^{1+1/k})$ edges. In the second phase, we add one hop for every edge $(C,C')$ in the spanner $\widehat{H}$. Since this spanner has $O(|A_{\sqrt{k}}|^{1+1/\sqrt{k}})=O(n)$ edges, it adds $O(n)$ hops overall. In addition, each clustered vertex is connected with a hop to its cluster center, and as the clusters are vertex disjoint, it adds $O(n)$ hops. 

Next, we consider the stretch and hop-bound argument for a fixed pair $u$ and $v$ at distance $d' \in [d,2d]$ in $G$. Let $P$ be their shortest path. A vertex $w$ is called \emph{clustered} if it belongs to the clusters of $\mathcal{C}_0$. By definition, a vertex $w$ is clustered iff $\dist_G(w,A_{\sqrt{k}})\leq 2d$. First, consider the case where at most one vertex on $P$ is clustered. The argument goes by partitioning $P$ into $O(\sqrt{k})$ disjoint consecutive segments of length $2d/\sqrt{k}$. For each such segment $P[x,y]$ we show that in the TZ-hopset $H'$ there is a  two-hop path between $x$ and $y$ of length at most $O(d)$. 
To see this, consider an unclustered vertex $x$ and let $i$ be the minimum index satisfying that 
$p_i(x) \in B_i(y)$. By the properties of the TZ-hopsets, $\dist(x,p_i(x))\leq i \cdot \dist(x,y)$. Since 
$\dist_G(x,A_{\sqrt{k}})> 2d$, and $\dist(x,y)=2d/\sqrt{k}$, we get that $i \leq \sqrt{k}$. 
We therefore have $O(\sqrt{k})$ segments on $P$, for each the TZ-hopset provides a two-hop path of length $O(d)$. In total, we get a $u$-$v$ path of $2\sqrt{k}$ hops, and of total length $O(\sqrt{k} \cdot d)$ as required. This establishes the stretch and hop-bound argument for a path containing at most one clustered vertex.

The other case follows by the second phase of the algorithm. We consider the far-most clustered vertex pair $u'$ and $v'$ on the $u$-$v$ shortest path. The stretch and hop-bound argument for the subpaths $P[u,u']$ and $P[v',v]$ follows by the argument above, and thus it remains to consider the path $P[u',v']$.
Let $C_{u'}$ and $C_{v'}$ be the cluster of $u'$ and $v'$ respectively. Since  $\dist_G(u',v')\leq 2d$, we get that the clusters $C_{u'}$ and $C_{v'}$ are adjacent in $\widehat{G}$, and therefore the hopset contains a path of at most $(2\sqrt{k}-1)$ hops between between the centers of $C_{u'}$  and $C_{v'}$.
As each such hop has weight of $O(d)$, overall the hopset contains a path with $O(\sqrt{k})$ hops connecting $u'$ and $v'$, of total length $O(\sqrt{k}\cdot d)$ as required. This completes the high-level idea of the construction, see Sec. \ref{sec:first-regime-hopset} for more details.

\paragraph{Three Stage Approach for $(k^\epsilon,k^{1-\epsilon})$ Hopsets.} The computation of the $(k^\epsilon,k^{1-\epsilon})$ hopsets for every $\epsilon \in [1/2,1)$ is very similar to the high level description mentioned above. The complementary range of $\epsilon \in (0,1/2)$ is considerably more involved. It also has a three stage structure in a very similar manner to the $(k^{\epsilon}, k)$ spanners. The first stage applies a truncated TZ hopset construction restricting to the first $k^{\epsilon}$ levels of clustering. Letting $A_{k^{\epsilon}}$ be the centers in the $k^{\epsilon}$-level of the Thorup-Zwick algorithm, the second phase computes an initial clustering $\mathcal{C}_0$ with centers of $A_{k^{\epsilon}}$. The vertices that do not belong to the clusters of $\mathcal{C}_0$ are called unclustered. For those vertices, the correctness will follow by the TZ-hopsets. The remaining clustered vertices are handled in two stages. A key stage of superclustering which rapidly reduces the number of clusters to $n^{1-1/k^{\epsilon}}$, and a final stage in which the number of clusters is small enough, to allow the computation of an $k^{\epsilon}$-spanner on that cluster graph.
A more detailed description appears in Sec. \ref{sec:second-regime-hopset}. 
\paragraph{Open Problems.}
The most important open problem left by this work concerns the existence of $(\alpha,\beta)$ spanners with $\widetilde{O}(n^{1+1/k})$ edges for $\alpha=O(1)$ and $\beta=O(k)$. This would provide a nearly optimal stretch, up to constants, for the entire range of distances. With our current constructions one can only get $(\alpha, O(k))$ spanners with $\alpha=2^{O(\sqrt{\log k})}$. Alternatively, for a multiplicative stretch of $\alpha=O(1)$, we currently get $\beta=k^{1+o(1)}$. Note that the lower-bound constructions of Abboud, Bodwin and Pettie \cite{abboud2018hierarchy} are only tight for constant values of $k$, and hence it might still be possible to even obtain $(1+\epsilon, O(k))$ spanners with $O_{k,\epsilon}(n^{1+1/k})$ edges. 
Another interesting open problem concerns the tightness of our hopsets constructions. We present a new family of $(\alpha=k^{\epsilon}, \beta=k^{1-\epsilon})$ hopsets with $\widetilde{O}(n^{1+1/k})$ edges for any constant $\epsilon \in (0,1)$. The most critical question is whether any $(\alpha,\beta)$ hopset with $\widetilde{O}(n^{1+1/k})$ edges must satisfy that $\alpha\cdot \beta=\Omega(k)$. 

\definecolor{uququq}{rgb}{0.25,0.25,0.25}
\definecolor{xdxdff}{rgb}{0.49,0.49,1}

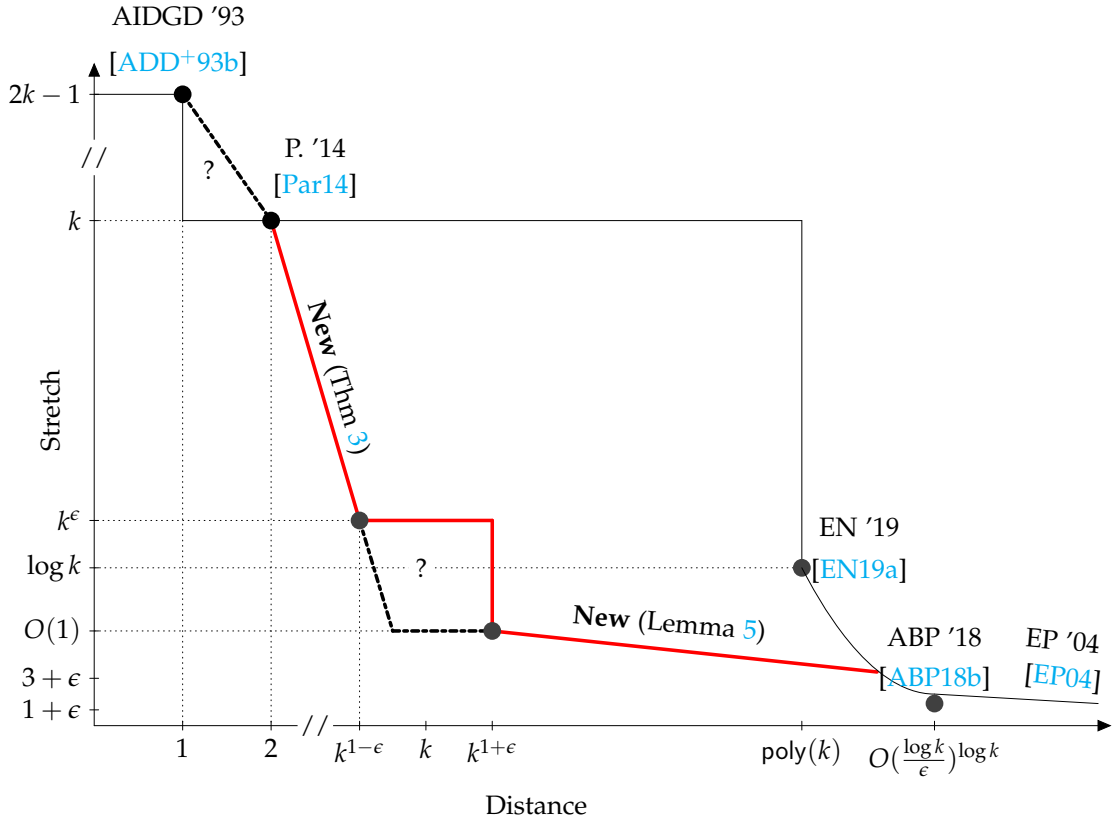
\begin{figure}
		\begin{center}
		\resizebox{0.9\linewidth}{!}{
		\begin{tikzpicture}[line cap=round,line join=round,>=triangle 45,x=7cm,y=5cm]
		\draw[->,color=black] (0,0) -- (2.3,0) ;
		\draw[color=black] (-0.1,1) node[rotate=90] {\large Stretch};
		\draw[color=black] (1,-0.25) node[rotate=0] {\large Distance};
		\draw[color=black] (2pt,1.6) -- (-2pt,1.6) node[left] {\large $k$};
		\draw[color=black] (0.75,2pt) -- (0.75,-2pt) node[below] {\large $k$};
		\draw[color=black] (0.6,2pt) -- (0.6,-2pt) node[below] {\large $k^{1-\epsilon}$};
		\draw[color=black] (0.9,2pt) -- (0.9,-2pt) node[below] {\large $k^{1+\epsilon}$};
		\draw[color=black] (1.6,2pt) -- (1.6,-2pt) node[below] {\large $\poly(k)$};			
		\draw[color=black] (1.9,2pt) -- (1.9,-2pt) node[below] {\large $O(\frac{\log{k}}{\epsilon})^{\log{k}}$};
		\draw[color=black] (0.4,2pt) -- (0.4,-2pt) node[below] {\large $2$};
		\fill [color=white] (0.5,0) circle (8pt);
		\draw[color=black] (0.5,0) node[rotate=0] {//};
		\draw[color=black]  (0.2,2pt) -- (0.2,-2pt) node[below] {\large $1$};
		\draw[color=black] (2pt,0.65) -- (-2pt,0.65) node[left] {\large $k^{\epsilon}$};
		\draw[color=black] (2pt,0.5) -- (-2pt,0.5) node[left] {\large $\log{k}$};
		\draw[color=black] (2pt,0.3) -- (-2pt,0.3) node[left] {\large $O(1)$};
		\draw[color=black] (2pt,0.15) -- (-2pt,0.15) node[left] {\large $3+\epsilon$};
		\draw[color=black] (2pt,0.05) -- (-2pt,0.05) node[left] {\large $1+\epsilon$};
		\draw[color=black] (2pt,2) -- (-2pt,2) node[left] {\large $2k-1$};
		\draw[->,color=black] (0,0) -- (0,2.1);
		\fill [color=white] (0,1.8) circle (8pt);
		\draw[color=black] (0,1.8) node[rotate=0] {//};
		
		\draw[color=black] (0.55,1.1) node[rotate=-73] {\large \textbf{New} (Thm \ref{thm:secondspanner})};
		\draw[dotted] (0,0.5)-- (1.6,0.5);
		\fill [color=white] (0.74,0.5) circle (6pt);
		\draw[color=black] (0.74,0.5) node[rotate=0] {\large ?};
		
		\draw (0,2)-- (0.2,2);
				\draw[ultra thick,dash dot] (0.2,2)-- (0.4,1.6);
		\draw[color=black] (0.26,1.76) node[rotate=0] {\large ?};
		\fill [color=black] (0.2,2) circle (4pt);
		\draw[color=black] (0.19,2.25) node[rotate=0] {\large AIDGD '93};
		\draw[color=black] (0.19,2.1) node[rotate=0] {\large \cite{AlthoferDDJS93}};
		\draw (0.2,1.99)-- (0.2,1.6);
		\draw[dotted] (0,1.6)-- (0.2,1.6);
		\draw[dotted] (0.2,1.6)-- (0.2,0);
				\draw[ultra thick, red] (0.4,1.6)-- (0.6,0.65);
		\draw (0.2,1.6)-- (1.6,1.6);
		\draw (1.6,1.6)-- (1.6,0.5);
		\draw[dotted] (0.4,1.6)-- (0.4,0);
		\fill [color=black] (0.4,1.6) circle (4pt);
		\draw[color=black] (0.5,1.83) node[rotate=0] {\large P. '14};
		\draw[color=black] (0.5,1.71) node[rotate=0] {\large \cite{Parter14}};
		
		\draw (1.9,0.1)-- (2.27,0.07);
		\draw[color=black] (2.19,0.27) node[rotate=-5] {\large EP '04};
		\draw[color=black] (2.19,0.15) node[rotate=-6] {\large \cite{ElkinP04}};
		\draw[color=black] (1.3,0.32) node[rotate=-5] {\large \textbf{New} (Lemma \ref{lem:3plusepsspanner})};
		\draw[ultra thick, red] (0.9,0.3)-- (1.77,0.17);
		\draw[ultra thick,dash dot] (0.6,0.65)-- (0.675,0.3);
		\draw[dotted] (0.6,0.65)-- (0.6,0);
		\draw[dotted] (0.6,0.65)-- (0,0.65);
		\draw[ultra thick, red] (0.6,0.65)-- (0.9,0.65);
		\draw[dotted] (0,0.3)-- (0.675,0.3);
		\draw[ultra thick,dash dot] (0.675,0.3)-- (0.9,0.3);
		\draw[ultra thick, red] (0.9,0.65)-- (0.9,0.3);
		\fill [color=uququq] (0.6,0.65) circle (4pt);
		\fill [color=uququq] (0.9,0.3) circle (4pt);
		\fill [color=uququq] (1.6,0.5) circle (4pt);
		\draw[color=black] (1.73,0.63) node[rotate=0] {\large EN '19};
		\draw[color=black] (1.73,0.49) node[rotate=0] {\large \cite{ElkinN19}};
	 \draw (1.6,0.5) parabola[bend at end] (1.9,0.1) node[below right]{};
	 	\fill [color=uququq] (1.9,0.07) circle (4pt);
	 \draw[color=black] (1.9,0.28) node[rotate=0] {\large ABP '18};
	 \draw[color=black] (1.9,0.15) node[rotate=0] {\large\cite{AbboudBP18}};
	 	
		\end{tikzpicture}}
	\end{center}
		\caption{\label{fig:plot-spanner} \normalsize{The landscape of the stretch function $f(d)/d$ vs. $d$ when fixing the size of the $f(d)$-spanner to $\widetilde{O}(n^{1+1/k})$ edges. The girth conjecture implies $f(d)\geq 2k-1$ for $d=1$ and $f(d)\geq k$ for $d=2$. Prior to our work, constant stretch of $(1+\epsilon)$ was obtained only for $d \geq O(\log k/\epsilon)^{\log k}$ while all closer vertex pairs $d \geq 3$ still suffered from a multiplicative stretch of $k$. Our results leave a gray open area (marked with $?$ in the plot) for the range of distances $d \in [k^{1-\epsilon}, k^{1+\epsilon}]$ for any small constant $\epsilon$. Specifically, currently for vertices at distance $k$ we have a stretch of $2^{O(\sqrt{\log k})}$}. }
\end{figure}

\begin{figure}[h!]
		\begin{center}
		\resizebox{0.9\linewidth}{!}{
		\begin{tikzpicture}[line cap=round,line join=round,>=triangle 45,x=7cm,y=5cm]
		\draw[->,color=black] (0,0) -- (2.3,0) ;
		\draw[color=black] (-0.1,1) node[rotate=0] {\large $\alpha$};
		\draw[color=black] (1,-0.25) node[rotate=0] {\large $\beta$};
		\draw[color=black] (2pt,1.6) -- (-2pt,1.6) node[left] {\large $k^{1-\epsilon}$};
		\draw[color=black] (0.6,2pt) -- (0.6,-2pt) node[below] {\large $k^{1-\epsilon}$};
		\draw[color=black] (0.4,2pt) -- (0.4,-2pt) node[below] {\large $k^{\epsilon}$};
		\draw[color=black] (0.9,2pt) -- (0.9,-2pt) node[below] {\large $k^{1+\epsilon}$};
		\draw[color=black] (1.6,2pt) -- (1.6,-2pt) node[below] {\large $\poly(k)$};			
		\draw[color=black] (1.85,2pt) -- (1.85,-2pt) node[below] {\large $O(\frac{\log{k}}{\epsilon})^{\log{k}}$};
		\draw[color=black] (2.2,2pt) -- (2.2,-2pt) node[below] {\large $\widetilde{O}(\sqrt{n})$};
		\draw[color=black] (2.2,0.22) node[rotate=0] {\large KS '97,SS' 99};
		\draw[color=black] (2.2,0.15) node[rotate=0] {\large \cite{klein1997randomized},\cite{shi1999time}};
		\fill [color=white] (0.5,0) circle (8pt);
		\draw[color=black] (0.5,0) node[rotate=0] {//};
		\fill [color=white] (2.04,0) circle (8pt);
		\draw[color=black] (2.04,0) node[rotate=0] {//};
		\draw[color=black]  (0.2,2pt) -- (0.2,-2pt) node[below] {\large $2$};
		\draw[color=black] (2pt,0.65) -- (-2pt,0.65) node[left] {\large $k^{\epsilon}$};
		\draw[color=black] (2pt,0.5) -- (-2pt,0.5) node[left] {\large $\log{k}$};
		\draw[color=black] (2pt,0.3) -- (-2pt,0.3) node[left] {\large $O(1)$};
		\draw[color=black] (2pt,0.2) -- (-2pt,0.2) node[left] {\large $3+\epsilon$};
		\draw[color=black] (2pt,0.1) -- (-2pt,0.1) node[left] {\large $1+\epsilon$};
		\draw[color=black] (2pt,0.05) -- (-2pt,0.05) node[left] {\large $1$};
		\draw[color=black] (2pt,2) -- (-2pt,2) node[left] {\large $2k-1$};
		\draw[->,color=black] (0,0) -- (0,2.1);
		\fill [color=white] (0,1.8) circle (8pt);
		\draw[color=black] (0,1.8) node[rotate=0] {//};
		
		\draw[color=black] (0.55,1.1) node[rotate=-73] {\large \textbf{New} (Thm. \ref{thm:secondspanner})};
		\draw[dotted] (0,0.5)-- (1.6,0.5);
		\fill [color=white] (0.74,0.5) circle (6pt);
		\draw[color=black] (0.74,0.5) node[rotate=0] {\large ?};
		\draw[dotted] (0,0.05)-- (2.2,0.05);
		\draw[dotted] (2.2,0)-- (2.2,0.05);
		\fill [color=black] (2.2,0.05) circle (4pt);

		\draw (0,2)-- (0.2,2);
		\fill [color=black] (0.2,2) circle (4pt);
		\draw[ultra thick,dash dot] (0.2,2)-- (0.4,1.6);
		\draw[color=black] (0.26,1.76) node[rotate=0] {\large ?};
		\draw[color=black] (0.2,2.21) node[rotate=0] {\large TZ '05};
		\draw[color=black] (0.2,2.08) node[rotate=0] {\large \cite{ThorupZ05}};
		\draw (0.2,1.99)-- (0.2,1.6);
		\draw[dotted] (0,1.6)-- (0.2,1.6);
		\draw[dotted] (0.2,1.6)-- (0.2,0);
		\draw[dotted] (0.2,1.6)-- (0.4,1.6);
		\draw[dotted] (0.4,1.6)-- (0.4,0);

		\draw (1.9,0.1)-- (2.2,0.07);
		\draw[color=black] (1.75,0.45) node[rotate=-50] {\large EN '16};
		\draw[color=black] (1.71,0.38) node[rotate=-50] {\large \cite{ElkinN16b}};
		\draw[color=black] (1.87,0.26) node[rotate=0] {\large HP '19};
		\draw[color=black] (1.87,0.18) node[rotate=0] {\large \cite{HuangP19}};
		\draw[color=black] (1.3,0.32) node[rotate=-5] {\large\textbf{New} (Lemma \ref{lem:hopset-best})};
		
		\draw[ultra thick, red] (0.9,0.3)-- (1.77,0.17);
		\draw[ultra thick, red] (0.4,1.59)-- (0.6,0.65);
		\draw[ultra thick,dash dot] (0.6,0.65)-- (0.675,0.3);
		\draw[dotted] (0.6,0.65)-- (0.6,0);
		\draw[dotted] (0.6,0.65)-- (0,0.65);
		\draw[ultra thick, red] (0.6,0.65)-- (0.9,0.65);
		\draw[dotted] (0,0.3)-- (0.675,0.3);
		\draw[ultra thick,dash dot] (0.675,0.3)-- (0.9,0.3);
		\draw[ultra thick, red] (0.9,0.65)-- (0.9,0.3);
		\fill [color=black] (0.4,1.6) circle (4pt);
		\fill [color=uququq] (0.6,0.65) circle (4pt);
		\fill [color=uququq] (0.9,0.3) circle (4pt);

	 \draw (1.6,0.5) parabola[bend at end] (1.9,0.1) node[below right]{};
	 \fill [color=uququq] (1.85,0.1) circle (4pt);
		\end{tikzpicture}}
	\end{center}
		\caption{The tradeoff between $\alpha$ (stretch) and the $\beta$ (hop-bound) in existing hopset constructions with $O(n^{1+1/k})$ edges. The question marks indicates the current gaps in these regimes.}  \label{fig:hopplot}
	\end{figure}
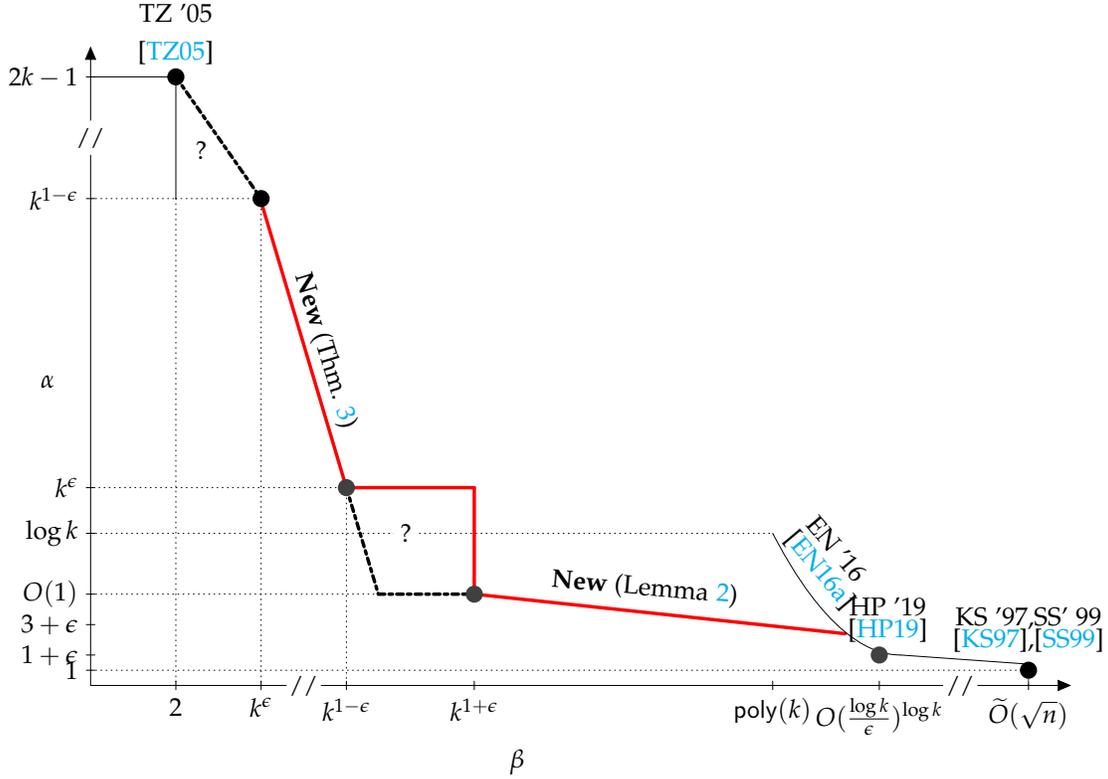

\subsection{Preliminaries}
\paragraph{Graph Notations and Definitions.} We consider an undirected $n$-vertex graph $G=(V,E)$, where $V$ is the set of vertices and $E$ is the edge-set. let $N_G(u)$ be the neighbor of $u$ in $G$. When $G$ is clear from the context, we may simply write $N(u)$. Unless specified otherwise we assume $G$ to be unweighted. For $u,v\in V$ we denote by $\dist_{G}(u,v)$ the distance from $u$ to $v$ in $G$. Similarly, for any subgraph $H\subseteq{G}$ we denote by $\dist_{H}(u,v)$ the distance between the vertices in $H$. For any vertex $u\in V$ and integer $d$, we denote by $\Ball_{G}(u,d)$ the set of vertices at distance at most $d$ from $u$ in $G$, that is $\Ball_{G}(u,d)=\{v\in V ~|~\dist_{G}(u,v)\le d\}$. When the context is clear we might simply write $\Ball(u,d)$. By $\partial \Ball_{G}(u,d)$ we denote the set of vertices at distance \textit{exactly} $d$ from $u$, that is $\partial \Ball_{G}(u,d)=\{v\in V ~|~\dist_{G}(u,v)=d\}$. 
For a weighted graph $G=(V,E,w)$, the \emph{aspect ratio} denoted by $\Lambda$ is the ratio between the largest and smallest vertex-pair distances in $G$. Unless stated otherwise, in our constructions, shortest path ties are broken in a consistent manner.

\paragraph{$(\alpha,\beta)$ Hopsets.} Hopsets are fundamental graph structures introduced by Cohen \cite{cohen2000polylog}. Let $G = (V, E, w)$ be an undirected \emph{weighted} graph and $H \subset {V\choose 2}$ be a set of edges called the \emph{hopset}.  
In the graph $G' = (V, E\cup H,w')$ the weight function $w'$ is defined by letting $w'(e)=w(e)$ for every $e \in E$ and $w'(e=(x,y)))=\dist_G(x,y)$. 
Define the \emph{$\beta$}-limited distance in $G'$, denoted $\dist^{(\beta)}_{G'}(u, v)$, 
to be the length of the shortest path from $u$ to $v$ that uses at most $\beta$ edges in $G'$.
We call $H$ a \emph{$(\beta,\epsilon)$-hopset}, where $\beta \ge 1, \epsilon>0$, 
if, for any $u, v \in V$, we have
$\dist_{G'}^{(\beta)}(u, v) \le (1 + \epsilon) \dist_G(u, v)$.

\paragraph{Clusters and Superclusters.} A \emph{cluster} $C \subseteq V$ is a subset vertices with a small weak diameter in $G$. Every cluster $C$ has a special vertex $r(C)$ that is denoted as the \emph{cluster center}. A \emph{clustering} $\mathcal{C}=\{C_1,\ldots, C_\ell\}$ is a collection of disjoint clusters, where the vertices of the clustering are denoted by $V(\mathcal{C})=\bigcup_i C_i$. Note that $V(\mathcal{C})$ is not necessarily $V(G)$.  

In our algorithms we measure the distance between clusters $C,C'$ by the distance between their centers $r(C),r(C')$ respectively. Formally, define $\cdist_G(C,C')=\dist_G(r(C),r(C'))$. In the same manner, for a vertex $v$ and a cluster $C$, define $\cdist_G(v,C)=\dist_G(v,r(C))$. For a collection of clusters $\mathcal{C}=\{C_1,\ldots, C_\ell\}$ and a cluster $C'$, let $\cdist_G(C', \mathcal{C})=\min_{C_j \in \mathcal{C}}\cdist_G(C',C_j)$.
In the same manner, for a vertex $v$ and a cluster collection $\mathcal{C}$, let $\cdist_G(v, \mathcal{C})=\min_{C_j \in \mathcal{C}}\cdist_G(v,C_j)$.  
Throughout, we break shortest-path ties based on IDs. For instance, the closest cluster in $\mathcal{C}$ to a given vertex $v$ is the minimum-ID cluster $C'$ in $\mathcal{C}$ that satisfies that $\cdist(v, \mathcal{C})=\dist_{G}(v,r(C'))$. The radius of a cluster $C$ is defined by $rad(C):=\max_{u\in V(C)} \dist_{G}(u,r(C))$.

A \emph{supercluster} $SC_i=\{C_{1}\cdots C_{j}\}$ is a set of disjoint clusters, with one special vertex, namely, the center of superclusters, denoted by  $r(SC_i)$.
Let $V(SC_i)=\bigcup_{C_j \in SC_i}V(C_j)$ be the vertices of the supercluster. 
A superclustering $\mathcal{SC}=\{SC_1,\ldots, SC_\ell\}$ is a collection of vertex disjoint superclusters. Let $V(\mathcal{SC})=\bigcup_{SC_j \in \mathcal{SC}}V(SC_j)$ denote the vertices of the superclustering $\mathcal{SC}$. Similarly to clusters, for a pair of superclusters $SC_i,SC_j$ define 
$\cdist_G(SC_i,SC_j)=\dist_{G}(r(SC_{i}),r(SC_{j}))$. For a cluster $C$ and a supercluster $SC$, define $\cdist(C,SC)=\dist_{G}(r(C),r(SC))$. For a cluster $C$ and a superclustering $\mathcal{SC}$, define $\cdist(C, \mathcal{SC})=\min_{SC'\in \mathcal{SC}}\cdist(C,SC')$. Similarly, for a vertex $v\in V$ and a supercluster $SC$ let $\cdist(v,SC)=\dist_{G}(v,r(SC))$, and for a superclustering $\mathcal{SC}$, define $\cdist(v,\mathcal{SC})=\min_{SC'\in \mathcal{SC}}\cdist(v,SC')$.
The \emph{radius} of a supercluster $SC_i$ is defined by $rad(SC_i):=\max_{u\in V(SC_i)}\dist_{G}(u,r(SC_i))$.


\subsection{Algorithmic Tools}
\paragraph{Multiplicative Spanners of Baswana and Sen \cite{BaswanaS07}.} Our algorithms are based on a truncated variant of the Baswana-Sen algorithm, namely Procedure $\Cluster$. The procedure gets as an input a graph $G$, a stretch parameter $k$ and integer $t$ that determines the number of clustering steps. The algorithm begins with singleton clusters $\mathcal{C}_{0}=\{\{v\} ~|~v\in V\}$. Then, at every step $i \in \{1,\ldots,t\}$, a clustering $\mathcal{C}_i$ is defined based on the given clustering $\mathcal{C}_{i-1}$. Every cluster in $\mathcal{C}_{i-1}$ is sampled with probability $n^{-1/k}$ to $\mathcal{C}_{i}$. Vertices that are not adjacent\footnote{We say that a vertex $v$ is adjacent to a cluster $C$ if $C$ contains at least one neighbor of $v$.} to sampled clusters are called \emph{unclustered}, and they do not appear in the following clusters. For each unclustered vertex the procedure adds to the spanner $H$ one edge to each of its adjacent clusters in $\mathcal{C}_{i-1}$. Any other vertex joins its closest sampled cluster. Finally, the algorithm also adds to the output subgraph $H$ the spanning trees of each cluster $C \in \mathcal{C}_i$ (rooted at the cluster center) for every $i \in \{1,\ldots, t\}$. 

\begin{fact}{[\textbf{Theorems 4.1,4,2} and \textbf{Lemma 4.1} in \cite{BaswanaS07}]}
\label{claim:cluster}
(1) For every $1\le i\le t$, $\mathbb{E}[|\mathcal{C}_{i}|]=n^{1-\frac{i}{k}}$, (2) the radius of each cluster $C \in \mathcal{C}_{i}$ is at most $i$, (3) the total number of edges, in expectation, added to $H$ is $O(t\cdot n^{1+1/k})$, and (4) for every edge $(u,v) \in G$ satisfying that at least one of the endpoints is not clustered in 
$\mathcal{C}_{i}$, it holds that $\dist_{H}(u,v)\le 2i-1$.
\end{fact}

\paragraph{Distance Oracles and Hopsets of Thorup and Zwick \cite{ThorupZ05}.}
\label{par:TZ} 
Our hopsets constructions are based on the $(2k-1,2)$ hopsets that are based on the distance oracles of Thorup and Zwick \cite{ThorupZ05}.  The construction of the hopset is based on defining an hierarchical collection of \emph{centers}  $A_{k-1}\subset...\subset A_{0}=V$, where each $A_{i}$ is obtained by sampling each $v \in A_{i-1}$ independently with probability of $n^{-1/k}$ and let $A_{k}=\phi$. The $i^{th}$ \emph{pivot} of every vertex $v$ denoted by $p_i(v)$ is the closest vertex in $A_{i}$ to $v$. For every $v$, its $i^{th}$ \emph{bunch} $B_i(v)$ contains all vertices in $A_i$ that are strictly closer to $v$ than $p_{i+1}(v)$. That is, $B_i(v)=\{u \in A_i - A_{i+1} ~\mid~ \dist_G(v,u)< \dist(v,p_{i+1}(v))\}$, and let $B(v)=\bigcup_{i=1}^{k-1} B_i(v)$. Note that we add a hop from each vertex to all vertices in $A_{k-1}$. Thorup and Zwick showed that for every $v$, $|B(v)|=O(k\cdot n^{1/k})$ in expectation.
The collection of bunches translates into $(2k-1,2)$ hopset $H$ as follows: for every $v$ and $u\in B(v)$, add to $H$ an edge $(u,v)$ of weight $\dist_G(u,v)$. In our applications, we sometimes use a truncated version of the Thorup and Zwick construction for a given input parameter $t \leq k$. Algorithm $\TZ(G,k,t)$ gets as input a graph $G$, stretch parameter $k$ and an integer $t$. The output of the algorithm is the TZ-hopset $H_{TZ}$ along with the subset $A_t$, that is, level-$t$ centers that contains $n^{1-t/k}$ vertices in expectation.

\begin{fact}
\label{fact:TZ}
Fix a vertex pair $u$ and $v$. Let $i_u$ the minimal index $i$ such that $p_i(u) \in B_i(v)$, and similarly let $i_v$ be the minimal index $i'$ such that $p_{i'}(v) \in B_{i'}(u)$. Let $i^*=\min\{i_u,i_v\}$. 
(i) For every $j \leq i^*$, it holds that $\dist_G(u,p_j(u))\leq j \cdot \dist_G(u,v)$, and
(ii) the hopset $H_{TZ}$ satisfies that
$\dist_G(u,v) \leq \dist^{(2)}_{H_{TZ} \cup G}(u,v)\leq  (2i^*+1) \dist_G(u,v)$. 
Furthermore (iii), it holds that $|B_{i}(v)|\le n^{1/k}$ in expectation.
\end{fact}
\begin{proof}
	(i) By induction on $j$. For $j=0$ this is clear since $p_{j}(u)=u$. Assume the claim is true for $j-1$, thus $\dist_{G}(v,p_{j-1}(v))\le (j-1)\cdot \dist_{G}(u,v)$ and $\dist_{G}(u,p_{j-1}(u))\le (j-1)\cdot \dist_{G}(u,v)$. For $j\leq i^*$, since $p_{j-1}(v) \notin B(u)$, it follows that 
$$\dist_{G}(u,p_j(u))\le \dist_{G}(u,p_{j-1}(v))\le \dist_{G}(u,v)+\dist_{G}(v,p_{j-1}(v))\le j\cdot \dist_{G}(u,v).$$
(ii) W.l.o.g., let $i^*=i_u$. Then, $\dist_{H_{TZ}}(u,v) \leq \dist_G(u, p_{i^*}(u))+\dist_G(p_{i^*}(u),v)$.
Then by (i) it follows that $\dist_{H_{TZ}}(u,v)\le i^*\cdot \dist_G(u,v)+(i^*+1)\dist_G(u,v)\leq (2\cdot i^{*}+1)\cdot \dist_{G}(u,v)$.
\end{proof}
Throughout in all our hopset constructions, whenever an edge $(u,v) \in V \times V$ is added to the hopset it is given a weight of $\dist_G(u,v)$. To avoid confusion, the edges not in $G$ are referred to as \emph{hops}.

\paragraph{Roadmap.} In Sec. \ref{sec:spanner-close}, we present the first spanner construction that provides a stretch of $7k/d$ for all vertices at distance at most $\sqrt{k}$. Then in Sec. \ref{sec:second-regime-spanner}, we present the key construction of $O(k^{\epsilon}, O_{\epsilon}(k))$ spanners. Sec. \ref{sec:3epsspanner} presents a simplified construction for $(3+\epsilon,\beta)$ spanners. 
Sec. \ref{sec:hopsets-new} describes the new hopset constructions. 
Finally, in Appendix \ref{sec:eff-con}, we describe the implementation details and applications, and in Appendix \ref{sec:best-con}, we show improved constructions of $(3+\epsilon,\beta)$ spanners and hopsets.

\def\APPENDTABLE{
\begin{table}
\centering
\small
\begin{tabular}{ |l|l|l|l| }
\hline
Reference &$\alpha$ & $\beta$  & Size  \\ 
\hline
\cite{klein1997randomized,shi1999time} & $1$ & $O(\sqrt{n}\log n)$ & $O(n)$ \\
\hline
\cite{cohen2000polylog} & $(1+\epsilon)$ & $k'=(\log n)^{O(\log k)}$ & $n^{1+1/k}\cdot k'$  \\
\hline
\cite{ThorupZ05} & $2k-1$ & $2$ & $O_k(n^{1+1/k})$  \\
\hline
\cite{ElkinN16b,HuangP19} & $(1+\epsilon)$ & $O(\log k/\epsilon)^{O(\log k)}$ & $\widetilde{O}(n^{1+1/k})$  \\
\hlineB{3}
\multirow{3}{*}{{\bf New} }& $k^{\epsilon}$ & $k^{1-\epsilon}$ & $O_k(n^{1+1/k})$   \\
&$O(1)$ & $k^{1+o(1)}$ & $O_k(n^{1+1/k})$   \\
& $3+\epsilon$ & $O_{\epsilon}(k^{\log(3+9/\epsilon)})$ & $O_k(n^{1+1/k})$   \\
\hline
\end{tabular}
\caption{\label{tab:hops} Comparison between $(\alpha,\beta)$-hopsets}\label{table:hopset}
\end{table}
}
	
\section{Improved Spanners for Close Vertex Pairs}\label{sec:spanner-close}
This section is devoted to showing Theorem \ref{thm:sqrt}. We consider unweighted graphs, and describe the
construction of a spanner that provides a stretch of $7k/d$ for every pair of vertices $u$ and $v$ at distance at most $d$ in $G$, provided that $d \leq \sqrt{k}/2$. In the language of $(\alpha,\beta)$-spanner, this spanner can be viewed as an $(O(\sqrt{k}), k)$ spanner. 
For simplicity, we fix a distance value $d \in [1,\sqrt{k}/2]$ and explain how to provide a stretch of $7k/d$ using a subgraph of expected size $O(k/d \cdot n^{1+1/k})$. The same procedure will be repeated for every $d \leq \sqrt{k}/2$.
\paragraph{Description of Algorithm $\ImprovedSpannerI$.} The algorithm consists of two key steps. The \textbf{first step} applies a truncated version of Baswana-Sen algorithm, applying only the first $\left \lfloor k/d \right \rfloor$ clustering steps. This results in a clustering $\mathcal{C}$ of expected size $n^{1-\frac{\left \lfloor k/d \right \rfloor}{k}}=O(n^{1+1/k-1/d})$, as well as a subset of edges added to the spanner $H$ (i.e., that takes care of the vertices that are not in the clusters of $\mathcal{C}$). 
\indent In the \textbf{second step}, the algorithm computes a cluster-graph $\widehat{G}(\mathcal{C},\mathcal{E})$ as follows. The vertices of the cluster-graphs, denoted as \emph{super-nodes}, are the clusters of $\mathcal{C}$. Every two clusters $C_i,C_j \in \mathcal{C}$ are connected in $\widehat{G}$ iff the distance between their centers in $G$ is at most $d+2k/d$. That is, $\mathcal{E}=\{(C_i,C_j) ~\mid~ \cdist_G(C_i,C_j)\leq d+2\cdot k/d\}$. 
The algorithm then computes a $(2d-1)$ spanner $\widehat{H}$ on this cluster graph, by using any standard multiplicative spanner procedure (e.g., the greedy spanner). Finally, the edges of this spanner are translated into $G$-edges as follows: for every $(C_i,C_j) \in \widehat{H}$, add to $H$ the shortest path in $G$ between the centers of $C_i$ and $C_j$. This completes the description of the algorithm. See Alg. \ref{alg:imp-spanner} for a pseudocode.
 
\begin{algorithm}
\begin{algorithmic}[1]
	\caption{$\ImprovedSpannerI(G,k,d)$.}
	\label{alg:imp-spanner}
\State \textbf{Input:} A graph $G=(V,E)$ and integers $k,d\le \sqrt{k}/2$. 
\State \textbf{Output:} A subgraph $H\subseteq{G}$ of expected size $O(\frac{k}{d}\cdot n^{1+1/k})$ that obtains stretch of $O(k/d)$ for pairs at distance $d$.
\State Let $(\mathcal{C},H_{0}) \leftarrow \Cluster(G,k,\left \lfloor k/d \right \rfloor)$.
\State Let $\widehat{G}=(\mathcal{C},\mathcal{E}=\{(C_i,C_j) ~\mid~   \cdist_G(C_i,C_j)\leq d+2k/d\})$.
\State $\widehat{H}\gets \BasicSpanner(\widehat{G},2d-1)$.
\State For each edge $(C,C')\in E(\widehat{H})$, add to $H$ the shortest path between $r(C)$ and $r(C')$ in $G$.
\State \Return $H$
\end{algorithmic}
\end{algorithm}
We next analyze Algorithm $\ImprovedSpannerI$ and prove Thm. \ref{thm:sqrt}.
\begin{proof}
\textbf{Stretch:} Fix a pair of vertices $u$ and $v$ at distance $d$ in $G$ for $d \in [1,\sqrt{k}/2]$, and
let $P$ be their shortest path in $G$. We consider two cases. First, assume that no edge $e=(u',v')$ on $P$ has both its endpoints in the clusters of $\mathcal{C}$. Then, by Fact \ref{claim:cluster}(4), we have that $\dist_H(u',v')\leq 2(k/d)-1$ for every edge $(u',v')$ in $P$. This gives a $u$-$v$ path in $H$ of total length $d \cdot (2(k/d)-1)=2k-d$. 
Next consider the complementary case where $P$ contains at least one edge with its both endpoints clustered. Let $u',v'$ be the leftmost and rightmost clustered vertices on $P$. Let us define the following subpaths $P_{1}:=P[u,u']$, $P_{2}:=[u',v']$ and $P_{3}:=[v',v]$, such that $P=P_{1}\circ P_{2} \circ P_{3}$. 
For every vertex $z$, let $C_z$ be its closest cluster in $\mathcal{C}$ with respect to the distance to the centers. Since $\cdist_{G}(C_{u'},C_{v'})\le 2k/d+\dist_{G}(u,v)= 2k/d+d$, it holds that $(C_{u'},C_{v'}) \in E(\widehat{G})$, thus by the properties of the $(2d-1)$-spanner $\widehat{H}$, it holds that $\dist_{\widehat{H}}(C_{u'},C_{v'})\le 2d-1$. 
As each edge in $\widehat{H}$ corresponds to a shortest path of length at most $2k/d+d$ in $H$, we have:
\begin{eqnarray*}
\dist_{H}(u',v')&\le& \max_{x\in C_{u'},y\in C_{v'}}\dist_{H}(x,y)
\\&\le& 2k/d+\dist_{\widehat{H}}(C_{u'},C_{v'})\cdot  (2k/d + d)
\leq 2k/d+(2d-1)\cdot  (2k/d+d)
\leq 4k+2d^2~.
\end{eqnarray*} 
Finally, again by Fact \ref{claim:cluster}(4), we have that $d_{H}(u,u')\le (2(k/d)-1)\cdot d_{G}(u,u')$ and $d_{H}(v',v)\le (2(k/d)-1)\cdot d_{G}(v',v)$. We therefore conclude that
\begin{eqnarray*}
\dist_{H}(u,v)&\le& \dist_{H}(u,u')+\dist_{H}(u',v')+ \dist_{H}(v',v)\\&\leq& (2(k/d)-1)\cdot d_{G}(u,u')
+ 4k+ 2 d^2+(2(k/d)-1)\cdot d_{G}(u,u')
\\&\le& (2(k/d)-1)\cdot d+4 k+2 d^2
\le 6k+2d^2\le (6.5 k/d)\cdot d~.
\end{eqnarray*}


\def\APPENDFIGFIRSTSPANNER{
\begin{figure}[h!]
\begin{center}
\includegraphics[scale=0.45]{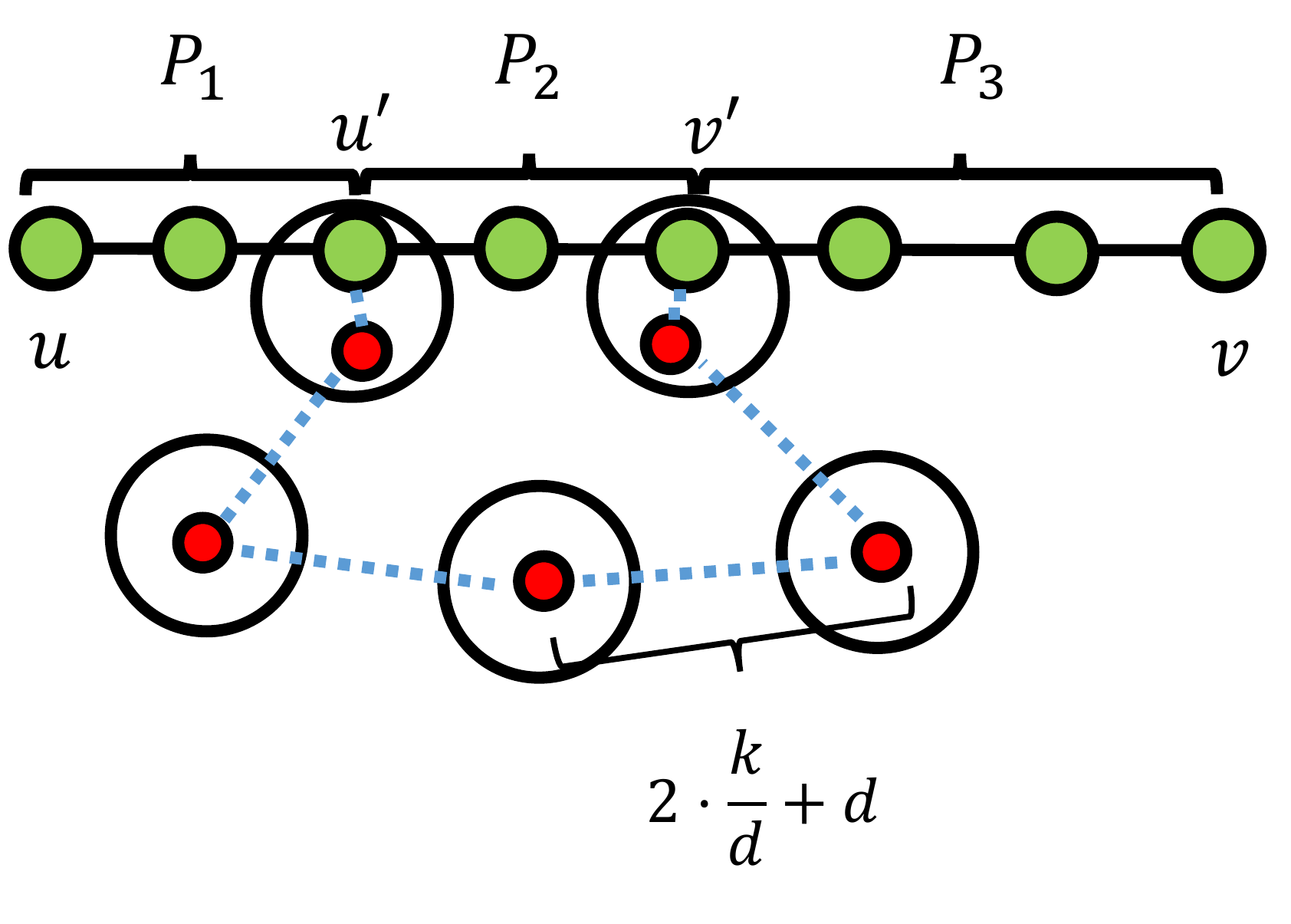}
\caption{\small Illustration for the proof of Theorem \ref{thm:sqrt}. Shown is a $u$-$v$ shortest path in $G$, where $u',v'$ are the leftmost and rightmost clustered vertices on the path. By definition, the clusters of $u'$ and $v'$ are neighbors in the cluster graph $\widehat{G}$. Thus they are connected by a path of $2d-1$ clusters in the cluster-spanner $\widehat{H}$, i.e., the spanner of the cluster graph.}
\label{fig:closespanner}
\end{center}
\end{figure}
}

\textbf{Size:} We show that for each fixed distance value $d$, the algorithm adds at most $O(k/d \cdot n^{1+1/k})$ edges to the spanner. By Claim \ref{claim:cluster}(3), the first $k/d$ steps of the Baswana-Sen clustering adds $O(k/d\cdot n^{1+1/k})$ edges in expectation. By Claim \ref{claim:cluster}(1), in expectation, $\mathcal{C}$ contains $N=n^{1+1/k-1/d}$ clusters. Thus, the cluster-graph $\widehat{G}$ has $N$ super-nodes, and the size of its $(2d-1)$-spanner $\widehat{H}$ is $N^{1+1/d}=O(n^{1+1/(k\cdot d)})$ in expectation. Since each edge in $\widehat{H}$ corresponds to a path of length at most $(2k/d+d)$ in $G$, we get that this step adds $d \cdot |E(\widehat{H})|=O(k/d \cdot n^{1+1/(k\cdot d)})$ edges to the spanner. The size argument follows. 
\end{proof}
\section{New $(k^{\epsilon},O_{\epsilon}(k))$ Spanners}\label{sec:second-regime-spanner}
In this section, we consider Theorem \ref{thm:secondspanner} and show the construction of $(\alpha,\beta)$ spanners that provide a nearly optimal stretch for all pairs at distance $\sqrt{k}/2 < d \leq k^{1-o(1)}$ and $d \geq k^{1+o(1)}$. For the sake of simplicity, we consider a fixed distance value $d$, and prove the following lemma: 
\begin{lemma}
\label{lem:secondspanner}
For any $n$-vertex unweighted graph $G=(V,E)$, any $0<\epsilon<\frac{1}{2}$ and integers $k>16^{1/\epsilon},d\ge 1$, there is a subgraph $H\subseteq{G}$ of expected size $O(k^{\epsilon}\cdot n^{1+1/k}+(64^{1/\epsilon}k^{1+\epsilon}+d)\cdot n)$ such that for every $u,v\in V$, at distance $d$ in $G$, it holds that $\dist_{H}(u,v)\le 4\cdot k^{\epsilon}\cdot d+1/6\cdot 64^{1/\epsilon}\cdot k$.
\end{lemma} 
We later on show how Theorem \ref{thm:secondspanner} follows by applying the construction of Lemma \ref{lem:secondspanner} algorithm for any $d \in [1,k^{1-\epsilon}]$. 
We now turn to describe our three stage procedure for computing the spanners of Lemma \ref{lem:secondspanner}.
\paragraph{Algorithm $\ImprovedSpannerII$.} The algorithm works in three stages. In the \textbf{first stage} it calls Procedure $\Cluster$ for $\left \lceil k^{\epsilon} \right \rceil$ clustering steps. This results with a clustering $\mathcal{C}_{0}$ of $O(n^{1-\frac{\lceil k^{\epsilon} \rceil}{k}})$ clusters in expectation, each with radius at most $\left \lceil k^{\epsilon} \right \rceil$. 

In the \textbf{second stage}, the number of clusters is dramatically reduced to $O(n^{1-\frac{1}{k^{\epsilon}}})$ while keeping the radius of each cluster to $O_{\epsilon}(k^{1-\epsilon})$. In the \textbf{last stage},  a cluster graph is computed on the collection of $O(n^{1-\frac{1}{\lceil k^{\epsilon} \rceil}})$  clusters, and a $(2\cdot \lceil k^{\epsilon} \rceil-3)$ spanner $\widehat{H}$ is computed on that graph. Each edge in $\widehat{H}$ will be translated into a $G$-path that is added to the final spanner.

\paragraph{Preliminary Stage: Truncated Baswana-Sen Algorithm.}
The algorithm starts by applying the first $\left \lceil k^{\epsilon} \right \rceil$ steps of Alg. $\Cluster$.
This results in a clustering $\mathcal{C}_{0}$ and a subgraph $H_{0}$. By the properties of the Baswana-Sen algorithm, $H_{0}$ has $O(k^{\epsilon}\cdot n^{1+1/k})$ edges in expectation and $\mathcal{C}_{0}$ consists of $n^{1-\frac{\left \lceil k^{\epsilon} \right \rceil}{k}}$ clusters in expectation with radius at most $\left \lceil k^{\epsilon} \right \rceil$. By Fact \ref{claim:cluster}(4), we have:
\begin{claim}\label{cl:helper-spanner-second}
For every unclustered vertex $u\in V$ and every vertex $v\in N_{G}(u)$, 
$\dist_{H_{0}}(u,v) \leq 2 \lceil k^{\epsilon} \rceil-1$.
\end{claim}

\paragraph{Middle Stage: Superclustering.} 
For clarity of presentation, throughout we assume that $\lceil k^{\epsilon} \rceil$ divides $4$, up to factor $4$ in the final stretch, this assumption can be made without loss of generality.
The middle step consists of $T:=\log_{\lceil k^{\epsilon} \rceil/4}(k^{1-2\epsilon})$ applications of 
Procedure $\ClusterAndAugment$. We refer to each application of this procedure by a \emph{phase}. 
For clarity of presentation, we also assume $T$ to be an integer, and in Sec. \ref{app:fraction} we describe how to remove this assumption.
In each phase $i$, the input to Procedure $\ClusterAndAugment$ is a clustering $\mathcal{C}_{i-1}=\{ C_1,\ldots, C_\ell\}$ of radius $r_{i-1}=O((2k^{\epsilon})^{i})$. The output of the phase is a clustering $\mathcal{C}_{i}$ of radius $r_i$, and a subgraph $H_i$ that takes care of all vertices that became unclustered in that phase. 
This output clustering is obtained by applying $t:=\lceil k^{\epsilon} \rceil/4$ steps of supercluster growing. As we will see, the superclustering procedure will be very similar to Proc. $\Cluster$ only that we will now treat each cluster $C \in \mathcal{C}_{i-1}$ as a \emph{node}.

Starting with the trivial superclustering $\mathcal{SC}_{i-1,0}=\{\{C_j\} ~\mid~ C_j \in \mathcal{C}_{i-1}\}$ whose radius is bounded by $r_{i-1,0}=r_{i-1}$, in the $j^{th}$ step of phase $i$ for $j\geq 1$, the algorithm is given a superclustering $\mathcal{SC}_{i-1,j-1}$. The radius of these superclusters will be bounded by $r_{i-1,j-1}$. The algorithm defines a superclustering $\mathcal{SC}_{i-1,j}$ along with a subgraph $H_{i,j}$ as follows. 
\begin{enumerate}
\item Each unclustered vertex $v \in V\setminus {V(\mathcal{SC}_{i-1,j-1})}$ satisfying that
$$\cdist(v, \mathcal{SC}_{i-1,j-1})\leq r_{i-1,j-1}+\alpha_{i-1,j}~$$
adds to $H_{i,j}$ the shortest paths to the closest center in $\mathcal{SC}_{i-1,j-1}$, where 
%
%
\begin{equation}\label{eq:alpha}
\alpha_{i-1,j}=\left \lceil \frac{4\cdot r_{i-1,0}}{\lceil k^{\epsilon} \rceil-3}+\left(1+\frac{4}{\lceil k^{\epsilon} \rceil-3}\right)\cdot \alpha_{i-1,j-1} \right \rceil~.
\end{equation}
\item Let $\mathcal{SC}' \subseteq \mathcal{SC}_{i-1,j-1}$ be the collection of superclusters obtained by 
sampling each supercluster $SC_\ell \in \mathcal{SC}_{i-1,j-1}$ independently with probability of $\frac{n_{0}}{n}$, where $n_{0}:=|\mathcal{C}_{i-1}|$. 
\item Set $\delta_{i-1,j}=r_{i-1,0}+r_{i-1,j-1}+2\cdot \alpha_{i-1,j}$. 
Each cluster $C \in SC_\ell$ at center-distance at most $\delta_{i-1,j}$ from $\mathcal{SC}'$ joins the supercluster of its closest center in $\mathcal{SC}'$, by adding the shortest path between the centers to $H_{i,j}$.
\item The superclustering $\mathcal{SC}_{i-1,j}$ consists of all sampled superclusters in $\mathcal{SC}'$ augmented by their nearby clusters (i.e., at center-distance at most $\delta_{i-1,j}$). The center of each sampled supercluster in $\mathcal{SC}'$ maintains its role in the augmented supercluster. The radius of this supercluster will be shown to be bounded by $r_{i-1,j}=r_{i-1,j-1}+2r_{i-1,0}+2\alpha_{i-1,j}$.
\item Each cluster $C \in SC_\ell$ at center-distance larger than $\delta_{i-1,j}$ from $\mathcal{SC}'$, adds to $H_{i,j}$ the shortest path from its center to the center of any supercluster in $\mathcal{SC}_{i-1,j-1}$ at center-distance at most $\delta_{i-1,j}$. 
\end{enumerate}
A vertex $u$ is said to be $(i-1,j)$-unclustered if $u$ belongs to the superclusters of $\mathcal{SC}_{i-1,j-1}$ but does not belong to the superclusters of $\mathcal{SC}_{i-1,j}$. 

The parameter $\alpha_{i-1,j}$ is set in a way that guarantees that for every $(i-1,j)$-unclustered vertex $v$ and every $u\in \Ball_{G}(v,\alpha_{i-1,j})$, the subgraph $H_{i,j}$ contains a $u$-$v$ path of length at most $\lceil k^{\epsilon} \rceil\cdot \alpha_{i-1,j}$ (hence providing a stretch of $\lceil k^{\epsilon} \rceil$ for every $v \in \partial \Ball_{G}(u,\alpha_{i-1,j})$).

Let $\mathcal{SC}_{i-1,t}$ be the output superclustering after the last $t$ step in the $i^{th}$ phase. Then the output clustering of the phase is given by $\mathcal{C}_i=\{\{V(SC_j)\} ~\mid~ SC_j \in \mathcal{SC}_{i-1,t}\}$.
That is, all the clusters in a given supercluster in $\mathcal{SC}_{i-1,t}$ form a single merged cluster in $\mathcal{C}_i$. Finally, let $H_i=\bigcup_{j} H_{i,j}$. 
This completes the description of the $i^{th}$ phase. After $T$ phases, the output clustering $\mathcal{C}_{T}$ is shown to contain at most $n^{1-1/k^{\epsilon}}$ clusters in expectation. Let $H=\bigcup_{i=0}^{T} H_i$ be the current spanner. 

\paragraph{Finalizing Stage: Spanner on the Cluster Graph.} 
Given the collection of $O(n^{1-1/k^{\epsilon}})$ clusters in $\mathcal{C}_{T}$, the algorithm first adds a shortest path from each unclustered vertex $v \notin V(\mathcal{C}_{T})$ to its closest cluster in $\mathcal{C}_{T}$ up to center-distance $r_{T}+d$, if such exists. Next, a cluster graph $\widehat{G}=(\mathcal{C}_{T},\mathcal{E})$ is defined by 
connecting two clusters $C,C' \in \mathcal{C}_{T}$ if their center-distance in $G$ is at most $2 r_{T}+2 d$.
That is, $\mathcal{E}=\{(C,C') ~\mid~ C,C' \in \mathcal{C}_{T} \mbox{~and~} \cdist_G(C,C')\leq 2 r_{T}+2d\}$. Let $\widehat{H}$ be a $k'$-spanner of $\widehat{G}$ for $k'=2\cdot \lceil k^{\epsilon} \rceil-3$. For edge $(C,C') \in \widehat{H}$, the shortest path between the centers of $C$ and $C'$ is added to the spanner $H$. 
This completes the description of the algorithm. 
\begin{algorithm}
	\begin{algorithmic}[1]
		\caption{$\ImprovedSpannerII(G,k,d,\epsilon)$}
		\label{alg:improvedspanner2}
		\State \textbf{Input:} A graph $G=(V,E)$, integers $k,d$ and $0<\epsilon<\frac{1}{2}$. 
		\State \textbf{Output:} A subgraph $H\subseteq{G}$ of expected size $O(k^{\epsilon}\cdot n^{1+1/k}+(64^{1/\epsilon}k+d)\cdot n)$, such that for any $u,v$ vertices at distance $d$ in $G$, $\dist_{H}(u,v)\le 4\cdot k^{\epsilon}\cdot d+O(64^{1/\epsilon} \cdot k)$.		
		\State Let $H\leftarrow \emptyset,T:=\log_{\lceil k^{\epsilon} \rceil/4}(k^{1-2\epsilon}),t:=\lceil k^{\epsilon}\rceil/4$.
		\State Let $(\mathcal{C}_{0},H_{0}) \leftarrow \Cluster(G,k,\lceil k^{\epsilon} \rceil)$.
		\For{$i=1$ to $T$}
		\State $(\mathcal{C}_{i},H_{i})\leftarrow \ClusterAndAugment(G,k,t,\epsilon,\mathcal{C}_{i-1})$.
		\State $H\leftarrow H\cup H_{i}$.
		\EndFor 
		\State For each $v\in V\setminus V(\mathcal{C}_T)$ with $\cdist_{G}(v,\mathcal{C}_T)\le r_{T}+d$, add to $H$ its shortest path to a center in $\mathcal{C}_{T}$.
		\State Let $\widehat{G}=(\mathcal{C}_T,\mathcal{E}=\{(C,C') ~\mid~ \cdist_{G}(C,C')\le 2\cdot r_{T}+2\cdot d\})$.
		\State Let $\widehat{H}=\BasicSpanner(\widehat{G},2\cdot \lceil k^{\epsilon} \rceil-3)$. 
		\State For each edge $(C,C')\in E(\widehat{H})$ add to $H$ the shortest path between the centers of $C,C'$.
		\State \Return $H$
	\end{algorithmic}

\end{algorithm}
	\begin{algorithm}[ht]
	\begin{algorithmic}[1]
		\caption{$\ClusterAndAugment(G,k,t,\epsilon,\mathcal{C})$}
		\label{alg:ClusterAndAugment}
		\State \textbf{Input:} A graph $G=(V,E)$, integers $k,t$, stretch parameter $0<\epsilon <1$ and a clustering $\mathcal{C}$.
		\State \textbf{Output:} A clustering $\mathcal{C'}$ and a subgraph $H\subseteq{G}$.
		\State Let $H\leftarrow \emptyset$, $r_{0} = rad_{G}(\mathcal{C})$,  $\alpha_{0}= 0$, and $n_{0}=|\mathcal{C}|$.
		\State Set $\mathcal{SC}_{0}\leftarrow \left\{\{C\}|C\in \mathcal{C}\right\}$, where for each $C\in \mathcal{C}$, $\{C\}$ is a supercluster containing all vertices in $C$, with $C$'s center as a center.
		\For {$j=1$ to $t$}
		\State Let $r_{j-1}=rad(SC_{j-1})$.
		\State Set $\alpha_{j}\leftarrow \left \lceil \frac{4\cdot r_{0}}{\lceil k^{\epsilon} \rceil-3}+(1+\frac{4}{\lceil k^{\epsilon} \rceil-3})\cdot \alpha_{j-1} \right \rceil$. 
		\State For each $v\in V\setminus{V(\mathcal{SC}_{j-1})}$ at center-distance at most $r_{j-1}+\alpha_{j}$ from the superclustering, add to $H$ a shortest path from $v$ to its closest center in $\mathcal{SC}_{j-1}$.
		\State $\mathcal{SC}_{j}\leftarrow Sample(\mathcal{SC}_{j-1},\frac{n_{0}}{n})$
		\For{every $SC\in \mathcal{SC}_{j-1}$}
		\For{every $C\in SC$} 
		\State If $\cdist_G(C, \mathcal{SC}_{j})\leq r_{0}+r_{j-1}+2\cdot \alpha_{j}$, add $C$ to its closest supercluster $SC' \in \mathcal{SC}_{j}$. Add to $H$ the shortest path from $r(C)$ to $r(SC')$.
		\State Otherwise, add to $H$ the shortest path from $r(C)$ to the center of any supercluster in $\mathcal{SC}_{j-1}$ at center-distance at most $r_{0}+r_{j-1}+2\cdot \alpha_{j}$ from $C$.
		\EndFor
		\EndFor
		\EndFor
		\State $\mathcal{C'}=\{\{V(SC_j)\} ~\mid~ SC_j \in \mathcal{SC}_t\}$.
		\State \Return $(\mathcal{C'},H)$.
	\end{algorithmic}
\end{algorithm}
\subsection*{Stretch Analysis.}
For the rest of the analysis, let $r_{0}:=rad(\mathcal{C}_{0})$ thus $r_0=\lceil k^{\epsilon} \rceil$, $t=\lceil k^{\epsilon} \rceil/4$, $T= \log_{\lceil k^{\epsilon} \rceil/4}(k^{1-2\epsilon})$, for $1\le i \le T$, $r_{i,0}=rad(\mathcal{C}_{i})$ and let $r_{T}:=rad(\mathcal{C}_T)$ be the radius of the clusters in the last clustering $\mathcal{C}_T$ at the end of the middle stage. 
For the sake of the stretch and size analysis, we will need the following two claims, bounding $\alpha_{i,j}$ and $r_{i,0}$, respectively, both of these claims follow by simple inductive arguments. 
\begin{claim}
\label{claim:alpha}
For every $i\in \{0,\ldots ,T-1\},j\in \left[\lceil k^{\epsilon}\rceil/4 \right]$, $\alpha_{i,j}\le (r_{i,0}+\frac{\lceil k^{\epsilon} \rceil-3}{4})\cdot\left((1+\frac{4}{\lceil k^{\epsilon} \rceil-3})^{j}-1\right)$.
\end{claim}

\begin{proof}
For the definition of $\alpha_{i,j}$ see Eq. (\ref{eq:alpha}). We first show by induction on $j$ that
$$\alpha_{i,j}\le  \left(1+\frac{4\cdot r_{i,0}}{\lceil k^{\epsilon} \rceil-3}\right)\cdot\sum_{p=0}^{j-1} \left(1+\frac{4}{\lceil k^{\epsilon} \rceil-3}\right)^{p}.$$ 
The base case, $j=1$ is trivial since $\alpha_{i,1} = \left\lceil \frac{4\cdot r_{i,0}}{\lceil k^{\epsilon} \rceil-3} \right \rceil$. 
Assuming that the claim holds for $j-1$, letting $\chi=\sum_{p=0}^{j-2}\left(1+\frac{4}{\lceil k^{\epsilon} \rceil-3}\right)^{p}$, we have:
\begin{eqnarray*}
\alpha_{i,j}&=& \left \lceil \frac{4\cdot r_{i,0}}{\lceil k^{\epsilon} \rceil-3}+\left(1+\frac{4}{\lceil k^{\epsilon} \rceil-3}\right)\cdot\alpha_{i,j-1} \right \rceil \\&\le& \left \lceil\frac{4\cdot r_{i,0}}{\lceil k^{\epsilon} \rceil-3}+\left(1+\frac{4}{\lceil k^{\epsilon} \rceil-3}\right)\cdot \left (1+\frac{4\cdot r_{i,0}}{\lceil k^{\epsilon} \rceil-3} \right)
\cdot \chi\right \rceil
\\&<& 1+\frac{4\cdot r_{i,0}}{\lceil k^{\epsilon} \rceil-3}+\left(1+\frac{4\cdot r_{i,0}}{\lceil k^{\epsilon} \rceil-3}\right)
\cdot
\sum_{p=1}^{j-1} \left(1+\frac{4}{\lceil k^{\epsilon} \rceil-3}\right)^{p}
\\&=& \left(1+\frac{4\cdot r_{i,0}}{\lceil k^{\epsilon} \rceil-3}\right)\cdot\sum_{p=0}^{j-1} \left(1+\frac{4}{\lceil k^{\epsilon} \rceil-3}\right)^{p}.
\end{eqnarray*}
Therefore, we have:
\begin{eqnarray*}
\alpha_{i,j}&\le& \left(1+\frac{4\cdot r_{i,0}}{\lceil k^{\epsilon} \rceil-3}\right)\cdot\sum_{p=0}^{j-1} \left(1+\frac{4}{\lceil k^{\epsilon} \rceil-3}\right)^{p} 
\\&=& \left(1+\frac{4\cdot r_{i,0}}{\lceil k^{\epsilon} \rceil-3}\right)\cdot \frac{\lceil k^{\epsilon} \rceil-3}{4} \cdot\left((1+\frac{4}{\lceil k^{\epsilon} \rceil-3})^{j}-1\right)
\\&=& \left(r_{i,0}+\frac{\lceil k^{\epsilon} \rceil-3}{4}\right)\cdot\left((1+\frac{4}{\lceil k^{\epsilon} \rceil-3})^{j}-1\right).
\end{eqnarray*}

\end{proof}
We next turn to bound the radii of the superclusters in each phase $i$ of the algorithm.
Note that in the $(i+1,j)$ step, since we add to each sampled supercluster, all clusters at center-distance at most $r_{i,0}+r_{i,j-1}+2\alpha_{i,j}$, the radius $r_{i,j}$ of the new supercluster is increased by an additive term of at most $2r_{i,0}+2\alpha_{i,j}$. 

\begin{claim}
\label{claim:radiusecondspanner}
If $k\ge 16^{1/\epsilon}$ then for each $0\le i \le T$ it holds that $r_{i,0} \le (2\cdot \lceil k^{\epsilon} \rceil)^{i}\cdot r_{0}$. 
In particular, for the final radius of the clustering $\mathcal{C}_T$ it holds that $r_{T}\leq 1/30 \cdot 64^{\frac{1-\epsilon}{\epsilon}}\cdot k^{1-\epsilon}$.
\end{claim}
\begin{proof}
We prove the claim by induction on $i$. The base case, $i=0$, is trivial. Assuming the claim is true for $i$ we show the correctness for $i+1$. Phase $i+1$ begins with clusters in $\mathcal{C}_i$ with radii $r_{i,0}$, and finishes after $t=\lceil k^{\epsilon} \rceil/4$ steps of Procedure $\ClusterAndAugment$ with clusters of radius at most $r_{i+1,0}$. We therefore bound $r_{i+1,0}$. At step $j$ of $\ClusterAndAugment$ we start with the superclustering $\mathcal{SC}_{i,j-1}$ of radius $r_{i,j-1}$, and we add to the superclusters, clusters of radius $r_{i,0}$ at center-distance at most $r_{i,0}+r_{i,j-1}+2\cdot \alpha_{i,j}$. Thus at the $j$th step we increase the radii of the superclusters by an additive factor of at most $2\cdot r_{i,0}+2\cdot \alpha_{i,j}$. Combining this with the bound on $\alpha_{i,j}$ of Claim \ref{claim:alpha}:
\begin{eqnarray}
r_{i,j} &=& r_{i,j-1}+2\cdot r_{i,0}+2\cdot \alpha_{i,j} \label{eq:rij}
=(2\cdot j+1)\cdot r_{i,0}+2\cdot \sum_{p=1}^{j}\alpha_{i,p} 
\\&\leq& (2\cdot j+1)\cdot r_{i,0} + 2\cdot \sum_{p=1}^{j}\left(r_{i,0}+\frac{\lceil k^{\epsilon} \rceil-3}{4}\right)\cdot\left((1+\frac{4}{\lceil k^{\epsilon} \rceil-3})^{p}-1\right) \nonumber
\\&\le& (2\cdot j+1)\cdot r_{i,0}-2\cdot j\cdot \left (r_{i,0}+\frac{\lceil k^{\epsilon} \rceil-3}{4} \right) 
+ 2\cdot \left (r_{i,0}+\frac{\lceil k^{\epsilon} \rceil-3}{4} \right)\cdot \sum_{p=1}^{j}\left(1+\frac{4}{\lceil k^{\epsilon} \rceil-3}\right)^{p}  \nonumber
\\&\le& r_{i,0}-j\cdot \frac{\lceil k^{\epsilon} \rceil-3}{2} + 2\cdot \left(r_{i,0}+\frac{\lceil k^{\epsilon} \rceil-3}{4} \right) \cdot \left(1+\frac{4}{\lceil k^{\epsilon} \rceil-3}\right)\cdot \frac{\lceil k^{\epsilon} \rceil-3}{4}\cdot \left((1+\frac{4}{\lceil k^{\epsilon} \rceil-3})^{j}-1\right) \nonumber
\\&=& r_{i,0}-j\cdot \frac{\lceil k^{\epsilon} \rceil-3}{2} + \left (r_{i,0}+\frac{\lceil k^{\epsilon} \rceil-3}{4} \right) \cdot \left(\frac{\lceil k^{\epsilon} \rceil+1}{2} \right) \cdot \left((1+\frac{4}{\lceil k^{\epsilon} \rceil-3})^{j}-1\right) \nonumber
\\&\le&\left (1+ \frac{\lceil k^{\epsilon} \rceil+1}{2}\cdot \left((1+\frac{4}{\lceil k^{\epsilon} \rceil-3})^{j}-1\right)\right) 
\cdot r_{i,0}-j\cdot \frac{\lceil k^{\epsilon} \rceil-3}{2}+\frac{k^{2\epsilon}}{8}\cdot \left((1+\frac{4}{\lceil k^{\epsilon} \rceil-3})^{j}-1\right)~. \nonumber
\end{eqnarray}
Thus by plugging $t=\lceil k^{\epsilon} \rceil/4$, we get:
\begin{eqnarray*}
r_{i+1,0}&:=&r_{i,t}\le \left (1+ \frac{\lceil k^{\epsilon} \rceil+1}{2}\cdot \left(e^{1+\frac{3}{\lceil k^{\epsilon} \rceil-3}}-1\right)\right)\cdot r_{i,0}-\frac{\lceil k^{\epsilon} \rceil}{4}\cdot \frac{\lceil k^{\epsilon} \rceil-3}{2}+\frac{k^{2\epsilon}}{8}\cdot \left(e^{1+\frac{3}{\lceil k^{\epsilon} \rceil-3}}-1\right)
\\&\le&\left (1+ \frac{\lceil k^{\epsilon} \rceil+1}{2}\cdot \left(e^{1+\frac{3}{\lceil k^{\epsilon} \rceil-3}}-1\right)\right)\cdot r_{i,0}+\frac{1}{2}\cdot \lceil k^{\epsilon} \rceil+\frac{k^{2\epsilon}}{8}\cdot \left(e^{1+\frac{3}{\lceil k^{\epsilon} \rceil-3}}-2\right)
\\&\le&\left (1+ 1.5\cdot (\lceil k^{\epsilon} \rceil+1)\right)\cdot r_{i,0}+\frac{1}{2}\cdot \lceil k^{\epsilon} \rceil+\frac{k^{2\epsilon}}{4} \leq \left(\frac{11}{6}\cdot \lceil k^{\epsilon} \rceil+2.5 \right)\cdot r_{i,0}\le 2\cdot \lceil k^{\epsilon} \rceil\cdot r_{i,0}~.
\end{eqnarray*}
The third inequality holds when $k\ge 11^{1/\epsilon}$, the fourth inequality follows as $r_{i,0}\ge \lceil k^{\epsilon} \rceil$, and the fifth inequality follows as $k\ge 16^{1/\epsilon}$.
Finally, by plugging $i=T$ we get that:
\begin{eqnarray}
r_{T}&\le& r_{0}\cdot (2\cdot \lceil k^{\epsilon} \rceil)^{\log_{\lceil k^{\epsilon} \rceil/4}(k^{1-2\epsilon})} \leq \lceil k^{\epsilon} \rceil\cdot (8\cdot (\lceil k^{\epsilon} \rceil/4))^{\log_{\lceil k^{\epsilon} \rceil/4}(k^{1-2\epsilon})} \nonumber
\\&\le&  \lceil k^{\epsilon} \rceil\cdot 8^{\frac{(1-2\epsilon)\cdot \log{k}}{\epsilon \cdot \log{k}-2}}\cdot k^{1-2\cdot \epsilon}\le (1+\frac{1}{k^{\epsilon}})64^{\frac{1-2\epsilon}{\epsilon}}\cdot k^{1-\epsilon}~\nonumber
 \leq 1/30\cdot 64^{\frac{1-\epsilon}{\epsilon}}\cdot k^{1-\epsilon}~,
\end{eqnarray}
where the fourth inequality holds for $k\ge 16^{1/\epsilon}$.
\end{proof}

\begin{figure}[h!]
\begin{center}
\includegraphics[scale=0.40]{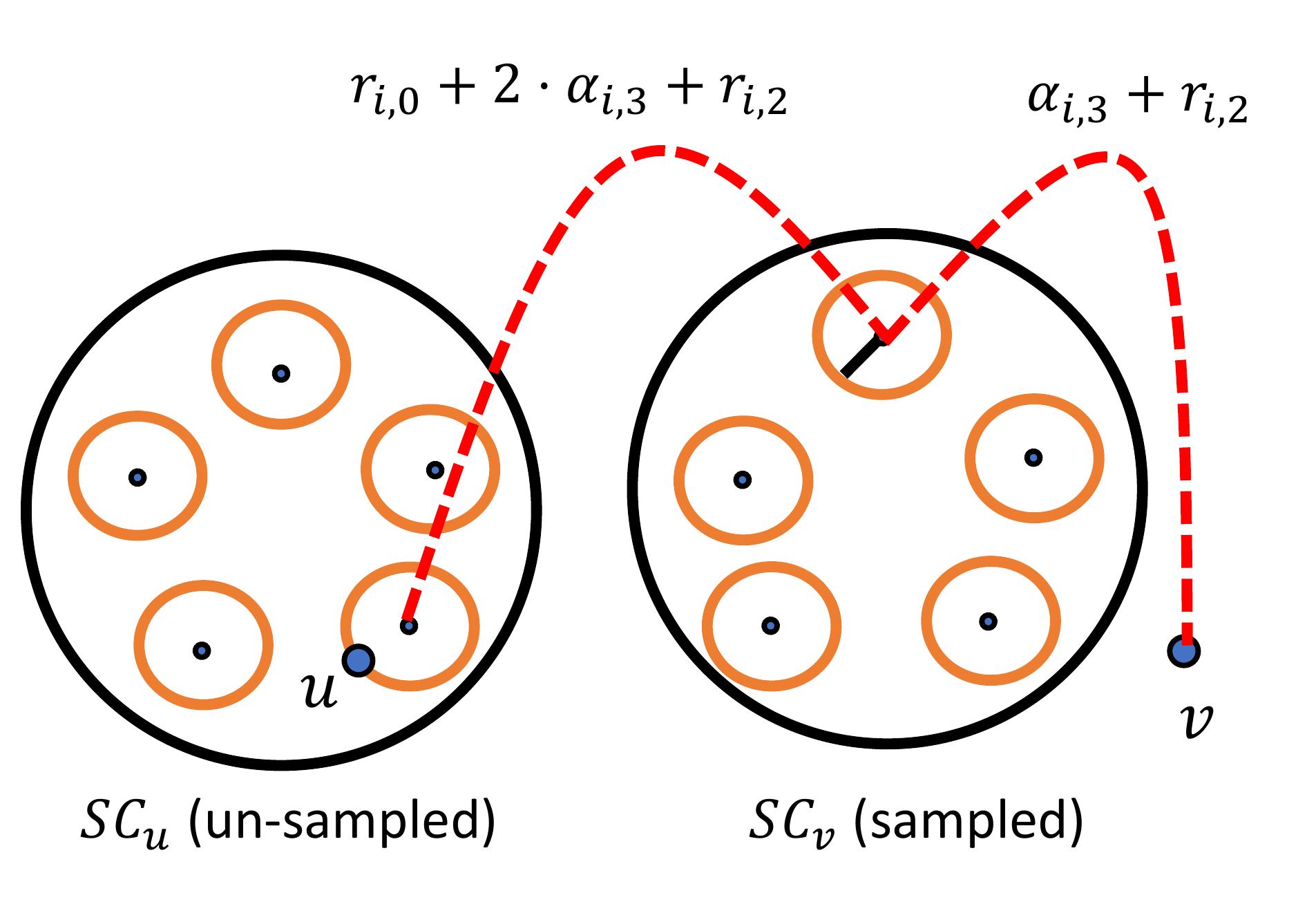}
\includegraphics[scale=0.40]{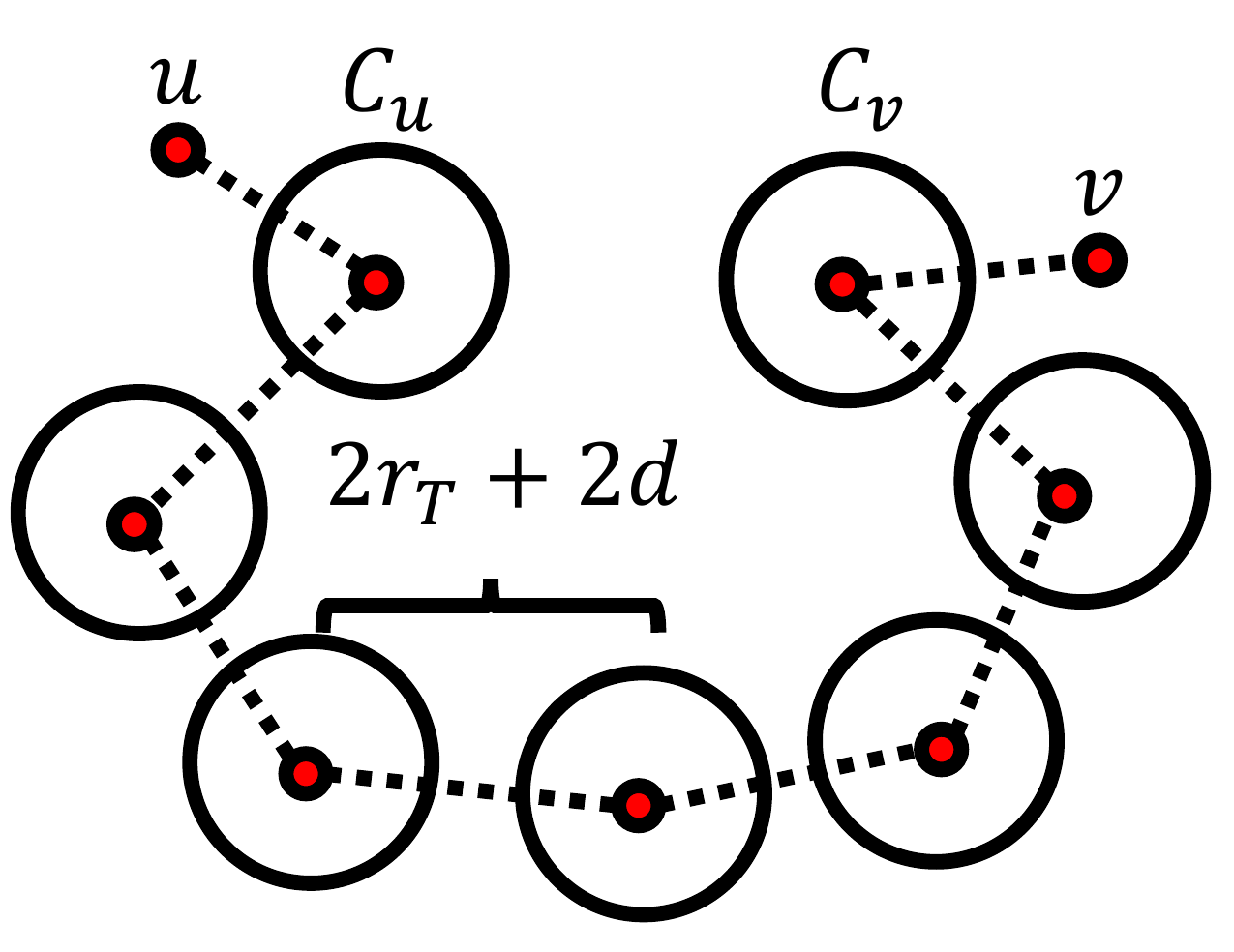}
\caption{\small Left: An illustration of the proof of Claim \ref{cl:ij-unc-secondspanner}. The vertex $u$ is an $(i,3)$-unclustered node and $v$ is a vertex at distance at most $\alpha_{i,3}$ from it. The algorithm adds a shortest path from $v$ to the center of some supercluster $SC_{v}$ in $\mathcal{SC}_{i,2}$. Since  $u$'s cluster $C_{u}$ is at center-distance at most $r_{i,0}+2\alpha_{i,3}+r_{i,2}$ from $SC_{v}$, the algorithm adds a path between the center of $C_{u}$ to the center of $SC_{v}$ of length at most $2r_{i,0}+3 \alpha_{i,3}+2r_{i,2}$.  Right: An illustration for the proof of Claim \ref{cl:clustered-secondspanner}. If there is a clustered vertex on the $u$-$v$ shortest path, then both $u$ and $v$ have a path in $H$ to their clusters $C_{u},C_{v}$ at center-distance at most $r_{T} + \dist_{G}(u,v)$. These clusters are neighbors in the cluster graph $\widehat{G}$, and thus there is a path of length at most $2\cdot \lceil k^{\epsilon} \rceil-3$ in the cluster-spanner $\widehat{H}$ between $C_{u}$ and $C_{v}$. Since each edge in $\widehat{H}$ is  translated into a path of length at most $2\cdot r_{T}+2\cdot d$, the path between $u$ and $v$ in $H$ is of length at most $4 k^{\epsilon} \cdot d+1/6\cdot 64^{\frac{1-\epsilon}{\epsilon}}\cdot k$.}
\label{fig:SecondSpannerAnalysis}
\end{center}
\end{figure}

\begin{definition}[Clustered and Unlcustered Vertices]
A vertex $v\in V$ is $0$-\emph{unclustered} if $v \notin V(\mathcal{C}_{0})$. A vertex $v$ is called $(i,j)$-unclustered if $v \in V(\mathcal{SC}_{i,j-1})\setminus{V(\mathcal{SC}_{i,j})}$. Finally, a vertex $v$ is called \emph{clustered} if $v\in V(\mathcal{C}_T)$, i.e., it belongs to the last level of clustering. 
\end{definition}

In the following claims we show that that the spanner provides a low-stretch path from $u$ to any vertex $v$ at some fixed distance from $u$ in $G$. The case of $0$-unclustered vertices follows by Claim \ref{cl:helper-spanner-second}, so it remains to consider $(i,j)$-unclustered vertices for $i\geq 1$. 
\begin{claim}
	\label{cl:ij-unc-secondspanner}  
	For $0\le i \le T-1,1\le j \le t$, for each $(i,j)$-unclustered vertex $u\in V$, $\dist_{H}(u,v) \le \lceil k^{\epsilon} \rceil \cdot \alpha_{i,j}$ for any $v\in \Ball_{G}(u,\alpha_{i,j})$.
\end{claim}
\begin{proof}
Let $v\in \Ball_{G}(u,\alpha_{i,j})$. In the beginning of the $(i+1,j)$ step we add to $H$ the shortest path from any vertex at center-distance at most $r_{i,j-1}+\alpha_{i,j}$ from $\mathcal{SC}_{i,j-1}$ to its closest center in the superclustering. In particular, since $u \in V(\mathcal{SC}_{i,j-1})$ and $\cdist_{G}(v,\mathcal{SC}_{i,j-1})\le \dist_{G}(v,u)+r_{i,j-1}\le \alpha_{i,j}+r_{i,j-1}$, we add the shortest path from $v$ to the center of some supercluster $SC_{v}\in \mathcal{SC}_{i,j-1}$. Since $u$ gets unclustered in this step, we add to $H$ the shortest path from its cluster, $C_{u}$, to any supercluster which is in $\mathcal{SC}_{i,j-1}$ and is at center-distance at most $r_{i,0}+r_{i,j-1}+2\alpha_{i,j}$. Since $\cdist_{G}(C_{u},SC_{v})\le r_{i,0}+\dist_{G}(u,v)+\cdist_{G}(v,SC_{v}) \le r_{i,0}+2\alpha_{i,j}+r_{i,j-1}$, we add the shortest path from the center of $C_{u}$ to the center of $SC_{v}$. Consequently, there is a path from $u$ to $v$ that goes through  the centers of $C_{u},SC_{v}$ of length at most:
\begin{eqnarray}
\dist_{H}(u,v)&\le& \cdist_{G}(u,C_{u})+\cdist_{G}(C_{u},SC_{v}) + \cdist_{G}(SC_{v},v) \nonumber
\\&\leq& 2 r_{i,0} + 2 \alpha_{i,j}+2r_{i,j-1}+\alpha_{i,j} \label{eq:ijuncspanner} 
\\&=&2\cdot r_{i,0}+3 \alpha_{i,j}+2\cdot r_{i,j-1} 
\le 3 \alpha_{i,j}+2 r_{i,0} + 2 \left((2 \cdot j-1)r_{i,0}+2\cdot \sum_{p=1}^{j-1}\alpha_{i,p}\right) \nonumber
\\&\le&3 \alpha_{i,j}+4  j\cdot r_{i,0}+4 \sum_{p=1}^{j-1}\alpha_{i,p}~, \label{eq:ijuncspannersecond}
	\end{eqnarray}
where in the second inequality, the bound on $r_{i,j-1}$ follows by Eq. (\ref{eq:rij}).
We next show by induction on $j$ that 
\begin{equation}\label{eq:alpha-induc}
(\lceil k^{\epsilon} \rceil-3)\cdot \alpha_{i,j}\ge 4  j\cdot r_{i,0}+4\cdot \sum_{p=1}^{j-1}\alpha_{i,p}~.
\end{equation}
The base case of $j=1$ holds vacuously. Assuming the correctness up to $j-1$, by Eq. (\ref{eq:alpha}):
\begin{eqnarray*}
\alpha_{i,j} &=& \frac{4 r_{i,0}}{\lceil k^{\epsilon} \rceil-3}+\left(1+\frac{4}{\lceil k^{\epsilon} \rceil-3}\right)\cdot\alpha_{i,j-1}
\ge \frac{4  j\cdot r_{i,0}+4\cdot \sum_{p=1}^{j-2}\alpha_{i,p}}{\lceil k^{\epsilon} \rceil-3}+\frac{4}{\lceil k^{\epsilon} \rceil-3}\cdot \alpha_{i,j-1} 
\\&=& \frac{4  j\cdot r_{i,0}+4\cdot \sum_{p=1}^{j-1}\alpha_{i,p}}{\lceil k^{\epsilon} \rceil-3}.
\end{eqnarray*}
Plugging Eq. (\ref{eq:alpha-induc}) in Eq. (\ref{eq:ijuncspannersecond}) we get that:
$\dist_{H}(u,v)\le 3 \alpha_{i,j}+(\lceil k^{\epsilon} \rceil-3)\cdot \alpha_{i,j}=\lceil k^{\epsilon} \rceil\cdot \alpha_{i,j}$.
\end{proof}
\begin{claim}
\label{cl:clustered-secondspanner}
Let $u,v\in V$ be vertices at distance $d$ in $G$, and let $P$ be a shortest path between them in $G$. If there is some clustered vertex $w\in P$, then $\dist_{H}(u,v) \le 4 k^{\epsilon} \cdot d+1/6\cdot 64^{\frac{1-\epsilon}{\epsilon}}\cdot k$.
\end{claim}
\begin{proof}
Let $C_{w}$ be the cluster to which $w$ belongs. In the beginning of the third stage of the algorithm we add shortest paths from unclustered vertices at center-distance at most $r_{T}+d$ to their closest cluster center in $\mathcal{C}_T$. Thus since $\cdist_{G}(u,\mathcal{C}_T)\le r_{T}+\dist_{G}(u,w)\le r_{T}+d$ and $\cdist_{G}(v,\mathcal{C}_T)\le r_{T}+\dist_{G}(v,w)\le r_{T}+d$, it holds that both $u,v$ add their shortest paths to the centers of some clusters $C_{u},C_{v}\in \mathcal{C}_T$ to the spanner. Observe that,
\begin{eqnarray*}
\cdist_{G}(C_{u},C_{v})&\le& \cdist_{G}(C_{u},u)+\dist_{G}(u,v)
+ \cdist_{G}(v,C_{v}) 
\\&\leq& r_{T}+\dist_{G}(w,u)+ d+\dist_{G}(v,w)+r_{T}
\le 2 r_{T}+2d~.
\end{eqnarray*}
Therefore, it holds that $\dist_{\widehat{G}}(C_{u},C_{v})\le 1$, thus $\dist_{\widehat{H}}(C_{u},C_{v})\le 2\cdot \lceil k^{\epsilon} \rceil-3$. Since each edge $(C,C')\in E(\widehat{H})$ translates into a path in $H$ of length at most $2r_{T}+2d$, we have that:
\begin{eqnarray*}
\dist_{H}(u,v) &\le&  \cdist_{G}(u,C_{u})+\dist_{\widehat{H}}(C_{u},C_{v})\cdot (2 r_{T}+2 d)+\cdist_{G}(C_{v},v)
\\&\le& \dist_{G}(u,w)+(2\cdot \lceil k^{\epsilon} \rceil -3)\cdot (2r_{T}+2 d)+\dist_{G}(v,w)+2 r_{T}
\\&<& (2k^{\epsilon} -1)\cdot (2 r_{T}+2 d)+ d+2 r_{T}\le 4 k^{\epsilon} \cdot d+4 k^{\epsilon} \cdot r_{T}\le 4 k^{\epsilon} \cdot d+1/6\cdot 64^{\frac{1-\epsilon}{\epsilon}}\cdot k~.
\end{eqnarray*}
\end{proof}

We next complete the proof of Lemma  \ref{lem:secondspanner}.
\begin{proof}[Proof of Lemma  \ref{lem:secondspanner}]
\textbf{Stretch.}
Let $u,v\in V$ be vertices at distance $d$ in $G$, and let $P$ be some shortest path between them in $G$. First observe that if there is some clustered vertex $w\in P$ the claim follows from Claim \ref{cl:clustered-secondspanner}, so we assume there is no such vertex. Partition the path $P$ into $\ell'\le d$ consecutive segments the following way: denote $v_{0}:=u$ and inductively define $v_{l+1}$ to be the vertex on $P$ at distance $\Delta_l$ on the segment $[v_{l},v]$, where:
	\begin{equation*}
	\Delta_l=\begin{cases}
	\min\{1,\dist(v_{l},v)\} & v_{l}\text{ is \ensuremath{0-unclustered}}\\
	\min\{\alpha_{i,j},\dist(v_{l},v)\} & v_{l}\text{ is \ensuremath{(i,j)-unclustered}}\\
	\end{cases}
	\end{equation*}
	Let $\ell'$ be the index of the last segment, thus $v_{\ell'}=v$.
	For any $l \in \{0,\ldots ,\ell'-1 \}$, if $v_{l}$ is $0$-unclustered then by Claim \ref{cl:helper-spanner-second} it holds that 
	$\dist_{H}(v_{l},v_{l+1})\leq 2\cdot \lceil k^{\epsilon} \rceil-1$. If $v_{l}$ is $(i,j)$-unclustered then since $\dist_{G}(v_{l},v_{l+1}) \le \alpha_{i,j}$, by Claim \ref{cl:ij-unc-secondspanner} it holds that $\dist_{H}(v_{l},v_{l+1})\leq \lceil k^{\epsilon} \rceil \cdot \alpha_{i,j}$. Thus except for at most the last segment $P[v_{\ell'-1},v]$, the spanner provides a multiplicative stretch of $\lceil k^{\epsilon}\rceil$ to each of the other segments. For the last segment $P[v_{\ell'-1},v]$, if $v_{\ell'-1}$ is $0$-unclustered then we are done by Claim \ref{cl:helper-spanner-second}. Otherwise, there exist $i,j \geq 1$ such that $v_{\ell'-1}$ is $(i,j)$-unclustered, and by Eq. (\ref{eq:ijuncspanner}) and Eq. (\ref{eq:rij}):
$$\dist_H(v_{\ell'}, v)\le 3 \alpha_{i,j}+2\cdot r_{i,0}+2 r_{i,j-1}\le 2 r_{i,j}\le 2 r_{T}.$$ Therefore by summing over all these at most $\ell'$ segments, and plugging the bound on $r_T$ from Claim \ref{claim:radiusecondspanner} we get that:
$$\dist_{H}(u,v)\leq \sum_{l=0}^{\ell'-1}\dist_{H}(v_{l},v_{l+1})\le (2\cdot \lceil k^{\epsilon} \rceil-1)\cdot \dist_{G}(u,v)+1/15\cdot 64^{\frac{1-\epsilon}{\epsilon}}\cdot k^{1-\epsilon}~.$$
\paragraph{Size Analysis.} By Fact \ref{claim:cluster}(3), $|E(H_{0})|=O(k^{\epsilon}\cdot n^{1+1/k})$ in expectation. 
Consider now the second stage. For any $1\le i \le T,1\le j\le t$, step $(i,j)$ starts by adding shortest paths from unclustered vertices to their closest centers in the superclustering. Since each vertex adds its shortest path to its closest center, and since we break ties in a consistent manner, this step adds at most $O(n)$ edges. The number of shortest paths added between unclustered clusters in each step can be bounded as follows.
\begin{claim}
\label{cl:expectedspinstep}
Fix a phase $i$. For any $1\le j \le t$, the algorithm adds in step $(i,j)$ a collection of $O(n)$ shortest paths in expectation.
\end{claim}
\def\APPENDCLMNUM{
\begin{proof}
Let $C$ be a cluster and let  $SC_{1},...,SC_{\ell}\in \mathcal{SC}_{i-1,j-1}$ be the superclusters that are at center-distance at most $r_{i-1,0}+r_{i-1,j-1}+2\alpha_{i-1,j}$ from the cluster $C$ in the superclustering $\mathcal{SC}_{i-1,j-1}$, ordered by non-decreasing distances from $C$. The cluster $C$ becomes unclustered only if all of the $\ell$ superclustered are unsampled. Since each supercluster in the $(i,j)$ step is sampled with probability $\frac{|\mathcal{C}_{i-1}|}{n}$, this happens with probability at most $(1-\frac{|\mathcal{C}_{i-1}|}{n})^{l}$. It follows that $C$ contributes $\ell\cdot (1-\frac{|\mathcal{C}_{i-1}|}{n})^{\ell}+1-(1-\frac{|\mathcal{C}_{i-1}|}{n})^{\ell} \leq \ell\cdot e^{-\ell\cdot |\mathcal{C}_{i-1}|/n}+1\le n/|\mathcal{C}_{i-1}|+1$
shortest paths in expectation. Since in the $i^{th}$ phase there are $|\mathcal{C}_{i-1}|$ clusters in expectation, we have that at each step we add at most $2n$ shortest paths in expectation. 
\end{proof}
}
Each shortest path is of length at most $r_{i-1,j}=r_{i-1,0}+r_{i-1,j-1}+2\alpha_{i-1,j}$. Therefore, by combining with Claim \ref{cl:expectedspinstep}, $|E(H_{i})|=O(k^{\epsilon} \cdot r_{i,0}\cdot n)$ in expectation. By summing over all $T$ phase, we get a total of $O(\sum_{i=1}^{T}k^{\epsilon}\cdot r_{i,0}\cdot n)=O(64^{1/\epsilon}\cdot k^{1+\epsilon}\cdot n)$ edges.

For the size analysis of last stage we will need the following claim, which follows by a simple induction.
\begin{claim}
\label{cl:sizesuperclusters}
\label{cl:numberofclusters}
For any $0\le i \le T,0\le j \le t$, the expected number of superclusters in $\mathcal{SC}_{i,j}$ is $n^{1-\frac{\left \lceil k^{\epsilon} \right\rceil \cdot (t+1)^{i}\cdot (j+1)}{k}}$, thus $|\mathcal{C}_{i}|=|\mathcal{SC}_{i,0}|=n^{1-\frac{\lceil k^{\epsilon} \rceil \cdot (t+1)^{i}}{k}}$, in expectation.
\end{claim}
%
From the claims above, we get that $|\mathcal{C}_T|=O(n^{1-(\left \lceil k^{\epsilon} \right \rceil/4)^{T}\cdot \frac{\left \lceil k^{\epsilon} \right \rceil}{k}})=O(n^{1-\frac{1}{k^{\epsilon}}})$, in expectation.
Consequently, $|E(\widehat{H})|=O(|\mathcal{C}_T|^{1+\frac{1}{\lceil k^{\epsilon} \rceil -1}})=O(n)$ (in expectation ). Since each edge in $\widehat{H}$ translates into a path of length at most $2r_{T}+2d$, this step contributes $O(r_{T}\cdot n)$ edges to the spanner. Overall, $|E(H)| = O(k^{\epsilon}\cdot n^{1+1/k}+(64^{1/\epsilon}k^{1+\epsilon}+d)\cdot n)$, in expectation.
\end{proof}
Theorem \ref{thm:secondspanner} now follows by noting that:
\begin{observation}
\label{claim:isenough}
Let $H\subseteq G$ be a subgraph satisfying that $\dist_{H}(u,v)\le \alpha\cdot d+\beta$ for every $u,v\in V(G)$ at distance $d\le \left \lceil \frac{\beta}{\alpha} \right\rceil$ in $G$, then $H$ is a $(2\cdot \alpha,3\cdot \beta)$-spanner of $G$.
\end{observation}
%

\section{New $(3+\epsilon,\beta)$ Spanner}
\label{sec:3epsspanner}
In this section, we show an optimized variant of our $(\alpha,\beta)$ spanner for the case where $\alpha=3+\epsilon$. For the purpose of efficient implementation, we settle for a slightly worse value of $\beta$. 
In Sec. \ref{sec:best-spanner}, we show an improved construction that achieves the bounds of Lemma \ref{lem:spanner-best}. 
Our main result is:
\begin{lemma}\label{lem:3plusepsspanner}
For any $n$-vertex unweighted graph $G=(V,E)$, integer $k$ and $\epsilon>0$, one can compute a $(3+\epsilon,\beta)$ spanner $H \subseteq G$ with $\beta=O((5+ 16/\epsilon)\cdot k^{\log(5+16/\epsilon)})$, and expected size $|E(H)|=O(n^{1+1/k}+k^{\log(5+16/\epsilon)}\cdot n/\epsilon)$.
\end{lemma}
We note that unlike the constructions in earlier sections, this construction works for all distances $d$, and there is no need to consider each distance class separately.
\paragraph{Algorithm Description.}	
For simplicity, we assume throughout that $4/\epsilon$ is an integer. The algorithm contains $T=\lceil \log k +1 \rceil$ clustering phases. Starting with the trivial clustering $\mathcal{C}_{0}= \{\{v\} ~\mid~ v\in V\}$ of radius $0$, in each phase $i\geq 1$, given is a clustering $\mathcal{C}_{i-1}$ of expected size $n_{i-1}=n^{1-\frac{2^{i-2}}{k}}$ (except for $i=1$ where $n_{0}=n$) with radius at most $r_{i-1}=2\cdot (5+ 16/\epsilon)^{i-2}$. The output of the $i^{th}$ phase is a clustering $\mathcal{C}_i$ of expected size $n_{i}=n^{1-\frac{2^{i-1}}{k}}$ and radius $r_i= 2\cdot (5+ 16/\epsilon)^{i-1}$, and a subgraphs $H_i$ that takes care of the unclustered vertices in $V(\mathcal{C}_i)\setminus V(\mathcal{C}_{i-1})$. 

We now zoom into the $i^{th}$ phase and explain the construction of the clustering $\mathcal{C}_i$ and the subgraph $H_i$. The phase is governed by two key parameters: the sampling probability $p_i$ of each cluster $C \in \mathcal{C}_{i-1}$ to join the clustering $\mathcal{C}_i$, and an augmentation radius $\alpha_i$.
Let $\alpha_1=1/2, p_1=n^{-1/k}$ and for every $i\geq 2$, define
\begin{equation}\label{eq:alpha-p}
\alpha_i = \frac{4}{\epsilon}\cdot r_{i-1} \mbox{~~and~~} p_i=|\mathcal{C}_{i-1}|/n~.
\end{equation}
The description of the $i^{th}$ phase for $i \in \{1,\ldots, T\}$ is as follows:
\begin{enumerate}

\item Each unclustered vertex $v \in V\setminus {V(\mathcal{C}_{i-1})}$ with $\cdist_{G}(v,\mathcal{C}_{i-1})\le r_{i-1}+\alpha_{i}$, adds to $H_{i}$ the shortest path to its closest center in $\mathcal{C}_{i-1}$.

\item Let $\mathcal{C}' \subseteq \mathcal{C}_{i-1}$ be the collection of clusters obtained by 
sampling each cluster $C_\ell \in \mathcal{C}_{i-1}$ independently with probability of $p_i$.

\item Each cluster $C \in \mathcal{C}_{i-1}$ such that $\cdist_{G}(C,\mathcal{C}')\le 4\cdot r_{i-1}+4\cdot \alpha_{i}$, joins the sampled cluster in $\mathcal{C}'$ with minimal distance between their centers, and adds the shortest path between their centers into $H_{i}$.

\item The clustering $\mathcal{C}_{i}$ consists of all the sampled clusters in $\mathcal{C}'$ augmented by their nearby clusters in $\mathcal{C}_{i-1}$ (i.e., at distance $4\cdot r_{i-1}+4\cdot \alpha_{i}$ between their centers). That is, each cluster in $\mathcal{C}_{i}$ is made of a star of clusters, with the head of the star is a sampled cluster $C \in \mathcal{C}'$ whose center $r(C)$ is connected to the centers of a subset of clusters in $\mathcal{C}_{i-1}$.

\item Each cluster $C \in \mathcal{C}_{i-1}$ such that $\cdist_{G}(C,\mathcal{C}')> 4\cdot r_{i-1}+4\cdot \alpha_{i}$, adds to $H_{i}$ the shortest path between its center to any center of  $C'\in \mathcal{C}_{i-1}$ with $\cdist_{G}(C,C')\le 2\cdot r_{i-1}+2\cdot \alpha_{i}$. 
\end{enumerate}
This completes the description of the $i^{th}$ phase. Let $H=\bigcup_{i=1}^{T} H_i$ and let $\mathcal{C}_T$ be the output clustering of the last phase $T$. In the analysis section we show that in expectation $\mathcal{C}_T$ consists of at most a single cluster $C$. In the latter case, the algorithm adds to the output spanner $H$, a BFS tree rooted at the centers of the clusters of $\mathcal{C}_T$. 

\paragraph{Stretch Analysis.}
We start by bounding the radius $r_i$ of the clustering $\mathcal{C}_i$ for every $i \in \{1,\ldots, T\}$.
We first make a simple observation:

\begin{observation}
\label{obs:rad3epsspanner}
For every $i \in \{1,\ldots, T\}$, $r_i \leq 2\cdot (5+ 16/\epsilon)^{i-1}$. In particular, the radius of cluster in the final clustering $\mathcal{C}_T$ is $r_T \le (10+ 32/\epsilon)\cdot k^{\log(5+ 16/\epsilon)}$.
\end{observation}
\begin{proof}
The claim is shown by induction on $i$. For the base case, where $i=1$, the $0$-level clusters are simply singletons, and each node joins its closest sampled vertex at distance at most $4\cdot r_{0}+4\cdot \alpha_1=2$. Therefore the clusters in $\mathcal{C}_1$ have radius of $2$. Assume that the claim holds up to $i-1 \geq 0$, and consider the $i^{th}$ clustering $\mathcal{C}_i$ which is defined in the $i^{th}$ phase based on the clustering $\mathcal{C}_{i-1}$. Each cluster in $\mathcal{C}_i$ is formed by a star: the head of the star is the sampled cluster $C$ that is connected to all clusters $C' \in \mathcal{C}_{i-1}$ with center-distance at most $4\cdot r_{i-1}+4\alpha_{i}$ from $C$. The radius of this cluster in $\mathcal{C}_i$ is bounded by
\begin{eqnarray*}
r_i &\leq& r_{i-1}+4\cdot r_{i-1}+4\alpha_{i}
\\&=& 5\cdot r_{i-1}+16/\epsilon \cdot r_{i-1} \leq 2\cdot \left (5+16/\epsilon \right)^{i-1}~,
\end{eqnarray*}
where the second inequality follows by plugging the bound on $r_{i-1}$ obtained from the induction assumption and the bound on $\alpha_i$ from Eq. (\ref{eq:alpha-p}).
\end{proof}
\begin{observation}
\label{obs:clusters3epsspanner-nclusters}
For $1\le i \le T$, $|\mathcal{C}_{i}|=n^{1-2^{i-1}/k}$ therefore after $T$ phases, there is one cluster in $\mathcal{C}_T$ in expectation.
\end{observation}
\begin{proof}
By induction on $i$. For the base case $i=1$, since $p_1=n^{-1/k}$, the number of sampled clusters is $n^{1-1/k}$ in expectation. Assume that the claim holds up to $i-1 \geq 1$, and consider the $i^{th}$ clustering. In the $i$th phase, each cluster in $\mathcal{C}_{i-1}$ is sampled with probability of $p_i=|\mathcal{C}_{i-1}|/n$, therefore, the number of sampled clusters in expectation is $|\mathcal{C}_{i-1}|^2/n$. By induction assumption, $|\mathcal{C}_{i-1}|=n^{1-2^{i-2}/k}$, and thus $|\mathcal{C}_{i}|=n^{1-2^{i-1}/k}$.
By plugging $T=\lceil \log k+1 \rceil$, we get that in expectation $|\mathcal{C}_{T}|\le 1$.
\end{proof}
Our stretch argument is based on the following definition of clustered vertices.
\begin{definition}
For every $1 \leq i\leq T$, a vertex $v \in V(\mathcal{C}_{i-1})\setminus{V(\mathcal{C}_{i})}$ is called $i$-\emph{unclustered}. In addition, every $v\in V(\mathcal{C}_T)$ is called \emph{clustered} (i.e., $v$ belongs to the final clustering). 
\end{definition}
For every $i$-unclustered vertex we provide a stretch guarantee as a function of $i$. Specifically, the earlier that the vertex $u$ stops being clustered, the smaller is the ball around $u$ for which the stretch guarantee is provided. 
\begin{claim}[$1$-unclustered]
\label{cl:1-unc}
For any $1$-unclustered vertex $u\in V$, it holds that $\dist_{H}(u,v) = 1$ for all $v\in N_{G}(u)$.
\end{claim}
\begin{proof}
For the first phase in step $(5)$, every $1$-unclustered vertex $u$ adds to the spanner its edges to any vertex at distance $2\alpha_1=1$, the claim follows.
\end{proof}

\begin{figure}[h!]
\begin{center}
\includegraphics[scale=0.35]{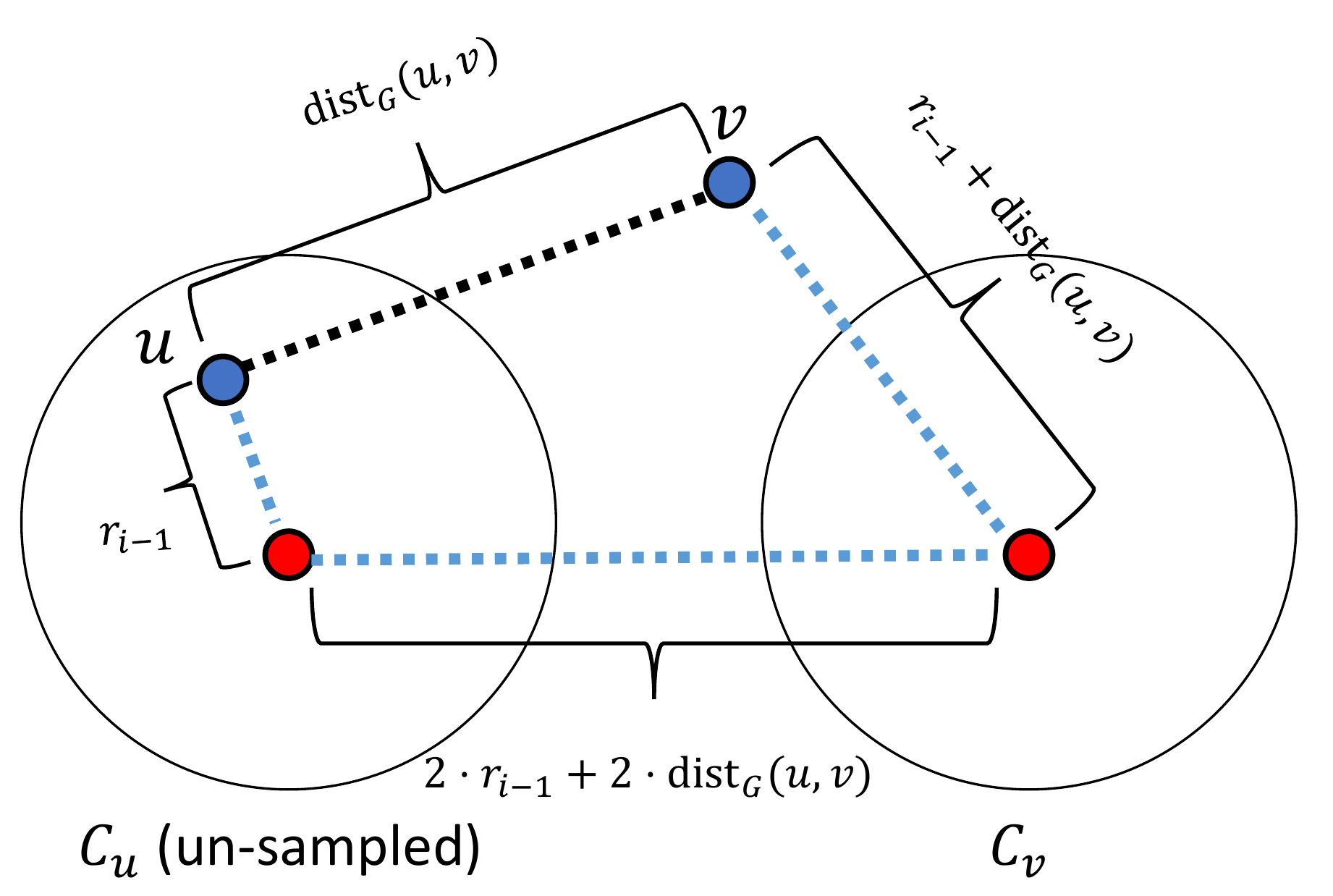}
\caption{\small An illustration for the proof of Claim \ref{claim:i-unc-3epsspanner}. Shown is a $u-v$ path between an $i$-unclustered vertex $u \in C_u$ where $C_{u}\in \mathcal{C}_{i-1}$, and a vertex $v\in \Ball_{G}(u,\alpha_{i})$. In the step (1) of the $i^{th}$ phase, the algorithm adds to the spanner a path between $v$ and the center of some cluster $C_{v} \in \mathcal{C}_{i-1}$. In step (5), the algorithm adds a path between the centers of $C_{u}$ and $C_{v}$. Thus, the spanner $H$ contains a short path between $u$ to $v$ that goes through the centers of $C_{u}$ and $C_{v}$.}
\label{fig:3epsspanner}
\end{center}
\end{figure}

\begin{claim}[$i$-unclustered]
\label{claim:i-unc-3epsspanner}
For any $i$-unclustered vertex $u$, for every $i \in \{2,\ldots, T\}$, and every $v \in \Ball_{G}(u,\alpha_i)$, it holds that $\dist_H(u,v)\leq 4\cdot r_{i-1}+3\cdot \dist_{G}(u,v)$. 
\end{claim}
\begin{proof}
Consider an $i$-unclustered vertex $u$ and $v \in \Ball_{G}(u,\alpha_i)$. By definition, $u \in V(\mathcal{C}_{i-1})$, thus $\cdist_G(v, \mathcal{C}_{i-1})\leq r_{i-1}+\dist_{G}(u,v)$. Let $C_v$ be closest cluster to $v$ with respect to the $\cdist()$ measure, and let $C_u$ be the cluster of $u$ in $\mathcal{C}_{i-1}$. 
Then, the algorithm adds to $H_i$ a shortest path from $v$ to the center of $C_{v}$ in step (1). 
We have:
\begin{eqnarray*}
\cdist_{G}(C_{u},C_{v})&\le& \cdist_{G}(C_{u},u) + \dist(u,v)+ \cdist_{G}(v,C_{v})\le 2 r_{i-1}
+ 2\dist_{G}(u,v)\leq 2r_{i-1}+2\alpha_i~.
\end{eqnarray*}
Since $u$ becomes unclustered phase $i$, in step $(5)$ the algorithm adds to $H_i$ the shortest path from the center of $C_{u}$ to the center of $C_{v}$. We therefore have:
\begin{eqnarray*}
\dist_{H}(u,v)&\le& \cdist_G(u,C_u)+\cdist_G(C_u,C_v)
+ \cdist_G(C_{v},v) 
\\&\leq& 
r_{i-1}+2r_{i-1}+ 2\dist_{G}(u,v)+r_{i-1}+\dist_{G}(u,v)
\\&\leq& 4r_{i-1}+3\dist_{G}(u,v)~.
\end{eqnarray*}
\end{proof}
In particular, for $v \in \partial \Ball_{G}(u,\alpha_i)$, by Eq. (\ref{eq:alpha-p}) it holds that $\dist_{H}(u,v)\le (3+\epsilon)\cdot \dist_{G}(u,v)$.  Since the algorithm adds to the spanner, a BFS tree w.r.t the centers of the clusters in $\mathcal{C}_T$, we have that:
\begin{claim}
\label{claim:clustered-3epsspanner}
Let $u,v\in V$ and let $P_{uv}$ be the shortest path between them in $G$. If there is a clustered vertex $w\in P_{uv}$ then $$\dist_{H}(u,v)\le \dist_{G}(u,v)+ 2\cdot (10+32/\epsilon)\cdot k^{\log(5+ 16/\epsilon)}.$$
\end{claim}
\begin{proof}
Since $w$ is clustered, its distance from the cluster center $z$ is at most $r_{T}$. As the algorithm adds the BFS tree of $z$ to the spanner, we have:
\begin{eqnarray*}
\dist_{H}(u,v)&\le& \dist_{H}(u,z)+\dist_{H}(z,v)=\dist_{G}(u,z)+\dist_{G}(z,v)
\\&\le& \dist_{G}(u,w)+\dist_{G}(w,z)+\dist_{G}(z,w)+\dist_{G}(w,v)\le \dist_{G}(u,v)+2\cdot r_{T}~,
\end{eqnarray*}
the claim follows by plugging the bound on $r_T$ from Obs. \ref{obs:rad3epsspanner}.
\end{proof}
We are now ready to prove Lemma \ref{lem:3plusepsspanner}.
\begin{proof}[Proof of Lemma \ref{lem:3plusepsspanner}]
Let $u,v\in V$ be vertices at some distance $d$ in $G$, and let $P:=\{u:=v_{0},\cdots , v_{d}:=v\}$ be the shortest path between them in $G$. First observe that by Claim \ref{claim:clustered-3epsspanner}, if there is some clustered vertex in $P$ then we are done. Assume from now on that no vertex on the path is clustered. We define a sequence of vertices between $u$ and $v$ in an iterative manner: let $u_{0}:=u$, and for every $j \in \{1,\ldots, d-1\}$ given that $u_{j}=v_{i}$, define $u_{j+1}=v_{i+\Delta_{j}}$ where: 
\begin{equation*}
\Delta_{j}=\begin{cases}
\min\{1,\dist_{G}(u_{j},v)\} & u_{j}\text{ is \ensuremath{1}-unclustered}\\
\min\{\alpha_i,\dist_{G}(u_{j},v)\} & u_{j}\text{ is \ensuremath{i\ge 2} unclustered}
\end{cases}
\end{equation*}
Let $\ell\leq d$ be the minimal index such that $u_{\ell}=v$, thus we have defined $\ell$ segments $P_{j}=[u_{j-1},u_{j}]$ for $j \in \{1,\ldots, \ell\}$. 
By Claim \ref{cl:1-unc}, the first case in the definition of $\Delta_{j}$ causes no stretch. By Claim \ref{claim:i-unc-3epsspanner} the second cases causes a multiplicative stretch of $3+\epsilon$, unless it ends in $v$. The latter case can happen only for the last segment, and in this case by Claim \ref{claim:i-unc-3epsspanner} we will get an extra additive stretch of $4\cdot r_{T-1}$, which by Observation \ref{obs:rad3epsspanner} is at most $8\cdot k^{\log(5+16/\epsilon)}$. In other words, for every segment $P[u_{j-1},u_{j}]$ such that $u_{j}\neq v$, we have that $\dist_H(u_{j-1},u_{j})\leq (3+\epsilon)\cdot |P[u_{j-1},u_{j}]|$. 
For the last segment $P[u_{\ell-1}, v]$ we have that $\dist_H(u_{\ell-1}, v)\leq 3\cdot |P[u_{\ell-1}, v]|+4\cdot r_{T-1}$. 
Therefore, 
\begin{eqnarray*}
\dist_{H}(u,v)&\le&\sum_{j=1}^{\ell}\dist_{H}(u_{j-1},u_{j})
\le 8\cdot k^{\log(5+16/\epsilon)} +\sum_{j=1}^{\ell}(3+\epsilon)\cdot \dist_{G}(u_{j-1},u_{j})
\\&\le& (3+\epsilon)\cdot \dist_{G}(u,v)+8\cdot k^{\log(5+16/\epsilon)}.
\end{eqnarray*}

\textbf{Size Analysis.} 
Fix a phase $i$. In step $(1)$ each vertex adds to $H_{i}$ a path of length at most $r_{i-1}+\alpha_i$ to its closest cluster center in $\mathcal{C}_{i-1}$. This adds $O(n)$ edges in total by breaking shortest path ties in a consistent manner. In step $(3)$ we add to $H_{i}$ the paths that connect clusters in $\mathcal{C}_{i-1}$ to the sampled clusters at distance at most $4\cdot r_{i-1}+4\cdot \alpha_{i}$. This also consume at most $O(n)$ edges by breaking shortest path ties in a consistent manner. In step $(5)$ we add to $H_{i}$ paths from clusters that were not sampled and did not join  other clusters, to nearby clusters at distance at most $2\cdot r_{i-1}+2\cdot \alpha_{i}$. To bound the number of edges added to the spanner in this step,  we will need the following:
\begin{lemma}
\label{lem:expectedamountofsp}
The number of shortest paths added to $H_{i}$ in step $(5)$ is at most $1/p_{i}$ in expectation. 
\end{lemma}
\begin{proof}
	Let $C\in \mathcal{C}_{i-1}$ and let $C_{1}\cdots C_{\ell}$ be the clusters with center-distance at most $2\cdot r_{i-1}+2\cdot \alpha_{i-1}$ from $C$. Note that the centers of each pair of these clusters is at distance at most $4\cdot r_{i-1}+4\cdot \alpha_{i-1}$, thus if one of these clusters is sampled, then all the others will join some sampled cluster. That is, none will take part in the fifth step of the algorithm. Thus the number of shortest paths from a given cluster $C$ added in this step is $\ell\cdot (1-p_{i})^{\ell}\le 1/p_{i}$ in expectation. 
\end{proof}
We now bound the total number of edges added in step $(5)$ of phase $i$. 
In phase $i$ there are $|\mathcal{C}_{i-1}|$ clusters and each adds to the spanner a shortest path of length at most $2\cdot r_{i-1} +2\cdot \alpha_{i}$. Thus by Lemma \ref{lem:expectedamountofsp} this adds at most $(2\cdot r_{i-1} +2\cdot \alpha_{i})\cdot |\mathcal{C}_{i-1}|/p_{i}$ edges in expectation. By plugging the values of $p_{i}$ as defined in Eq. (\ref{eq:alpha-p}), it follows that for $i=1$, we get a total of $n^{1+1/k}$ edges. For $i\ge 2$ , the total number of edges is bounded by $r_i \cdot n$. By plugging the values of $r_{i-1}$ of Observation \ref{obs:rad3epsspanner} we get that the expected number of edges in all the $H_i$ subgraphs is bounded by
\begin{eqnarray*}
\sum_{i=1}^T |E(H_i)|&\leq & n^{1+1/k}+\sum_{i=2}^{T}r_{i}\cdot n
=n^{1+1/k}+O(n \cdot \sum_{i=2}^{T}(5+ 16/\epsilon)^{i-1})
=O(n^{1+1/k}+k^{\log(5+16/\epsilon)}\cdot n/\epsilon)~.
\end {eqnarray*}
The last step adds a constant number of BFS trees, thus contributing $O(n)$ edges. Lemma \ref{lem:3plusepsspanner} follows.
\end{proof}
	
\section{A New Family of $(k^{\epsilon}, k^{1-\epsilon})$ Hopsets}\label{sec:hopsets-new}
In this section we present new construction of $(\alpha,\beta)$ hopsets with $O_k(n^{1+1/k})$ edges and $\alpha\cdot \beta=O(k)$. The structure of the section is as follows. First, in Subsec. \ref{sec:first-regime-hopset} we show the construction of $O(k^{\epsilon}, k^{1-\epsilon})$ hopsets in the simpler regime of $\epsilon \in [1/2,1)$. Then, in Subsec. \ref{sec:second-regime-hopset} we show the high level construction of $(O(k^{\epsilon}), O_{\epsilon}(k^{1-\epsilon}))$ hopsets for the complementary regime of $\epsilon \in (0,1/2)$. Finally, in Sec. \ref{sec:3epshopset} we also show the construction of $(3+\epsilon,\beta)$ hopsets.

\subsection{$(k^{\epsilon}, k^{1-\epsilon})$ Hopsets for $\epsilon \in [1/2,1)$.}\label{sec:first-regime-hopset}
This subsection is devoted for proving Theorem \ref{thm:new-hopset} of $\epsilon \in [1/2,1)$. We show the following:
\begin{theorem}
\label{thm:hopset-first-regime}
For any $n$-vertex weighted graph $G=(V,E,w)$, integer $k\geq 1$ and $\epsilon \in [1/2,1)$ such that $k\ge 10^{1/\epsilon}$, one can compute an $(\alpha,\beta)$ hopset $H$ for $\alpha=18\cdot k^{\epsilon}$ and $\beta=8\cdot k^{1-\epsilon}+1$ where $|E(H)|=O(k^{\epsilon}\cdot n^{1+1/k} \cdot \log \Lambda)$ edges in expectation. 
\end{theorem}
For simplicity we focus on a fixed distance class by considering all vertex pairs at distance $[d/2,d]$. The algorithm is then applied for each of the $\log \Lambda$ distance classes. Throughout, when adding edges $(u,v)$ to the hopset, we set the weight of these edges to $\dist_G(u,v)$. In addition, to distinguish between real $G$-edges and $H$-edges, we refer the latter by \emph{hops}.
The algorithm has a similar structure to Alg. $\ImprovedSpannerI$, and the key difference is that we use here the Thorup-Zwick hopsets to handle the sparse case, rather then applying the truncated Baswana-Sen algorithm.  
Specifically, the \textbf{first step} of the algorithm computes a $(2k-1,2)$ hopset $H_{TZ}$ be applying the algorithm of Thorup and Zwick. Let $A_{\left \lceil k^{\epsilon} \right \rceil}$ be the $\lceil k^{\epsilon} \rceil$th level of centers, and define $\mathcal{C}$ as the clusters of weighted radius 
\begin{equation}\label{eq:first-hopset-rad}
r_0=d\cdot \frac{\lceil k^{\epsilon} \rceil +1}{k^{1-\epsilon}}~
\end{equation}
centered at the vertices of $A_{\left \lceil k^{\epsilon} \right \rceil}$. Specifically, every vertex $u$ of distance at most $r_0$ from $A_{\left \lceil k^{\epsilon} \right \rceil}$ joins the cluster of its closest center in $A_{\left \lceil k^{\epsilon} \right \rceil}$, this vertex is now \emph{clustered}. In the hopset, we connect each clustered vertex $u$ to the center of its cluster. 

In the \textbf{second step}, a cluster graph $\widehat{G}$ is defined on the clusters of $\mathcal{C}$ in the exact same manner as in Alg. $\ImprovedSpannerI$. That is, any two clusters $C,C'$ in $\mathcal{C}$ are neighbors in $\widehat{G}$ if their $G$-center-distance is at most $2\cdot r_{0}+d$. Note that as before, the cluster graph $\widehat{G}$ is unweighted. Letting $\widehat{H}$ be the $(2k^{1-\epsilon}-1)$-spanner on $\widehat{G}$, then for every edge $(C,C')\in \widehat{H}$, the algorithm adds to the hopset, an hop $(r(C),r(C'))$ between the centers $r(C),r(C')$ of the clusters $C,C'$. This completes the description of the algorithm. 

\begin{algorithm}
\begin{algorithmic}[1]
\caption{$\ImprovedHopsetI(G,k,d,\epsilon)$}
		\State \textbf{Input:} A weighted graph $G=(V,E,w)$, integers $k,d$ and $\frac{1}{2}\le \epsilon\le 1$ such that $k\ge 10^{1/\epsilon}$. 
		\State \textbf{Output:} A $(k^{1-\epsilon}, k^{\epsilon})$ hopset $H$ for vertices at distance $[d/2,d]$.
		\State Let $(\mathcal{C},H_{0}) \leftarrow \ClusterHop(G,k,\lceil k^{\epsilon} \rceil,r_0)$ (see Eq. (\ref{eq:first-hopset-rad})).
		\State Let $H\leftarrow H_{0}$.
		\State Let $\widehat{G}=(\mathcal{C},E'=\{(C,C')~ \mid~ C,C'\in \mathcal{C},\cdist_{G}(C,C')\le 2\cdot r_{0}+d\})$.
		\State Let $\widehat{H}=\BasicSpanner(\widehat{G},2\cdot k^{1-\epsilon}-1)$.
		\State For each edge $(C,C')\in E(\widehat{H})$, add $(r(C),r(C'))$ to $H$.
		\State \Return $H$
	\end{algorithmic}
\end{algorithm}

	\begin{algorithm}[H]
		\begin{algorithmic}[1]
			\caption{$\ClusterHop(G,k,l,r_{0})$ \label{alg:TZ}}
			\State \textbf{Input:} An $n$-vertex weighted graph $G=(V,E,w)$ and integers $l\le k-1,r_{0} \geq 1$.
			\State \textbf{Output:} A clustering $\mathcal{C}_{0}$ of $G$ with $O(n^{1-\frac{l}{k}})$ clusters in expectation of radius $r_{0}$, and a set of hop edges $H'$ of expected size $|E(H')|=O(l \cdot n^{1+1/k})$.
			\State Set $H'\leftarrow \emptyset$.
			\State $(A_{l},H_{TZ})\leftarrow \TZ(G,k,l)$. \Comment{Apply the truncated TZ algorithm for $l$ steps.}
			\State $H'\leftarrow H_{TZ}$
			\State Let $\mathcal{C}$ be clusters of radius $r_0$ centered at $A_{l}$. 
			\For{every cluster $C_j\in \mathcal{C}$ with center $z_j$}
			\State Add the hop $(u,z_j)$ to $H'$ for every $u \in C_j$.
			\EndFor
			\State \Return $\langle \mathcal{C},H' \rangle$
		\end{algorithmic}
	\end{algorithm}
We are now ready to complete the proof of Thm. \ref{thm:hopset-first-regime}. 
Throughout, a vertex $v$ is \emph{clustered} if $v\in V(\mathcal{C})$, and \emph{unclustered} otherwise.
\begin{claim}\label{cl:helper-sparse-hopset}
For every unclustered vertex $x$ and every vertex $y\in \Ball_{G}(x,r_{0}/(\lceil k^{\epsilon} \rceil +1))$, it holds that:
$$\dist^{(2)}_{G \cup H}(x,y) \leq (2 \lceil k^{\epsilon} \rceil+1) \cdot \dist_G(x,y).$$
\end{claim}
\begin{proof}
Since $x$ is unclustered, we have that $\Ball_{G}(x,r_0)\cap A_{\left \lceil k^{\epsilon} \right \rceil}=\emptyset$. That is, $\dist_G(x,A_{\left \lceil k^{\epsilon} \right \rceil})> r_0$.  
Let $i_x$ be the minimal index $i$ such that $p_{i}(x) \in B_{i}(y)$, define $i_y$ analogously and let $i^*=\min\{i_x,i_y\}$.
Assume towards contradiction that $i^* \geq \lceil k^{\epsilon} \rceil$.
First assume that $i_x \leq i_y$, then by Fact \ref{fact:TZ} (i) we have that
\begin{eqnarray*}
\dist_G(x,A_{\lceil k^{\epsilon} \rceil})&\le& \dist_G(x,p_{\lceil k^{\epsilon} \rceil}(x))
\leq \lceil k^{\epsilon} \rceil \cdot \dist_{G}(x,y)
\leq \lceil k^{\epsilon} \rceil \cdot \frac{r_{0}}{\lceil k^{\epsilon} \rceil +1}< r_0~,
\end{eqnarray*}
thus leading to a contradiction as $x$ is unclustered.  
Otherwise, assume that $i_y \leq i_x$, then again by Fact \ref{fact:TZ}(i) we have that
\begin{eqnarray*}
\dist_G(x,A_{\lceil k^{\epsilon} \rceil})&\leq& \dist_G(x,y)+\dist_G(y,p_{\lceil k^{\epsilon} \rceil}(y))
\leq (\lceil k^{\epsilon} \rceil+1) \cdot \dist_{G}(x,y)\le r_0~,
\end{eqnarray*}
contradiction again.
Thus, we have that $i^* \leq \lceil k^{\epsilon} \rceil$ and by Fact \ref{fact:TZ} (ii),
$\dist^{(2)}_{H_{TZ} \cup G}(x,y)\leq (2\lceil k^{\epsilon} \rceil+1) \cdot \dist_G(x,y)$ as desired.%
%
%
\end{proof}

\begin{figure}[h!]
\begin{center}
\includegraphics[scale=0.35]{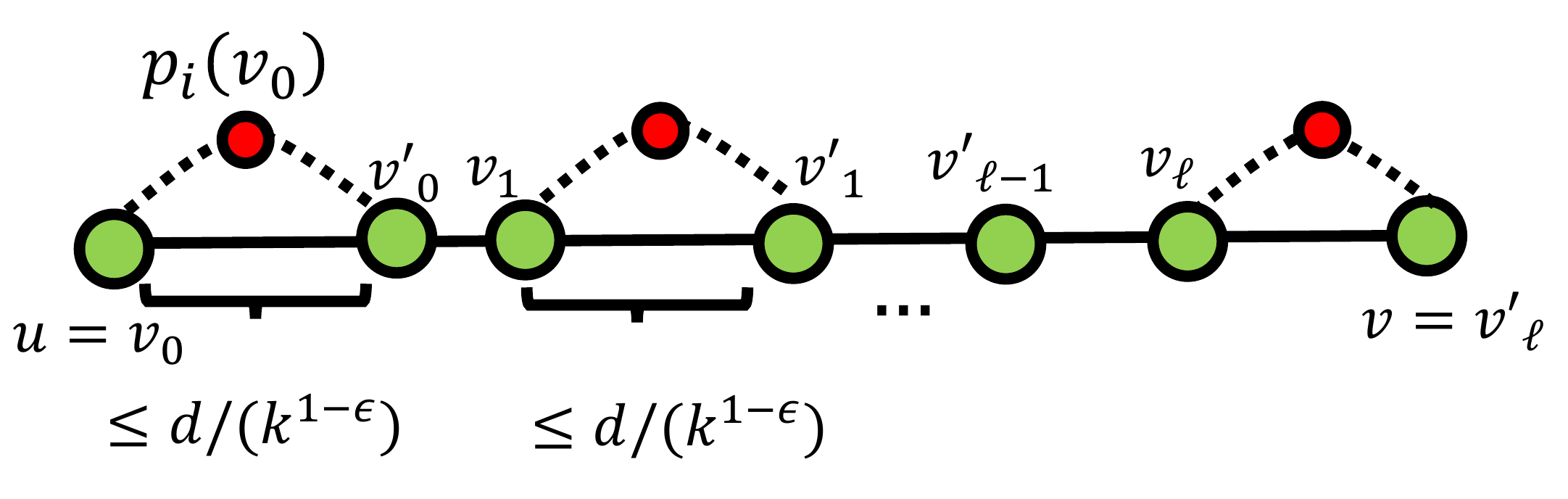}
\caption{\small Illustration of the proof of Theorem \ref{thm:hopset-first-regime}. Shown is a $u-v$ path $P$ where all its vertices are unclustered. By Claim \ref{cl:helper-sparse-hopset}, from each vertex $x$ on the path there is a $2$-hop path to any vertex $y\in \Ball_{G}(x,d/k^{1-\epsilon})$ with at most $(2 \lceil k^{\epsilon} \rceil+1)$-stretch. Thus by taking a $2$-hop path from a vertex $x$ to the furthest vertex $x' \in P$ at distance at most $d/k^{1-\epsilon}$ from $x$, and then going on its subsequent incident edge of $P[x',v]$ yields a $3$-hop path with a stretch of $(2 \lceil k^{\epsilon} \rceil+1)$.}
\label{fig:closehopset}
\end{center}
\end{figure}
We next complete the proof of Theorem  \ref{thm:hopset-first-regime}.
\begin{proof}
\textbf{Stretch and Hop-Bound.}
Fix a pair $u,v$ at distance $d'\in [d/2,d]$ in $G$, and let $P$ be their shortest-path in $G$. First, assume that at most one vertex in $P$ is clustered. 
Partition the path $P$ into $\ell \le k^{1-\epsilon}$ consecutive segments from $u$ to $v$ in the following way: denote $v_{0}:=u$, and for every $i\ge 0$ inductively define $v'_{i}$ to be the far-most vertex on $P[v_{i},v]$ at distance\footnote{Shortest path distance and not hop-distance.} at most $d/k^{1-\epsilon}$ from $v_{i}$. Observe that it might be the case that $v'_{i}$ is simply $v_{i}$, in the case where the edge incident to $v_{i}$ on the segment $P[v_i,v]$ segment has weight larger than $d/k^{1-\epsilon}$. Also, let $v_{i+1}$ be the vertex incident to $v'_i$ on the segment $P[v'_{i},v]$. In the case where $v'_i=v$,  also let $v_{i+1}=v'_{i}=v$.

Observe that for every $i$, if $v'_{i}\neq v_{i+1}$ (which happens when $v'_i\neq v$) then it must hold that $\dist_{G}(v_{i},v_{i+1})\ge d/k^{1-\epsilon}$. Therefore the path $P$ is partitioned into at most $k^{1-\epsilon}$ segments, where the $i^{th}$ segment is $P_{i}:=[v_{i-1},v_{i}]$ and $P=P_1 \circ \ldots \circ P_{\ell}$, where $\ell$ is the index of the last segment that reaches $v$.

For any $i \in \{0,\ldots, \ell-1\}$, $\dist_{G}(v_{i},v'_{i})\le d/k^{1-\epsilon}$ and (except maybe for $i=\ell-1$) $\dist_{G}(v_{i},v_{i+1})\ge d/k^{1-\epsilon}$. For $v_{i}\neq v'_{i}$, by the assumption that $P$ has at most one clustered vertex, it holds that at least one of the vertices $v_i,v'_i$ is unclustered. Thus by Claim \ref{cl:helper-sparse-hopset}, it holds that 
$\dist^{(2)}_{G \cup H_{TZ}}(v_{i},v'_{i})\leq (2\cdot \lceil k^{\epsilon} \rceil+1)\cdot \dist_{G}(v_{i},v'_{i})$, and consequently using the edge $(v'_{i},v_{i+1})\in P$, $\dist^{(3)}_{G \cup H_{TZ}}(v_{i},v_{i+1})\leq (2\cdot \lceil k^{\epsilon} \rceil+1)\cdot \dist_{G}(v_{i},v_{i+1})$. In the last segment we might have that $v_{\ell}=v'_{\ell-1}=v$, and thus since $\dist_{G}(v_{\ell-1},v_{\ell})\leq d/k^{1-\epsilon}$, by Claim \ref{cl:helper-sparse-hopset} again, we get that $\dist^{(2)}_{G \cup H_{TZ}}(v_{\ell-1},v)\leq (2\cdot \lceil k^{\epsilon} \rceil+1)\cdot \dist_{G}(v_{\ell-1},v)$. Finally,
\begin{eqnarray*}
\dist_{G\cup H}^{(3\cdot k^{1-\epsilon})}(u,v)&\le& \sum_{i=0}^{\ell-1}\dist_{G\cup H}^{(3)}(v_{i},v_{i+1})
\leq (2\cdot \lceil k^{\epsilon} \rceil+1)\cdot \dist_{G}(u,v)~.
\end{eqnarray*}

Next, assume that $P$ contains at least two clustered vertices. Let $u',v'$ be the leftmost and rightmost clustered vertices on $P$, and let $P_{1}:=[u,u'],P_{2}:=[u',v'],P_{3}:=[v',v]$ such that $P=P_{1}\circ P_{2}\circ P_{3}$. 
Let $C_{u'}, C_{v'} \in \mathcal{C}$ be the clusters of $u',v'$ respectively, and let $z_{u'}$ and $z_{v'}$ be the centers of these clusters. Thus the hopset $H$ contains a hop from $u'$ to $z_{u'}$ and a hop from $v'$ to $z_{v'}$.
%
Since $\cdist_G(C_{u'},C_{v'})\leq 2r_{0}+\dist(u',v')\leq 2r_{0}+d~,$
it holds that $(C_{u'},C_{v'}) \in \widehat{G}$. This in turn implies that $\dist_{\widehat{H}}(C_{u'},C_{v'})\leq (2\cdot k^{1-\epsilon}-1)$. Thus in the hopset $H$ we have a path of at most $2+(2\cdot k^{1-\epsilon}-1)$ hops from $u'$ to $v'$: one hop from $u'$ to $z_{u'}$, then $(2\cdot k^{1-\epsilon}-1)$ hops from $z_{u'}$ to $z_{v'}$, and the last hop from $z_{v'}$ to $v'$. Since all the clusters in $\mathcal{C}$ have radius $r_0=d\cdot \frac{\lceil k^{\epsilon} \rceil +1}{k^{1-\epsilon}}$, and since we connect clusters $C,C'$ in the cluster-graph $\widehat{G}$ if $\cdist_G(C,C')\leq 2r_{0}+d$, we have that the distance in $G$ between the centers of adjacent clusters in $\widehat{G}$ is at most  $d+2r_0$. Therefore, each edge in $\widehat{G}$ has a weight at most $d+2r_0$. Overall since $\epsilon\ge 1/2$, we have:
\begin{eqnarray*}
\dist^{(2\cdot k^{1-\epsilon}+1)}_{H \cup G}(u',v')&\leq& 2\cdot r_{0}+(2\cdot k^{1-\epsilon}-1)(d+2r_0)
\leq (6\cdot  k^{\epsilon} +7) \cdot d ~.
\end{eqnarray*}
Finally, since each of the segments $P[u,u']$ and $P[v',v]$ contains at most one clustered vertex, by the argument of the first case, there is a path of at most $(3\cdot k^{1-\epsilon})$ hops from $u$ to $u'$, and from $v'$ to $v$ that provides a multiplicative stretch of a $(2\cdot \lceil k^{\epsilon} \rceil +1)$ for each of the segments $P[u,u']$ and $P[v',v]$. Therefore, we have that
\begin{eqnarray*}
\dist_{G\cup H}^{\left(8\cdot k^{1-\epsilon}+1 \right)}(u,v)&\le& \dist_{G\cup H}^{\left(3\cdot k^{1-\epsilon}\right)}(u,u')
+\dist_{G\cup H}^{(2k^{1-\epsilon}+1)}(u',v')
+\dist_{G\cup H}^{\left(3k^{1-\epsilon}\right)}(v',v)
\\&\le& (2\lceil k^{\epsilon} \rceil +1)d+(6\cdot  k^{\epsilon} +7)d
\leq (8k^{\epsilon} +10) d \leq 9 k^{\epsilon}\cdot d~,
\end{eqnarray*}
where the last inequality holds as $k^{\epsilon}\ge 10$.
Since we assume $\dist_{G}(u,v)\in [d/2,d]$, the final stretch is bounded by $\alpha\leq 18\cdot k^{\epsilon}$ stretch. See Figure \ref{fig:closehopset} for an illustration. \newline
\textbf{Size.} We show that the total number of edges added to $H$ is bounded by $O(k^{\epsilon}\cdot n^{1+1/k})$ in expectation. Step (I) adds at most $O(k^{\epsilon}\cdot n^{1+1/k})$ edges by Fact  \ref{fact:TZ}(iii).
Step (II) adds $n$ edges, between possibly each vertex to its closest cluster in $\mathcal{C}$. Finally, in Step (III) we add $|E(\widehat{H})|$ edges to the hopset. Since $\widehat{G}$ contains $n^{1-\lceil k^{\epsilon} \rceil/k}$ clusters, the $(2k^{1-\epsilon}-1)$ spanner $\widehat{H}$ contains at most $O(n^{(1-1/k^{1-\epsilon})(1+1/k^{1-\epsilon})})=O(n)$ edges.
\end{proof}

\subsection{$(O(k^{\epsilon}),O_{\epsilon}(k^{1-\epsilon}))$ Hopsets for $0<\epsilon <1/2$.}\label{sec:second-regime-hopset}
Finally, we consider $(k^{\epsilon},k^{1-\epsilon})$ hopsets for the complementary regime of $\epsilon \in (0,1/2)$, which will complete the proof of Theorem \ref{thm:new-hopset}. This regime is considerably more involved and bares similarities with the spanner construction of Sec. \ref{sec:second-regime-spanner}. 
Specifically, we show:
\begin{theorem}
\label{thm:secondhopset}
For any $n$-vertex weighted graph $G=(V,E,w)$, integer\footnote{The statement can work for any $k$ upon suffering from larger constants.} $k\ge 16^{1/\epsilon}$, and $0<\epsilon<\frac{1}{2}$, one can compute an 
$(\alpha,\beta)$ hopset $H$ for $\alpha=9\cdot k^{\epsilon}$ and $\beta=36^{1/\epsilon} \cdot k^{1-\epsilon}$, of expected size $O((k^{\epsilon}\cdot n^{1+1/k}+\frac{k^{\epsilon}}{\epsilon}\cdot n)\log \Lambda)$.
\end{theorem} 
To prove the theorem, we will show the following key lemma which restricts attention to a fixed distance class $[d/2,d]$.
\begin{lemma}
\label{lem:secondhopset}
For any $n$-vertex weighted graph $G=(V,E,w)$, integers $k\ge 16^{1/\epsilon},d\ge 1$ and $0<\epsilon<\frac{1}{2}$, there is a hopset $H$ of expected size $O(k^{\epsilon}\cdot n^{1+1/k}+\frac{k^{\epsilon}}{\epsilon}\cdot n)$ such that for every $u,v\in V$, at distance $d'\in [d/2,d]$ in $G$, it holds that $\dist_{H}^{(36^{1/\epsilon}\cdot k^{1-\epsilon})}(u,v)\le 4.5\cdot k^{\epsilon}\cdot d$.
\end{lemma} 
\paragraph{Algorithm $\ImprovedHopsetII$.} Fix a distance range $[d/2,d]$. The same procedure will be repeated for every distance range. The algorithm works in three stages. In the \textbf{first stage} it calls Procedure $\ClusterHop$ for $\left \lceil k^{\epsilon} \right \rceil$ steps and radius parameter $r_{0}=d/R'$ where $R'=O_{\epsilon}(k^{1-2\cdot \epsilon})$ (see Algorithm \ref{alg:TZ}). This results in a partial hopset $H_{0}$, and a clustering $\mathcal{C}_{0}$ of $O(n^{1-\frac{1}{k^{1-\epsilon}}})$ clusters in expectation, each with radius $r_0=d/R'$. We say that a vertex is $0$-unclustered if it is not in $V(\mathcal{C}_{0})$. By the end of this stage, we will have the guarantee that for every $0$-unclustered vertex $v$ and every vertex $u \in \Ball_G(v,O_{\epsilon}(d/k^{1-\epsilon}))$, 
the current hopset contains a $2$-hop path $P'$ from $u$ to $v$ of length at most $O_{\epsilon}(d/k^{1-2\epsilon})$. 

The \textbf{second stage} applies contains $T=O(\frac{1}{\epsilon})$ phases of superclustering. In each phase $i$, the procedure runs for $t=\lceil k^{\epsilon} \rceil /4 $ steps. We refer to the $j^{th}$ step of the $i^{th}$ phase, by step $(i,j)$. Step $(i,j)$ begins with a superclustering $\mathcal{SC}_{i-1,j-1}$. The output of that step is a superclustering $\mathcal{SC}_{i-1,j}$ along with a collection of hop edges to be added to the hopset. 

We say that a vertex is $(i,j)$-unclustered if it is in $V(\mathcal{SC}_{i,j-1})\setminus V(\mathcal{SC}_{i,j})$. The analysis shows that for each $(i,j)$-unclustered vertex $v$, the edges added to the hopset in the $(i,j)$ step provides a $3$-hop $u$-$v$ path $P'$ of length $O(k^{\epsilon}\cdot \alpha_{i,j})$ to any vertex $u \in \Ball_{G}(v,\alpha_{i,j})$. The parameter $\alpha_{i,j}$ grows at each step $(i,j)$ but it is bounded by $\alpha_{i,j}=O(1.5^i \cdot e^{4j/k^{\epsilon}} \cdot d/k^{1-(i+2)\cdot \epsilon})$. 

At the beginning of the \textbf{third stage} we have a clustering $\mathcal{C}_T$ with  $O(n^{1-\frac{1}{k^{\epsilon}}})$ clusters in expectation, of (weighted) radius at most $d$. First, the algorithm connects each vertex $v$, satisfying that $\cdist(v, \mathcal{C}_T)\leq r_T+d$, to its closest center. Next, a cluster graph $\widehat{G}=(\mathcal{C}_T,\mathcal{E})$ is defined by letting $\mathcal{E}=\{(C,C') ~\mid~ \cdist(C,C')\leq 2r_{T}+d, ~C,C' \in \mathcal{C}_T \}$. Let $\widehat{H}$ be an $(2\cdot \lceil k^{\epsilon} \rceil-3)$ spanner of $\widehat{G}$. For every edge $(C,C') \in \widehat{H}$, the algorithm adds an hop between $C$ and $C'$ to the hopset $H$. This completes the high level description of the algorithms. 
%
%

\paragraph{First Stage: Initial Clustering.}
The algorithm starts by applying Procedure $\ClusterHop$ for $\left \lceil k^{\epsilon} \right \rceil$ steps with radius parameter $r_{0}:=d/R'$ where $R':=1/2\cdot 36^{\frac{1}{\epsilon}}\cdot k^{1-2\epsilon}$. This results in a clustering $\mathcal{C}_{0}$ and a partial hopset $H_{0}$. By the properties of Procedure $\ClusterHop$ and the chosen parameters, the clustering has $n^{1-\left \lceil k^{\epsilon} \right \rceil/k}$ clusters, in expectation, of radius at most $d/R'$. The partial hopset $H_{0}$ contains $O(k^{\epsilon}\cdot n^{1+1/k})$ edges, in expectation. By the exact same argument as in Claim \ref{cl:helper-sparse-hopset}, we have that:
\begin{claim}
For every $0$-unclustered vertex $u\in V$ and every $v\in \Ball_{G}(u,r_{0}/(\lceil k^{\epsilon} \rceil +1 ))$, it holds that:
$\dist^{(2)}_{G \cup H}(u,v) \leq (2 \lceil k^{\epsilon} \rceil+1) \cdot \dist_G(u,v)~.$
\end{claim}

\paragraph{Middle Stage: superclustering.} 
For clarity of presentation, throughout we assume that $\lceil k^{\epsilon} \rceil$ divides $4$, up to factor $4$ in the final stretch, this assumption can be made without loss of generality.
The middle step consists of $T:=\log_{\lceil k^{\epsilon} \rceil/4}(k^{1-2\epsilon})$ applications of 
Procedure $\ClusterAndAugmentHop$. For clarity of presentation we assume that $T$ is an integer, and in Sec. \ref{app:fraction} we describe how to remove this assumption.
We refer to each application of this procedure by a \emph{phase}. 
In each phase $i$, the input to Procedure $\ClusterAndAugmentHop$ is a clustering $\mathcal{C}_{i-1}=\{ C_1,\ldots, C_\ell\}$ of radius $r_{i-1}=r_{i-1,0}$ where $r_{i-1}=O(k^{\epsilon}\cdot (2k^{\epsilon})^{i-1})$. The output of the phase is a clustering $\mathcal{C}_{i}$ of radius $r_i$, and a hopset $H_i$ that takes care of all vertices that became unclustered in that phase. This output clustering is obtained by applying $t:=\lceil k^{\epsilon} \rceil/4$ steps of supercluster growing. As we will see, the clustering procedure will be very similar to Proc. $\ClusterAndAugment$ from Sec. \ref{sec:second-regime-spanner}, only that we add hops rather than shortest paths.

Starting with the trivial superclustering $\mathcal{SC}_{i-1,0}=\{\{C_j\} ~\mid~ C_j \in \mathcal{C}_{i-1}\}$ of radius $r_{i-1,0}=r_{i-1}$, in the $j^{th}$ step of phase $i$ for $j\geq 1$, the algorithm is given a superclustering $\mathcal{SC}_{i-1,j-1}$ of radius $r_{i-1,j-1}$. The algorithm then outputs a superclustering $\mathcal{SC}_{i-1,j}$ along with a hopset $H_{i,j}$ by taking the following steps.
\begin{enumerate}
\item Each unclustered vertex $v \in V\setminus {V(\mathcal{SC}_{i-1,j-1})}$ at center-distance at most $r_{i-1,j-1}+\alpha_{i-1,j}$ from centers of $\mathcal{SC}_{i-1,j-1}$ adds to $H_{i,j}$ a weighted hop to its closest center in $\mathcal{SC}_{i-1,j-1}$, where $\alpha_{i-1,0}=0$ and for $j\geq 1$,
\begin{equation}\label{eq:alpha}
\alpha_{i-1,j}=\frac{4\cdot r_{i-1,0}}{\lceil k^{\epsilon} \rceil-3}+\left(1+\frac{4}{\lceil k^{\epsilon} \rceil-3}\right)\cdot \alpha_{i-1,j-1} ~.
\end{equation}
\item Let $\mathcal{SC}' \subseteq \mathcal{SC}_{i-1,j-1}$ be the collection of superclusters obtained by 
sampling each supercluster $SC_\ell \in \mathcal{SC}_{i-1,j-1}$ independently with probability of $\frac{n_{0}}{n}$, where $n_{0}:=|\mathcal{C}_{i-1}|$. 
\item Each cluster $C \in SC_\ell$ at center-distance at most $r_{i-1,0}+r_{i-1,j-1}+2\cdot \alpha_{i-1,j}$ from $\mathcal{SC}'$ joins the supercluster of its closest center. All the vertices in $C$ add to $H_{i,j}$ a hop to the center of their new supercluster.
\item The superclustering $\mathcal{SC}_{i-1,j}$ consists of all sampled superclusters augmented by their nearby clusters (i.e., at center-distance $r_{i-1,0}+r_{i-1,j-1}+2\cdot \alpha_{i-1,j}$). The center of each new supercluster is the center of the sampled supercluster.
\item Each cluster $C \in SC_\ell$ at center-distance larger than $r_{i-1,0}+r_{i-1,j-1}+2\cdot \alpha_{i-1,j}$ from $\mathcal{SC}'$, adds to $H_{i,j}$
a weighted hop from its center to the center of any supercluster $SC\in \mathcal{SC}_{i-1,j-1}$ with $\cdist_{G}(C,SC)\le r_{i-1,0}+r_{i-1,j-1}+2\cdot \alpha_{i-1,j}$. 
\end{enumerate}
The parameter $\alpha_{i-1,j}$ is set in a way that guarantees that for every $(i-1,j)$-unclustered vertex $u$ and every $v\in \Ball_{G}(u,\alpha_{i-1,j})$, the hopset $H_{i,j}$ contains a $3$ hop path $u$-$v$ of length at most $k^{\epsilon}\cdot \alpha_{i-1,j}$. 

Let $\mathcal{SC}_{i-1,t}$ be the output superclustering after $t$ steps, then the output clustering $\mathcal{C}_i$ is formed by merging all clusters in the supercluster $SC_{\ell'}$ to a single cluster in $\mathcal{C}_{i}$ for every $SC_{\ell'}\in \mathcal{SC}_{i-1,t}$. That is, $\mathcal{C}_i=\{\{V(SC_{\ell'})\}  ~\mid~ SC_{\ell'} \in \mathcal{SC}_{i-1,t}\}$. Let $H_i=\bigcup_{j=1}^{t} H_{i,j}$. 
This completes the description of the $i^{th}$ phase. After $T$ phases, the output clustering $\mathcal{C}_{T}$ is shown to contain at most $n^{1-1/k^{\epsilon}}$ clusters in expectation. 
Let $H=\bigcup_{i=0}^{T} H_i$ be the current hopset after these phases.

\paragraph{Finalizing Stage: Spanner on Cluster Graph.} 
Given the collection of $O(n^{1-1/k^{\epsilon}})$ clusters in $\mathcal{C}_{T}$, the algorithm first adds a weighted hop from each unclustered vertex $v \notin V(\mathcal{C}_{T})$ to the center of its closest cluster in $\mathcal{C}_{T}$ up to center-distance $r_{T}+d$, if such exists. Next, a cluster graph $\widehat{G}=(\mathcal{C}_{T},\mathcal{E})$ is defined by 
connecting two clusters $C,C' \in \mathcal{C}_{T}$ if their center-distance in $G$ is at most $2r_{T}+2d$.
That is, $\mathcal{E}=\{(C,C') ~\mid~ C,C' \in \mathcal{C}_{T} \mbox{~and~} \dist_G(C,C')\leq 2r_{T}+2d\}$. Let $\widehat{H}$ be a $k'$-spanner of $\widehat{G}$ for $k'= 2\cdot \lceil k^{\epsilon} \rceil -3$. For edge $(C,C') \in \widehat{H}$, a weighted hop between the center of $C$ and the center of $C'$ is added to the hopset $H$. 
This completes the description of the algorithm. 
\begin{algorithm}[h!]
	\begin{algorithmic}[1]
		\caption{$\ImprovedHopsetII(G,k,d,\epsilon)$}
		\State \textbf{Input:} A weighted graph $G=(V,E,w)$, integers $k,d$ and $0<\epsilon<\frac{1}{2}$. 
		\State \textbf{Output:} A hopset $H\subseteq{G}$ of expected size $O(k^{\epsilon}\cdot n^{1+1/k}+k^{\epsilon}/\epsilon\cdot n)$, such that for any pair $u,v$ at distance $d'\in [1/2\cdot d,d]$ in $G$, $\dist_{G\cup H}^{(36^{1/\epsilon} \cdot k^{1-\epsilon})}(u,v)\le 4.5 \cdot k^{\epsilon}\cdot d$.		
		\State Let $T:=\log_{\lceil k^{\epsilon} \rceil /4 }(k^{1-2\epsilon}),t:=\lceil k^{\epsilon} \rceil /4 $.
		\State Set $R':=1/2\cdot 36^{\frac{1}{\epsilon}}\cdot k^{1-2\epsilon}$ and $r_0=d/R'$.
		\State Let $(\mathcal{C}_{0},H_{0}) \leftarrow \ClusterHop(G,k,\lceil k^{\epsilon} \rceil,r_0)$.
		\For{$i=1$ to $T$}
		\State $(\mathcal{C}_{i},H_{i})\leftarrow \ClusterAndAugmentHop(G,k,\epsilon,\mathcal{C}_{i-1})$.
		\State $H\leftarrow H\cup H_{i}$.
		\EndFor 
		\State Let $\mathcal{C}=\mathcal{C}_{T}$.
		\State For each $v\in V\setminus V(\mathcal{C})$, if $\cdist_{G}(v,\mathcal{C})\le r_{T}+d$, add a hop from $v$ to the center of its closest cluster.
		\State Let $\widehat{G}=(\mathcal{C},\mathcal{E}=\{(C,C')~|~ \cdist_{G}(C,C')\le 2\cdot  r_{T}+2\cdot d\})$.
		\State Let $\widehat{H}=\BasicSpanner(\widehat{G},2\cdot \lceil k^{\epsilon} \rceil-3)$. 
		\State For each edge $(C,C')\in E(\widehat{H})$ add to $H$ a weighted hop between the centers of $C,C'$.
		\State \Return $H$
	\end{algorithmic}
\end{algorithm}
\begin{algorithm}[H]
	\begin{algorithmic}[1]
		\caption{$\ClusterAndAugmentHop(G,k,\epsilon,\mathcal{C})$}
		\label{alg:ClusterAndAugmentHop}
		\State \textbf{Input:} A weighted graph $G=(V,E,w)$, integer $k$, stretch parameter $0<\epsilon <1$ and a clustering $\mathcal{C}$.
		\State \textbf{Output:} A clustering $\mathcal{C'}$ and a hopset $H$.
		\State Let $H\leftarrow \emptyset$, $r_{0} = rad_{G}(\mathcal{C})$, $\alpha_{0}= 0, n_{0}=|\mathcal{C}|$ and $t=\lceil k^{\epsilon} \rceil/4$.
		\State Set $\mathcal{SC}_{0}\leftarrow \left\{\{C\}~\mid~C\in \mathcal{C}\right\}$, where for each $C\in \mathcal{C}$, $\{C\}$ is a supercluster containing all vertices in $C$, with $C$'s center as the center of the supercluster.
		\For {$i=1$ to $t$}
		\State Set $\alpha_{i}\leftarrow\frac{4\cdot r_{0}}{4\cdot t-3}+(1+\frac{4}{4\cdot t-3})\cdot\alpha_{i-1}$.
		\State For each $v\in V\setminus{\mathcal{SC}_{i-1}}$, at center-distance at most $r_{i-1}+\alpha_{i}$ from the superclustering $\mathcal{SC}_{i-1}$, add to $H$ a hop from $v$ to the center of its closest supercluster.
		\State $\mathcal{SC'}\leftarrow Sample(\mathcal{SC}_{i-1},\min\{n^{-1/k},\frac{n_{0}}{n}\})$
		\For{every $SC\in \mathcal{SC}_{i-1}$}
		\For{every $C\in SC$}
		\State If $C$ has a sampled supercluster in $\mathcal{SC'}$ at center-distance at most $r_{0}+r_{i-1}+2\cdot \alpha_{i}$ from it, add $C$ to its closest supercluster in $\mathcal{SC'}$. Add to $H$ a hop from each vertex in $V(C)$ to the center of that supercluster.
		\State Otherwise, add a hop from the center of $C$ to the center of any supercluster in $\mathcal{SC}_{i-1}$ at center-distance at most $r_{0}+r_{i-1}+2\cdot \alpha_{i}$ from $C$.
		\EndFor
		\EndFor
		\EndFor
		\State Let $\mathcal{C'}=\{\{V(SC_j)\} ~\mid~ SC_j \in \mathcal{SC}_t\}$.
		\State \Return $(\mathcal{C'},H)$.
	\end{algorithmic}
\end{algorithm}
The analysis of this algorithm is very similar to the analysis of the spanner construction from Section \ref{sec:second-regime-spanner}. 
\paragraph{Stretch Analysis.}
Recall that $r_{0}:=\frac{d}{R'}$, $T= \log_{\lceil k^{\epsilon} \rceil/4}(k^{1-2\epsilon}),t=\lceil k^{\epsilon}\rceil/4$. For $1\le i \le T$ let $r_{i,0}=rad(\mathcal{C}_{i})$ and let $r_{T}:=rad(\mathcal{C}_T)$ be the radius of the clusters in the last clustering $\mathcal{C}_T$ at the end of the middle stage. The stretch and size arguments are simil
For the sake of the stretch and size analysis, we will need the following two claims, bounding $\alpha_{i,j}$ and $r_{i,0}$, respectively. The next claim is the analog of Claim \ref{claim:alpha} in Section \ref{sec:second-regime-spanner}:
\begin{claim}
\label{claim:alphahop}
	For $i\in [T],j\in [t]$, $\alpha_{i,j}\le r_{i,0}\cdot\left((1+\frac{4}{\lceil k^{\epsilon} \rceil-3})^{j}-1\right)$.
\end{claim}
\begin{proof}
	We show by induction on $j$ that
	\begin{equation}
	\label{eq:indalphahop}
	\alpha_{i,j}\le  \frac{4\cdot r_{i,0}}{\lceil k^{\epsilon} \rceil-3}\cdot\sum_{p=0}^{j-1} \left(1+\frac{4}{\lceil k^{\epsilon} \rceil-3}\right)^{p}.
	\end{equation} 
	The base case, $j=1$ is trivial since $\alpha_{i,1} = \frac{4\cdot r_{i,0}}{\lceil k^{\epsilon} \rceil-3}$. By Eq. (\ref{eq:alpha}) and the induction assumption,
	\begin{eqnarray*}
		\alpha_{i,j}&=& \frac{4\cdot r_{i,0}}{\lceil k^{\epsilon} \rceil-3}+\left(1+\frac{4}{\lceil k^{\epsilon} \rceil-3}\right)\cdot\alpha_{i,j-1}\\&\le& \frac{4\cdot r_{i,0}}{\lceil k^{\epsilon} \rceil-3}+\left(1+\frac{4}{\lceil k^{\epsilon} \rceil-3}\right)\cdot \frac{4\cdot r_{i,0}}{\lceil k^{\epsilon} \rceil-3} \cdot\sum_{p=0}^{j-2}\left(1+\frac{4}{\lceil k^{\epsilon} \rceil-3}\right)^{p}
		\\&=& \frac{4\cdot r_{i,0}}{\lceil k^{\epsilon} \rceil-3}\cdot\sum_{p=0}^{j-1} \left(1+\frac{4}{\lceil k^{\epsilon} \rceil-3}\right)^{p}.
	\end{eqnarray*}
	By Eq. (\ref{eq:indalphahop}) we then have:
	\begin{eqnarray*}
		\alpha_{i,j}&\le& \frac{4\cdot r_{i,0}}{\lceil k^{\epsilon} \rceil-3}\cdot\sum_{p=0}^{j-1} \left(1+\frac{4}{\lceil k^{\epsilon} \rceil-3}\right)^{p}=\frac{4\cdot r_{i,0}}{\lceil k^{\epsilon} \rceil-3}\cdot \frac{\lceil k^{\epsilon} \rceil-3}{4} \cdot\left((1+\frac{4}{\lceil k^{\epsilon} \rceil-3})^{j}-1\right)
		\\&=& r_{i,0}\cdot\left((1+\frac{4}{\lceil k^{\epsilon} \rceil-3})^{j}-1\right).
	\end{eqnarray*}
\end{proof}
We next turn to bound the radius of the clustering $\mathcal{C}_{i}$ for every $i\in [T]$. The next claim is the analog of Claim \ref{claim:radiusecondspanner} in Section \ref{sec:second-regime-spanner}:
\begin{claim}
	\label{claim:radiushopset}
	For each $0\le i \le T$ it holds that $r_{i,0} \le (1.5\cdot \lceil k^{\epsilon} \rceil)^{i}\cdot r_{0}$, thus $r_{T}\le d/648$.
\end{claim}
\begin{proof}
	We show the correctness of the claim by induction on $i$. The base case, for $i=0$ follows by the definition of $r_{i,0}$. Assuming the claim holds up to $i$ we show the correctness for $i+1$. Phase $i+1$ begins with clusters with radius at most $r_{i,0}$, and terminates after $t=\lceil k^{\epsilon} \rceil /4$ superclustering steps of Procedure $\ClusterAndAugmentHop$ with clusters with radius $r_{i+1,0}$. We therefore bound $r_{i+1,0}$. At step $j$ of Proc. $\ClusterAndAugmentHop$, we have superclusters of radius $r_{i,j-1}$. The algorithm then connects clusters of radius $r_{i,0}$, at distance at most $2\cdot \alpha_{i,j}$ from the sampled supercluster. In the $(i,j)$ step, the radius of the new supercluster is then increased by an additive factor of $2\cdot r_{i,0}+2\cdot \alpha_{i,j}$. Combining this with the bound on $\alpha_{i,j}$ of Claim \ref{claim:alpha}, we have:
\begin{eqnarray}\label{eq:rijhop}
r_{i,j} &=& r_{i,j-1}+2\cdot r_{i,0}+2\cdot \alpha_{i,j}=(2\cdot j+1)\cdot r_{i,0}+2\cdot \sum_{p=1}^{j}\alpha_{i,p}
\\&\le& (2\cdot j+1)\cdot r_{i,0}+2\cdot r_{i,0}\cdot\sum_{p=1}^{j}\left(\left(1+\frac{4}{\lceil k^{\epsilon} \rceil-3}\right)^{p}-1 \right) \nonumber
\\&=& r_{i,0}+2\cdot r_{i,0}\cdot \left(1+\frac{4}{\lceil k^{\epsilon} \rceil-3}\right)\cdot \frac{\lceil k^{\epsilon} \rceil-3}{4}\cdot  \left(\left(1+\frac{4}{\lceil k^{\epsilon} \rceil-3} \right)^{j}-1\right) \nonumber
\\&=&\left(1+\frac{\lceil k^{\epsilon} \rceil+1}{2} \cdot\left(\left(1+\frac{4}{\lceil k^{\epsilon} \rceil-3}\right)^{j}-1\right)\right)\cdot r_{i,0} \nonumber
\end{eqnarray}
Thus plugging $t=\lceil k^{\epsilon} \rceil/4$:
\begin{eqnarray*}
r_{i+1,0}&:=&r_{i,t}\le \left(1+\frac{\lceil k^{\epsilon} \rceil+1}{2} \cdot \left(\left(1+\frac{4}{\lceil k^{\epsilon} \rceil-3}\right)^{\lceil k^{\epsilon} \rceil /4}-1\right)\right) \cdot r_{i,0}
\\&\le&\left(1+\frac{\lceil k^{\epsilon} \rceil+1}{2} \cdot \left(e^{1+\frac{3}{\lceil k^{\epsilon} \rceil-3}}-1\right)\right)\cdot r_{i,0}\le 1.5\cdot \lceil k^{\epsilon} \rceil\cdot  r_{i,0}
\end{eqnarray*}
The third inequality is valid when $k\ge 16^{1/\epsilon}$. Finally, the final radius $r_{T}$ is bounded by:
\begin{eqnarray*}
	r_{T}&\le& r_{0}\cdot (1.5\cdot \lceil k^{\epsilon} \rceil)^{\log_{\lceil k^{\epsilon} \rceil /4 }(k^{1-2\epsilon})}\le \frac{d}{1/2\cdot 36^{\frac{1}{\epsilon}}\cdot k^{1-2\epsilon}}\cdot (6\cdot (\lceil k^{\epsilon} \rceil /4))^{\log_{\lceil k^{\epsilon} \rceil /4 }(k^{1-2\epsilon})}
	\\&\le&  \frac{2\cdot d}{36^{\frac{1}{\epsilon}}\cdot k^{1-2\epsilon}}\cdot 6^{\frac{(1-2\epsilon)\cdot \log{k}}{\epsilon \cdot \log{k}-2}}\cdot k^{1-2\epsilon}\le\frac{2\cdot d}{36^{\frac{1}{\epsilon}}}\cdot 36^{\frac{1-2\epsilon}{\epsilon}}= d/648~,
\end{eqnarray*}
where the fourth inequality holds for $k\ge 16^{1/\epsilon}$.
\end{proof}

\begin{definition}[Clustered and Unlcustered Vertices]
\emph{A vertex $u\in V$ is} $0$-unclustered \emph{if $v \notin V(\mathcal{C}_{0})$. A vertex $v$ is} $(i,j)$-unclustered \emph{if $v\in V(\mathcal{C}_{i,j-1})\setminus{V(\mathcal{C}_{i,j})}$. A vertex $v$ is \emph{clustered} if $v \in V(\mathcal{C})$}.
\end{definition}

By Claim \ref{cl:helper-sparse-hopset} we have that $\dist_{G\cup H}^{(2)}(u,v)\le (2\cdot \lceil k^{\epsilon} \rceil +1)\cdot \dist_{G}(u,v)$ for every $0$-unclustered vertex $u$ and any $v\in \Ball(u,r_{0}/(\lceil k^{\epsilon} \rceil +1))$. We now consider the remaining vertices.
\begin{claim}
\label{cl:(i,j)-unc-hop}  
For $0\le i \le T-1,1\le j \le t$, for each $(i,j)$-unclustered vertex $u\in V$, $\dist_{H}^{(3)}(u,v) \le \lceil k^{\epsilon} \rceil \cdot \alpha_{i,j}$ for any $v\in \Ball_{G}(u,\alpha_{i,j})$.
\end{claim}
\begin{proof}
Let $v\in \Ball_{G}(u,\alpha_{i,j})$. In the beginning of the $(i+1,j)$ step, the algorithm adds to $H$ hops from any vertex at center-distance at most $r_{i,j-1}+\alpha_{i,j}$ from $\mathcal{SC}_{i,j-1}$ to the center of its closest supercluster. In particular, since $\cdist_{G}(v,\mathcal{SC}_{i,j-1})\le  \dist_{G}(v,u)+r_{i,j-1} \le r_{i,j-1}+\alpha_{i,j}$, we add a hop from $v$ to the center of some supercluster $SC_{v}\in \mathcal{SC}_{i,j-1}$, and the weight of this hop is at most $r_{i,j-1}+\alpha_{i,j}$. Since $u$ is unclustered in this step, the algorithm adds to $H$ a hop from the center of its cluster $C_{u}$ to the center of any supercluster which is in $\mathcal{SC}_{i,j-1}$ and is at center-distance at most $r_{i,0}+r_{i,j-1}+2\cdot \alpha_{i,j}$ from $C_{u}$. Since $\cdist_{G}(C_{u},SC_{v})\le r_{i,0}+\dist_{G}(u,v)+\cdist_{G}(v,SC_{v}) \le r_{i,0}+r_{i,j-1}+2\cdot \alpha_{i,j}$, thus the algorithm adds a hop from $z_{u}$ (the center of $C_{u}$) to $z_{v}$ (the center of the supercluster $SC_{v}$) of weight at most 
$$\dist_{G}(z_{u},u)+\dist_{G}(u,v)+\dist_{G}(v,z_{v})\le r_{i,0}+r_{i,j-1}+2\cdot \alpha_{i,j},$$ and consequently, there is a $3$-hop from $u$ to $v$ through $z_{u}$ and $z_{v}$ of length at most:
\begin{eqnarray*}
\dist_{H}^{(3)}(u,v)&\le& r_{i,0}+(r_{i,0}+r_{i,j-1}+2\cdot \alpha_{i,j})+(r_{i,j-1}+\alpha_{i,j})
\\&=&3\cdot \alpha_{i,j}+2\cdot r_{i,0}+2\cdot r_{i,j-1}
\leq 3\cdot \alpha_{i,j}+2\cdot r_{i,0}+2\cdot ((2 \cdot j-1)\cdot r_{i,0}+2\cdot \sum_{p=1}^{j-1}\alpha_{i,p})
\\&\le&3\cdot \alpha_{i,j}+4 \cdot j\cdot r_{i,0}+4\cdot \sum_{p=1}^{j-1}\alpha_{i,p}~,
\end{eqnarray*}
where the second inequality follows from Eq. (\ref{eq:rijhop}).
We show by induction on $j$ that $(\lceil k^{\epsilon} \rceil-3)\alpha_{i,j} \ge 4 \cdot j\cdot r_{i,0}+4\cdot \sum_{p=1}^{j-1}\alpha_{i,p}.$ The base case holds trivially. Assuming the correctness of the claim up to $j-1$ we have:
\begin{eqnarray*}
\alpha_{i,j} &=& \frac{4\cdot r_{i,0}}{\lceil k^{\epsilon} \rceil-3}+(1+\frac{4}{\lceil k^{\epsilon} \rceil-3})\cdot\alpha_{i,j-1}\ge \frac{4 \cdot j\cdot r_{i,0}+4\cdot \sum_{p=1}^{j-2}\alpha_{i,p}}{\lceil k^{\epsilon} \rceil-3}+\frac{4}{\lceil k^{\epsilon} \rceil-3}\cdot \alpha_{i,j-1} 
\\&=& \frac{4 \cdot j\cdot r_{i,0}+4\cdot \sum_{p=1}^{j-1}\alpha_{i,p}}{\lceil k^{\epsilon} \rceil-3}~.
\end{eqnarray*}
We therefore have that,
\begin{eqnarray*}
\dist_{H}^{(3)}(u,v)&\le&3\cdot \alpha_{i,j}+4 \cdot j\cdot r_{i,0}+4\cdot \sum_{p=1}^{j-1}\alpha_{i,p}
\leq 3\cdot \alpha_{i,j}+(\lceil k^{\epsilon} \rceil-3)\cdot \alpha_{i,j}=\lceil k^{\epsilon} \rceil \cdot \alpha_{i,j}.
\end{eqnarray*}
Figure \ref{fig:SecondSpannerAnalysis}, though illustrating claims regarding spanners, can be used to illustrate the proofs of Claims \ref{cl:(i,j)-unc-hop} and \ref{cl:clustered-hop-second} as well. 
\end{proof}
\begin{claim}
\label{cl:clustered-hop-second}
Let $u,v\in V$ be vertices at distance $d$ in $G$, and let $P$ be a shortest path between them. If there is some clustered vertex $w\in P$, then $\dist_{G\cup H}^{(2\lceil k^{\epsilon}\rceil-1)}(u,v) \le 4.5\cdot k^{\epsilon} \cdot d$.
\end{claim}
\begin{proof}
Let $C_{w}$ be the cluster to which $w$ belongs. In the beginning of the third stage of the algorithm we add hops from unclustered vertices at center-distance at most $r_{T}+d$ to their closest cluster center. Thus since $\cdist_{G}(u,\mathcal{C})\le \dist_{G}(u,w)+r_{T}\le r_{T}+d$ and $\cdist_{G}(v,\mathcal{C})\le \dist_{G}(v,w)+r_{T}\le r_{T}+d$, it holds that both $u,v$ have hops to centers of some clusters $C_{u},C_{v}\in \mathcal{C}$. Since 
\begin{eqnarray*}
\cdist_{G}(C_{u},C_{v})&\le& \cdist_{G}(C_{u},u)+\dist_{G}(u,v)+\cdist_{G}(v,C_{v})
\\&\le& 2r_{T}+\dist_{G}(w,u)+d+\dist_{G}(v,w) \le 2r_{T}+2d,
\end{eqnarray*}
it holds that $\dist_{\widehat{G}}(C_{u},C_{v})\le 1$, thus $\dist_{\widehat{H}}(C_{u},C_{v})\le 2\cdot  \lceil k^{\epsilon}\rceil-3$. Since each edge $(C,C')\in E(\widehat{H})$ translates into a hop in $H$ between the centers of $C_{u},C_{v}$ of weight at most $2\cdot r_{T}+2d$, we have the following,
\begin{eqnarray*}
\dist_{G\cup H}^{(2\cdot \lceil k^{\epsilon}\rceil-1)}(u,v) &\le&  \cdist_{G}(u,C_{u})+\dist_{\widehat{H}}(C_{u},C_{v})\cdot (2\cdot r_{T}+2\cdot d)+\cdist_{G}(C_{v},v)
\\&\le& \dist_{G}(u,w)+r_{T}+(2\cdot \lceil k^{\epsilon} \rceil -3)\cdot (2\cdot r_{T}+2\cdot d)+\dist_{G}(v,w)+r_{T}
\\&\le& (2\cdot k^{\epsilon} -1)\cdot (2\cdot r_{T}+2\cdot d)+d+2\cdot r_{T}
\\&\le& 4\cdot k^{\epsilon} \cdot d+4\cdot k^{\epsilon}  \cdot r_{T}\le 4.5\cdot k^{\epsilon} \cdot d,
\end{eqnarray*}
where the last inequality follows by plugging the bound on the final radius $r_{T}$ from Claim \ref{claim:radiushopset}. See Figure \ref{fig:SecondHopsetAnalysis}(b) for an illustration.
\end{proof}

\begin{proof}[Proof of Theorem  \ref{thm:secondhopset}]
\textbf{Stretch and Hop-Bound.}
Let $u,v\in V$ be vertices at distance $d'\in [d/2,d]$ in $G$, and let $P$ be some shortest path between them in $G$. First observe that if there is some clustered vertex $w\in P$ the claim follows from Claim \ref{cl:clustered-hop-second}. Assume from now on that there is no such vertex. Partition the path $P$ into consecutive segments in the following way: denote $v_{0}:=u$ and inductively, given $v_{l}\neq v$ define $v'_{l}$ to be the furthest vertex on $P[v_{l},v]$ at distance at most $\Delta_{l}$ from $v_{l}$, where:
\begin{equation*}
\Delta_{l}=\begin{cases}
\min\{\frac{r_0}{(\left \lceil k^{\epsilon} \right \rceil + 1)},\dist(v_{l},v)\} & v_{l}\text{ is \ensuremath{0-unclustered}}\\
\min\{\alpha_{i,j},\dist(v_{l},v)\} & v_{l}\text{ is \ensuremath{(i,j)-unclustered}}\\
\end{cases}
\end{equation*}
Note that $v'_{l}$ might be equal to $v_{l}$ if the incident edge to $v_{l}$ on $P[v_{l},v]$ has weight larger than $\Delta_{l}$. In addition, for each $v'_{l}\neq v$ let $v_{l+1}$ be the consecutive neighbor of $v'_{l}$ on $P[v'_{l},v]$. When $v'_{l}=v$, simply let $v_{l+1}=v'_{l}$. Let $\ell$ be the minimal index such that $v_{\ell}=v$. This defines a partition of $P$ to $\ell$ segments by setting for all $1\le i\le\ell$, $P_{i}=[v_{i-1},v_{i}]$ and $P_{\ell}=[v_{\ell-1},v]$ is the last segment that reaches $v$. Note that for every segment $P_{l+1}$ (except at most the last one) $\dist_{G}(v_{l},v_{l+1})\ge \frac{r_0}{(\left \lceil k^{\epsilon} \right \rceil + 1)}$. If $v_{l}$ is $0$-unclustered this is clear, and if $v_{l}$ is $(i,j)$-unclustered, then by Eq. (\ref{eq:alpha}) it holds that $\dist_{G}(v_{l},v_{l+1})\ge \alpha_{i,j}\ge \frac{r_{i-1,0}}{\lceil k^{\epsilon}\rceil-3}>\frac{r_0}{(\left \lceil k^{\epsilon} \right \rceil + 1)}$. Thus we have that $\ell\le (\left \lceil k^{\epsilon} \right \rceil + 1)\cdot R'<36^{1/\epsilon} \cdot k^{1-\epsilon}$. 
For any $l \in \{0,\ldots, \ell-1\}$, if $v_{l}$ is $0$-unclustered, by Claim \ref{cl:helper-sparse-hopset}, it holds that: 
\begin{eqnarray*}
\dist^{(3)}_{G \cup H}(v_{l},v_{l+1})\leq \dist^{(2)}_{G \cup H}(v_{l},v'_{l})+\dist_{G}(v'_{l},v_{l+1})\le (2\cdot \lceil k^{\epsilon} \rceil+1)\cdot \dist_{G}(v_{l},v_{l+1})~.
\end{eqnarray*} 
If $v_{l}$ is $(i,j)$-unclustered then since $\dist_{G}(v_{l},v'_{l}) \le \alpha_{i,j}$, by Claim \ref{cl:(i,j)-unc-hop}, it holds that $\dist^{(3)}_{G \cup H}(v_{l},v'_{l})\leq \lceil k^{\epsilon} \rceil \cdot \alpha_{i,j}$. There are two cases to consider. First assume that $\dist_{G}(v_{l},v_{l+1})\ge \alpha_{i,j}$. In such a case $\dist^{(4)}_{G \cup H}(v_{l},v_{l+1})\leq \lceil k^{\epsilon}\rceil\cdot \dist_{G}(v_{l},v_{l+1})$. Next, assume that $\dist_{G}(v_{l},v_{l+1})< \alpha_{i,j}$. By the definition of the segment, in such a case it must hold that $v'_{l}=v_{l+1}=v_{\ell}=v$. Thus by proof of Claim \ref{cl:(i,j)-unc-hop} and Eq. (\ref{eq:rijhop}), $$\dist_{G\cup H}^{(3)}(v_{l},v)\le 3\cdot \alpha_{i,j}+2\cdot r_{i,0}+2\cdot r_{i,j-1}\le 2\cdot r_{i,j}\le 2\cdot r_{T}\le d/300.$$ Therefore by summing over all $\ell$ segments, we get that 
	$$\dist^{(4\cdot \ell)}_{G \cup H}(u,v)\leq d/300+\sum_{i=0}^{\ell-1}\dist_{G\cup H}^{(4)}(v_{l},v_{l+1})\le (2\cdot \lceil k^{\epsilon} \rceil+2)\cdot \dist_{G}(u,v)~.$$
\paragraph{Size Analysis.} 
By Fact \ref{fact:TZ}(iii), it holds that $|E(H_{0})|=O(k^{\epsilon}\cdot n^{1+1/k})$. In the same manner as shown in  Claim \ref{cl:expectedspinstep} it holds that for any $1\le i \le T,1\le j \le \lceil k^{\epsilon}\rceil/4$, in step $(i,j)$ we add $O(n)$ hops. It follows that for all $1\le i \le T$, in expectation $|H_{i}|=O(k^{\epsilon}\cdot n)$. Since there are $T=\log_{\lceil k^{\epsilon} \rceil /4 }k^{1-2\epsilon}=O(\frac{1}{\epsilon})$ phases we have $\sum_{i=1}^{T}|E(H_{i})|=O(\frac{k^{\epsilon}}{\epsilon}\cdot n)$ hops in total. As shown in Claim \ref{cl:numberofclusters}, $|\mathcal{C}_T|=O(n^{1-(\lceil k^{\epsilon} \rceil /4)^{T}\cdot \frac{\left \lceil k^{\epsilon} \right \rceil}{k}})=O(n^{1-\frac{1}{k^{\epsilon}}})$ in expectation. Consequently,
$|E(\widehat{H})|=O(|\mathcal{C}_{T}|^{1+\frac{1}{\lceil k^{\epsilon} \rceil -1}})=O(n)$ in expectation. 
Since each edge in $\widehat{H}$ translates into a single hop, this step contributes $O(n)$ hops. 
Overall, $|E(H)| = O(k^{\epsilon}\cdot n^{1+1/k}+\frac{k^{\epsilon}}{\epsilon}\cdot n)$, in expectation.
\end{proof}
Theorem \ref{thm:secondhopset} follows by applying the algorithm for each of the $\log \Lambda$ distance classes.

\subsection{New $(3+\epsilon,\beta)$ Hopset}
\label{sec:3epshopset}
In this subsection we show a considerably simplified construction of $(3+\epsilon,\beta)$ hopsets. For example for $\epsilon=1$, we get a $(4,O(k^{\log{23}}))$ hopset. For the sake of the efficient implementation in Sec. \ref{sec:eff-hopsets}, we settle for a slightly worse value of $\beta$. In Appendix \ref{sec:best-hopset}, we show an improved construction that achieves the bounds of Lemma \ref{lem:hopset-best}. Our main result in this section is as follows:
	
\begin{lemma}\label{lem:3plusepshopset}
For any $n$-vertex weighted graph $G=(V,E,w)$, integer $k$ and $\epsilon>0$, one can compute a $(3+\epsilon,\beta)$ hopset $H$ where $\beta=16\cdot (5+ 18/\epsilon)\cdot k^{\log(5+18/\epsilon)}$ of expected size $|E(H)|=O((n^{1+1/k}+\log{k}\cdot n)\log \Lambda)$.
\end{lemma}
	
\paragraph{Algorithm Description.}	
For simplicity we fix a distance range $[d,2d]$. Lemma \ref{lem:3plusepshopset} follows by taking care of all $\log \Lambda$ ranges. Furthermore, for simplicity we assume throughout that $4/\epsilon$ is an integer. The algorithm has two stages. In the \textbf{first stage} it calls Procedure $\ClusterHop$ (from Sec. \ref{sec:first-regime-hopset}) for a \emph{single} iteration with a radius parameter 
$$r_{0}=d/(2\cdot R'), \mbox{~~where~~} R'=(5+16/\epsilon)^{\lceil \log{k} \rceil}~.$$
For completeness, and due to its simplicity we add a complete description of this single iteration. The procedure samples each vertex $v\in V$ into a subset $A_{1}\subseteq{V}$ with probability $n^{-1/k}$. For each $v\in V$ let $p(v)$ its closest vertex in $A_{1}$. For each sampled vertex $a\in A_{1}$ define its cluster $$C(a)=\{u ~|~p(u)=a, ~\dist_{G}(a,u)\le r_{0}\}~.$$

The $0$-level clustering is given by $\mathcal{C}_{0}=\{C(a)~ \mid ~a\in A_{1}\}$. Finally, add to the hopset $H_0$, the hops $(v,p(v))$ for every $v\in V$. In addition, for every unclustered vertex $v\in V\setminus V(\mathcal{C}_{0})$ add to $H_{0}$, the hop to each vertex $u$ satisfying that $\dist_{G}(v,u)\le \dist_{G}(v,p(v))$. The procedure outputs the clustering $\mathcal{C}_{0}$ with $n^{1-1/k}$ clusters, and a partial hopset $H_{0}$ of size $n^{1+1/k}$, in expectation. 

The \textbf{second stage} has $T=\lceil \log k\rceil$ clustering phases. Starting with $\mathcal{C}_{0}$, in each phase $i\geq 1$, given is a clustering $\mathcal{C}_{i-1}$ of expected size $n_{i-1}=n^{1-\frac{2^{i-1}}{k}}$ and with radius at most 
$$r_{i-1}\le r_0\cdot (5+ \frac{16}{\epsilon})^{i-2}~.$$ 
The output of the $i^{th}$ phase is a clustering $\mathcal{C}_i$ of expected size $n_{i}=n^{1-\frac{2^{i}}{k}}$, and a hopset $H_i$ that takes care of the unclustered vertices in $V(\mathcal{C}_{i-1})\setminus V(\mathcal{C}_{i})$. 
	
	We now zoom into $i^{th}$ phase and explain the construction of the clustering $\mathcal{C}_i$ and the hopset $H_i$. The phase is governed by two key parameters: the sampling probability of each cluster $p_i$ and the augmentation radius $\alpha_i$.
Similarly to the algorithm of Section \ref{sec:3epsspanner}, for every $i\geq 1$, define
\begin{equation}\label{eq:alpha-p}
\alpha_i = \frac{4}{\epsilon}\cdot r_{i-1} \mbox{~~and~~} p_i=|\mathcal{C}_{i-1}|/n~.
\end{equation}
The description of the $i^{th}$ phase for every $i \in \{1,\ldots, T\}$ is as follows:
\begin{enumerate}
\item Each unclustered vertex $v \in V\setminus {V(\mathcal{C}_{i-1})}$ with $\cdist_{G}(v,\mathcal{C}_{i-1})\le r_{i-1}+\alpha_{i}$ adds to $H_{i}$ a hop to its closest cluster center in $\mathcal{C}_{i-1}$.

\item Let $\mathcal{C}' \subseteq \mathcal{C}_{i-1}$ be the collection of clusters obtained by 
sampling each cluster $C_\ell \in \mathcal{C}_{i-1}$ independently with probability of $p_i$.

\item Each cluster $C \in \mathcal{C}_{i-1}$ such that $\cdist_{G}(C,\mathcal{C}')\le 4\cdot r_{i-1}+4\cdot \alpha_{i}$, joins the its closest cluster in $\mathcal{C}'$ (based on the $\cdist$ measure).
All the vertices in $C$ add to $H_i$, a hop to the center of their new cluster.

\item The clustering $\mathcal{C}_{i}$ consists of all sampled clusters in $\mathcal{C}'$ augmented by their nearby clusters in $\mathcal{C}_{i-1}$ (i.e., at center-distance at most $4\cdot r_{i-1}+4\cdot \alpha_{i}$).

\item The center of each cluster $C \in \mathcal{C}_{i-1}$ such that $\cdist_{G}(C,\mathcal{C}')> 4\cdot r_{i-1}+4\cdot \alpha_{i}$, adds to $H_{i}$ a hop to the center of any cluster $C'\in\mathcal{C}_{i-1}$ with $\cdist_{G}(C,C')\le 2\cdot r_{i-1}+2\cdot \alpha_{i}$. 
\end{enumerate}
This completes the description of the $i^{th}$ phase. Let $H=\bigcup_{i=0}^{T} H_i$ and let $\mathcal{C}_T$ be the output clustering of the last phase $T$. In the analysis section we show that $\mathcal{C}_T$ contains at most a one cluster $C$, in expectation . We then add the output hopset $H$, a hop from each vertex at center-distance at most $r_{T}+2d$ from $C$ to the center of the cluster\footnote{For the sake of the efficient implementation of Sec. \ref{sec:eff-hopsets}, we connect only vertices up to center-distance $r_{T}+2d$ to the center of the last cluster, even-though with respect to the size of the hopset, we can afford ourselves to connect the center of $C$ to all nodes.}.
	
\paragraph{Stretch Analysis.}
For the sake of the stretch analysis, we use the following definition:
\begin{definition}
A vertex $v\in V$ is $0$-unclustered if $v\notin V(\mathcal{C}_{0})$. For any $i\ge 1$, a vertex $v\in V$ is $i$-unclustered if it is in $V(\mathcal{C}_{i-1})\setminus{V(\mathcal{C}_{i})}$. Finally, a vertex $v\in V$ is clustered if it is in the last cluster.
\end{definition}
We first make a simple observation:
\begin{observation}
\label{obs:rad3epshopset}
For every $0 \in \{1,\ldots, T\}$, $r_i \leq (5+ 16/\epsilon)^{i}\cdot r_{0}$. In particular, the radius of cluster in the final clustering $\mathcal{C}_T$ is $r_T \le d/2$.
\end{observation}
\begin{proof}
The claim is shown by induction on $i$. The base case $i=0$ is trivial. Assume that the claim holds up to $i-1 \geq 0$, and consider the $i^{th}$ where the clustering $\mathcal{C}_i$ is defined based on the clustering $\mathcal{C}_{i-1}$. Each cluster in  $\mathcal{C}_i$ is formed by a star: the head of the star is the sampled cluster $C$, connected to other clusters $C' \in \mathcal{C}_{i-1}$ with center-distance at most $4\cdot r_{i-1}+4\alpha_{i}$. The radius of this star of clusters is bounded by:
$$r_i \leq r_{i-1}+4\cdot r_{i-1}+4\alpha_{i}= 5\cdot r_{i-1}+16/\epsilon \cdot r_{i-1} \leq \left (5+16/\epsilon \right)^{i}\cdot r_{0}~,$$
where the second inequality follows by plugging the bound on $r_{i-1}$ obtained from the induction assumption, and the bound on $\alpha_i$ from Eq. (\ref{eq:alpha-p}).
\end{proof}
	
\begin{observation}
\label{obs:clusters}
For $0\le i \le \lceil \log{k}\rceil$, in expectation $|\mathcal{C}_{i}|=n^{1-\frac{2^{i}}{k}}$ therefore after $T=\lceil \log{k}\rceil$ phases, there is at most one cluster in $\mathcal{C}_T$, in expectation.
\end{observation}
\begin{proof}
By induction on $i$. For the base case consider $i=0$. In the first stage the algorithm samples each vertex with probability $n^{-1/k}$, thus we have $n^{1-1/k}$ clusters in expectation. Assuming that claim holds up to $i-1$, in phase $i$, each cluster in $\mathcal{C}_{i-1}$ is sampled independently with probability $|\mathcal{C}_{i-1}|/n$. Thus, $|\mathcal{C}_{i}|=|\mathcal{C}_{i-1}|^{2}/n=n^{1-\frac{2^{i}}{k}}$ in expectation.
\end{proof}

\begin{claim}[$0$-unclustered]
\label{cl:0-unc-3epshopset}
For any $0$-unclustered vertex $u\in V$, it holds that $\dist_{H}^{(1)}(u,v) = \dist_{H}(u,v)$ for all $v\in \Ball_{G}(u,r_{0})$
\end{claim}
\begin{proof}
If $u$ is an $0$-unclustered vertex, then $A_{1}\cap \Ball_{G}(u,r_{0})=\emptyset$. Since we add hops from $u$ to any vertex at distance at most $\dist_{G}(u,p(u))>r_{0}$, it holds that $H$ contains a hop to any $v\in \Ball_{G}(u,r_{0})$.
\end{proof}
\begin{claim}[$i$-unclustered]
\label{claim:i-unc-3epshopset}
For any $i$-unclustered vertex $u$ for $i\ge 1$ and every $v \in \Ball_{G}(u,\alpha_{i})$, it holds that $\dist_H^{(3)}(u,v)\leq 3\cdot \dist_{G}(u,v)+4\cdot r_{i-1}$.  
\end{claim}
\begin{proof}
Consider an $i$-unclustered vertex $u$ and $v \in \Ball_{G}(u,\alpha_{i})$. Let $C_{u}\in \mathcal{C}_{i-1}$ be the cluster of $u$. Thus $\cdist_G(v, \mathcal{C}_{i-1})\leq r_{i-1}+\dist_{G}(v,u)$. Let $C_v$ be the cluster with the closest center to $v$ in $\mathcal{C}_{i-1}$, then the algorithm adds to $H_i$ a hop from $v$ to the center of $C_v$. We have that:
\begin{eqnarray*}
\cdist_{G}(C_{u},C_{v})\le\cdist_{G}(C_{u},u)+\dist_{G}(u,v)+\cdist_{G}(v,C_{v})
\le 2\cdot r_{i-1}+2\cdot \dist_{G}(u,v)\le 2\cdot r_{i-1}+2\alpha_i.
\end{eqnarray*}
Since $u$ becomes unclustered in phase $i$, the algorithm adds to $H_i$ a hop from the center of $C_{u}$, to the center of $C_{v}$. Overall, we have the following $3$-hop $u$-$v$ path: $(u,r(C_{u}))\circ (r(C_{u}),r(C_{v}))\circ (r(C_{v}),v)$, where $r(C_{u}),r(C_{v})$ are the centers of $C_{u},C_{v}$ respectively. The length of this path is bounded by:
\begin{eqnarray*}
\dist_{H}^{(3)}(u,v)&\le& \cdist_G(u,C_u)+\cdist_G(C_u,C_v)+\cdist_G(C_{v},v)
\\& \leq & 
r_{i-1}+(2\cdot r_{i-1}+2\cdot \dist_{G}(u,v))+r_{i-1}+\dist_{G}(u,v)= 3\cdot \dist_{G}(u,v)+4\cdot r_{i-1}~.
\end{eqnarray*}
\end{proof}
\begin{claim}
\label{claim:3epsclustered}
Let $u,v\in V$ be a pair of vertices at distance $[d,2d]$ in $G$, and let $P_{uv}$ be their shortest path in $G$. If there is a clustered vertex $w\in P_{uv}$ then $\dist_{H}^{(2)}(u,v)\le \dist_{G}(u,v)+ 2\cdot r_{T}\le 3\cdot d$.
\end{claim}
\begin{proof}
Since $w$ is clustered, it follows that in the last step we add hops from $u$ and $v$ to the center $z$ of the last cluster, therefore:
\begin{eqnarray*}
\dist_{H}^{(2)}(u,v)&\le& \dist_{H}(u,z)+\dist_{H}(z,v)= \dist_{G}(u,z)+\dist_{G}(z,v)
\\&\le& \dist_{G}(u,w)+\dist_{G}(w,z)+\dist_{G}(z,w)+\dist_{G}(w,v)\le \dist_{G}(u,v)+2\cdot r_{T}\le 3\cdot d,
\end{eqnarray*}
where the last inequality follows by Observation \ref{obs:rad3epshopset}.
\end{proof}
We are now ready to complete the stretch argument and show Lemma \ref{lem:3plusepshopset}.
\begin{proof}[Lemma \ref{lem:3plusepshopset}]
Let $u,v\in V$ be vertices at distance $d'\in [d,2d]$ in $G$, and let $P$ be the shortest path between them in $G$. First observe that by Claim \ref{claim:3epsclustered}, if there is some clustered vertex in $P$ then we are done, thus we assume there is none. We define a sequence of vertices between $u$ and $v$ in the following way: let $v_{0}:=u$, and iteratively, given $v_{j}\neq v$, set  $v'_{j}$ as the furthest vertex from $v_{j}$ on the segment $P[v_{j},v]$, at distance at most $\Delta_{j}$ from $v_j$, where: 
\begin{equation*}
\Delta_{j}=\begin{cases}
r_{0} & v_{j}\text{ is \ensuremath{0}-unclustered}\\
\min\{\alpha_{i},\dist_{G}(v_{j},v)\} & v_{j}\text{ is \ensuremath{i}-unclustered for \ensuremath{i\ge 1}}
\end{cases}
\end{equation*}
Furthermore, for each $v'_{j}$, set $v_{j+1}$ to be the be the next vertex on $P[v'_{j},v]$, or just $v$ in case $v'_{j}=v$. By Claim \ref{cl:0-unc-3epshopset}, if $v_{j}$ is $0$-unclustered then $\dist_{G\cup H}^{(2)}(v_{j},v_{j+1})=\dist_{G}(v_{j},v_{j+1})$. If $v_{j}$ is $i$-unclustered for $\geq 1$, then there are two cases. 
First consider the case where $\dist_{G}(v_{j},v_{j+1})\ge \alpha_{i}$. In this case by Claim \ref{claim:i-unc-3epshopset}, it holds that $\dist_{H}^{(3)}(v_{j},v'_{j})\le 3\cdot \dist_{G}(v_{j},v'_{j})+4\cdot r_{i-1}$, thus since $\dist_{G}(v_{j},v_{j+1}) > \alpha_{i}$, we have that,
\begin{eqnarray*}
\dist_{H}^{(4)}(v_{j},v_{j+1})\le 3\cdot \dist_{G}(v_{j},v'_{j})+4\cdot r_{i-1}+\dist_{G}(v'_{j},v_{j+1})\le (3+\epsilon)\cdot \dist_{G}(v_{j},v_{j+1}).
\end{eqnarray*}
Now, assume that $\dist_{G}(v_{j},v_{j+1})<\alpha_{i}$, that is $v_{j+1}=v$. By Claim \ref{claim:i-unc-3epshopset} it holds that \newline $\dist_{H}^{(3)}(v_{j},v)\le 3\cdot \dist_{G}(v_{j},v)+4\cdot r_{i-1}$. By Observation \ref{obs:rad3epshopset}, $r_{i-1} \le r_{T-1}\le \frac{d}{2\cdot (5+16/\epsilon)}$, an therefore, we get that
\begin{eqnarray*}
\dist_{H}^{(3)}(v_{j},v)\le 3\cdot \dist_{G}(v_{j},v)+4\cdot r_{T-1}\le 3\cdot \dist_{G}(v_{j},v)+\frac{\epsilon}{8}\cdot d~.
\end{eqnarray*}
Since the path $P$ is partitioned into at most $\ell=4\cdot R'$ segments, we have:
\begin{eqnarray*}
\dist_{H}^{(16\cdot R')}(u,v)&\le&\sum_{j=0}^{4\cdot R'}\dist_{H}^{(2)}(v_{j},v_{j+1})\le \frac{\epsilon}{8}\cdot d+\sum_{j=0}^{d}(3+\epsilon)\cdot \dist_{G}(v_{j},v_{j+1})
\\&\le& (3+1.125\cdot \epsilon)\cdot \dist_{G}(u,v)~.
\end{eqnarray*}
\end{proof}
Lemma \ref{lem:3plusepshopset} follows by using $\epsilon'=\frac{8}{9}\cdot \epsilon$, and repeating the algorithm for each of the $\log \Lambda$ distance ranges.

\paragraph{Size Analysis.} We next bound the total number of hops added to the hopset.
The first step contributes $O(n^{1+1/k})$ edges, in expectation. This follows by Fact \ref{fact:TZ}(iii).
At each phase $i \geq 1$ the algorithm adds the following hops to $H_i$. Each vertex at center-distance at most $2\cdot r_{i-1}+\alpha_{i}$ adds a hop to its closest cluster center in $\mathcal{C}_{i-1}$. This adds at most $n$ hops. In addition, in step (3), we add at most $n$ hops from each vertex to its new cluster center. 

By a similar argument to the proof of Lemma \ref{lem:expectedamountofsp}, it follows that in step (5) of the $i^{th}$ phase, the algorithm adds $1/p_{i}$ hops for each cluster $C \in \mathcal{C}_{i-1}$ in expectation, where $p_{i}$ are defined in Eq. \ref{eq:alpha-p}. Summing over all $|\mathcal{C}_{i-1}|$ clusters, this adds $|\mathcal{C}_{i-1}|/p_{i}=n$ hops by Eq. (\ref{eq:alpha-p}), in expectation. Overall in phase $i$, $O(n)$ edges are added to the hopset $H_i$, and summing over all $T$ phases gives a total of $O(\log{k}\cdot n)$ edges in expectation.  Finally, the last step adds at most $n-1$ edges, this completes proof of Lemma \ref{lem:3plusepshopset}.
\section{Efficient Computation of Spanners, Hopsets, and Applications}\label{sec:eff-con}
\subsection{Efficient Constructions of $(3+\epsilon,\beta)$ Spanners and Applications}\label{sec:spanner-app}
In this section, we provide efficient implementation of $(3+\epsilon,\beta)$ spanners in various computational settings, and show their applications to shortest path computation.

\paragraph{A Modified Meta-Algorithm.} 
The algorithm is similar to the algorithm of Section \ref{sec:3epsspanner} up to modifying the number of phases and the sampling probability of Eq. (\ref{eq:alpha-p}). As in \cite{ElkinN19}, we introduce an efficiency parameter $\rho \in (0,1]$ that determines the trade-off between the $\beta$ value, the number of edges in the spanner, and the construction time.
For a given parameter $\rho$, define:
$$i_{0}=\lceil \log{(k\cdot \rho)} \rceil, i_{1}=\lceil 2/\rho -1 \rceil, \mbox{~~and~~} T':=i_{0}+i_{1}~.$$ 
The modified algorithm applies $T':=i_{0}+i_{1}$ phases instead of  $T=\lceil \log{k}+1 \rceil$ phases. The sampling probabilities are modified as follows: for every $1\le i \le i_{0}$, let $p_{i}$ be as defined in Eq. (\ref{eq:alpha-p}), and for every $i_{0}\le i\le i_{1}$, set $p_{i}=p_{i_{0}}$. 
We first show the correctness of this modified algorithm and then analyze its implementation in several computational settings. 
\begin{observation}
	For all $1\le i \le T'$ it holds that $p_{i}\ge n^{-\rho/2}$.
	\label{obs:pinro}
\end{observation}
\begin{proof}
	For each $i\le i_{0}$, by Eq. (\ref{eq:alpha-p}) we have $p_{i}= |\mathcal{C}_{i-1}|/n$, thus for all $1\le i \le T'$ it holds that $p_{i}\ge |\mathcal{C}_{i_{0}-1}|/n$. By Observation \ref{obs:clusters3epsspanner-nclusters} it holds that in expectation $|\mathcal{C}_{i_{0}}|=n^{1-2^{i_{0}-1}/k}\ge n^{1-2^{\log{(k\cdot \rho)}-1}/k}=n^{1-\rho/2}$, thus $p_{i}\ge n^{-\rho/2}$.
\end{proof}
\begin{lemma}\label{lem:few-bfs}
In each phase $i \in \{1,\ldots, T'\}$ consider the collection of BFS trees of depth $2\cdot r_{i-1}+2\cdot \alpha_{i}$ rooted at the centers of the clusters of $\mathcal{C}_{i-1}$ that are unclustered in $\mathcal{C}_i$. Then each vertex $v \in V$ appears in $O(n^{\rho/2} \cdot \log n)$ such trees w.h.p. 
\end{lemma}
\begin{proof}
Fix a vertex $v\in V$, and let $C_{1},\ldots ,C_{\ell}$ be the clusters in $\mathcal{C}_{i-1}$ with center-distance at most $2\cdot r_{i-1}+2\cdot \alpha_{i-1}$ from $v$. Note that 
$\dist_G(r(C_j),r(C_{j'})\leq 4\cdot r_{i-1}+4\cdot \alpha_{i-1}$ for every $j,j' \in \{1,\ldots, \ell\}$ where $r(C)$ is the center of the cluster $C$. This implies that if one of these clusters in sampled, the all these clusters will be part of the clustering $\mathcal{C}_i$ and would not be part of step (5) in this $i$ phase. Therefore all clusters the expected number of BFS traversals that reach $v$ in step (5) of the phase $i$ is at most $\ell\cdot (1-p_{i})^{\ell}\le 1/p_{i}$. 
Since by Observation \ref{obs:pinro} for all $i$, $1/p_{i}\le n^{\rho/2}$, we have that each vertex is traversed by $O(n^{\rho/2})$ BFS traversals in expectation, and by the Chernoff Bound, w.h.p, each vertex is traversed by at most $O(\log{n}\cdot n^{\rho/2})$ traversals.
\end{proof}	
\begin{lemma}\label{lem:few-centers}
After $T'$ phases, the number of remaining clusters $\mathcal{C}_{T'}$ is at most $O(\log n)$ w.h.p.
\end{lemma}
\begin{proof}
Until the $i_{0}^{th}$ phase the algorithm runs similarly to the algorithm in Section \ref{sec:3epsspanner}, thus by Observation \ref{obs:clusters3epsspanner-nclusters}, in expectation $|\mathcal{C}_{i_{0}}|=n^{1-2^{i_{0}-1}/k}$. In each of the remaining $i_{1}$ phases we sample with probability $p_{i_{0}}$ which by Observation \ref{obs:pinro} is at least $n^{-\rho/2}$, thus after $i_{1}$ such sampling steps, we will have  $|\mathcal{C}_{T'}|=n^{1-2^{(i_{0}-1)}/k-(\rho\cdot i_{1})/2}\le n^{1-(\rho/2)(1+i_{1})}\le 1$ in expectation. By Chernoff it follows that w.h.p $|\mathcal{C}_{T'}|=O(\log{n})$.
\end{proof}

\begin{observation}
\label{obs:rad-meta-eff-spanner}
The modified algorithm computes a $(3+\epsilon,\beta)$ spanner $H \subseteq G$ with $\beta=(5+16/\epsilon)^{\log{\rho}+2/\rho}\cdot k^{\log(5+16/\epsilon)}$ and $O(n^{1+1/k}+\beta\cdot n)$ edges w.h.p.
\end{observation}
\begin{proof}
The final radius after $T'$ phases is given by plugging $T'$ in  Observation \ref{obs:rad3epsspanner}:
\begin{eqnarray}\label{eq:rad-T}
r_{T'}\le 2\cdot (5+16/\epsilon)^{i_{0}+i_{1}-1}\le 2\cdot (5+16/\epsilon)^{\log{(k\cdot \rho)}+2/\rho}=2\cdot (5+16/\epsilon)^{\log{\rho}+2/\rho}\cdot k^{\log(5+16/\epsilon)}~.
\end{eqnarray}
Since we only modified the sampling probabilities of the spanner in \ref{sec:3epsspanner}, it follows that the stretch arguments are unchanged, except that the final radius $r_{T'}$ is larger . Thus by Claim \ref{claim:clustered-3epsspanner}, we get that for every $u,v\in V$ it holds that $\dist_{H}(u,v)\le (3+\epsilon)\cdot \dist_{G}(u,v)+2\cdot r_{T'}$.

We now bound the size of the spanner. The first phase, adds $n^{1+1/k}$ edges, in expectation, and in any subsequent phase, the algorithm adds $O(n)$ shortest paths in expectation, each of length at most $r_{i}$. Thus overall, this adds $O(n\cdot \sum_{i=1}^{T'}r_{i})=O(n\cdot \beta)$ edges. The last step adds $O(n)$ edges due to adding a constant number of BFS trees.
\end{proof}

\subsubsection{The Centralized Setting}
\label{subsec:centralized}
The trade-off between the $\beta$ value, the spanner size and the running time of the algorithm is summarized below:
\begin{lemma}
For any graph $G=(V,E)$, integer $k\ge 1$ and any $\epsilon>0,1\ge \rho>0$, one can compute a $(3+\epsilon,\beta)$ spanner $H \subseteq G$ for $\beta = O((5+16/\epsilon)^{\log{\rho}+2/\rho}\cdot k^{\log(5+16/\epsilon)})$ with $O(n^{1+1/k}+(5+16/\epsilon)^{\log{\rho}+2/\rho}\cdot k^{\log(5+16/\epsilon)}\cdot n)$ edges and $O((\log{(k\cdot \rho)}+1/\rho)\cdot |E|\cdot n^{\rho})$ time.
\end{lemma}
\begin{proof}
The algorithm has $T'=\lceil \log{(k\cdot \rho)} \rceil + \lceil 2/\rho -1 \rceil \le \log(k\cdot \rho)+2/\rho+1$ phases. In each phase there are five steps that are implemented as follows. In step (1), we add shortest paths of length $r_{i-1}+\alpha_i$ from unclustered vertices to their closest center. Since each vertex connected to at most its closest center, while breaking ties in a consistent manner this can be done in $O(|E|)$ time. In the same manner, also step (3) can be implemented in $O(|E|)$ time. Finally, we consider the fifth phase where we grow BFS trees from all centers that did not join the clustering $\mathcal{C}_i$. By Lemma \ref{lem:few-bfs}, each vertex appears in at most $O(n^{\rho}\log n)$ trees, and thus these trees can be computed in $O(n^{\rho}|E|)$ time. 
The overall running time is then bounded by $O(n^{\rho}|E| \cdot T')$ as desired. 
\end{proof}
By computing efficiently the $(3+\epsilon,\beta)$ spanners, we also get the following fast computation of the $S \times V$ distances. The next is Corollary 19 of \cite{ElkinN19} while enjoying a better tradeoff in the expense of increasing the multiplicative stretch from $(1+\epsilon)$ to $(3+\epsilon)$:
\begin{corollary}
There exists an algorithm that computes for any graph $G=(V,E)$, integer $k\ge 1$, any parameters $\epsilon>0,\rho \in (0,1)$ and any vertex set $S\subseteq{V}$, a $(3+\epsilon,\beta)$ approximate shortest paths for $S\times V$, for $\beta = O((5+16/\epsilon)^{\log{\rho}+2/\rho}\cdot k^{\log(5+16/\epsilon)})$ in $O((\log{(k\cdot \rho)}+1/\rho)\cdot |E|\cdot n^{\rho}+|S|\cdot (n^{1+1/k}+(5+16/\epsilon)^{\log{\rho}+2/\rho}\cdot k^{\log(5+16/\epsilon)}\cdot n)$ time. 
\end{corollary}

\subsubsection{The Distributed Setting}

\paragraph{The \local\ Model.} We now consider the implementation details of our spanner construction in the standard \local\ model \cite{Peleg:2000}. In this model, the algorithm's execution proceeds in synchronous rounds, and in every round, each node can send a message (possibly of unbounded size) to 
each of its neighbors. Each node holds a processor with a unique and arbitrary ID of $O(\log n)$ bits.

One of the key effects of improving the value of the $\beta$ in our spanner is that we can compute our $(4+\epsilon,\beta)$ spanner in $O(\beta)$ rounds, hence for $\epsilon=1$, in $O(k^{\log 21})$ rounds. This should be compared against the local computation of $(1+\epsilon,\beta)$ spanners in $O_{\epsilon}(\log k)^{\log k}$ rounds. 
\begin{lemma}\label{lem:3plusepsspannerloc}
For any graph $G=(V,E)$, integer $k$ and $\epsilon>0$, one can compute in the \local\ model a $(4+\epsilon,\beta)$ spanner $H \subseteq G$ for $\beta = O((5+16/\epsilon)\cdot k^{\log(5+16/\epsilon)})$ in $\widetilde{O}(\beta)$ rounds w.h.p. 
\end{lemma}

The \local\ implementation is exactly as in Section \ref{sec:3epsspanner} with two modifications. First, we will now make all the arguments hold with high probability of $1-1/n^c$ for some constant $c$, rather than in expectation. Specifically, in last step of the algorithm there are now $O(\log n)$ clusters in $\mathcal{C}_T$, w.h.p. Instead of adding a BFS tree w.r.t to each center, we will add a truncated BFS tree up to depth $5\cdot r_{T}$ from each of the $O(\log n)$ centers. We now show that this slightly increases the stretch of the spanner, by proving the analogue of Cl. \ref{claim:clustered-3epsspanner}:
\begin{claim}
Fix a pair $u,v\in V$ and let $P$ be their shortest path in $G$. If there is a clustered vertex $w\in V(P)$, then:
\begin{equation*}
\dist_{H}(u,v)\le (4+\epsilon)\cdot \dist_{G}(u,v)+(16+4\cdot\epsilon)\cdot r_{T}~.
\end{equation*}
\end{claim} 
\begin{proof}
First assume that $\dist_{G}(u,v)\le 4\cdot r_{T}$. Let $w$ be some clustered vertex on $P$, and let $s$ be the center of the cluster to which $w$ belongs. In this case since $\dist_{G}(w,s)\le r_{T}$, it holds that $\dist_{G}(u,s)\le 5\cdot r_{T},\dist_{G}(v,s)\le 5\cdot r_{T}$, hence $\dist_{H}(u,s)=\dist_{G}(u,s)$ and $\dist_{H}(v,s)=\dist_{G}(v,s)$. Consequently, $\dist_{H}(u,v)\le \dist_{G}(u,s)+\dist_{G}(s,v)\le \dist_{G}(u,v)+2\cdot r_{T}$. Since we changed only the last step, it implies that for any $1\le d \le 4\cdot r_{T}$, we have that for vertices $u,v\in V$ at distance $d$ in $G$, $\dist_{H}(u,v)\le (3+\epsilon)\cdot \dist_{G}(u,v)+4\cdot r_{T}$. For $u,v\in V$ with $\dist_{G}(u,v)>4\cdot r_{T}$ it holds that:
\begin{eqnarray*}
\dist_{H}(u,v)&\le& \left \lceil \frac{\dist_{G}(u,v)}{4\cdot r_{T}} \right \rceil \cdot ((3+\epsilon)\cdot 4\cdot r_{T} +4\cdot r_{T})
\\&<& \left(\frac{\dist_{G}(u,v)}{4\cdot r_{T}} +1\right) \cdot ((3+\epsilon)\cdot 4\cdot r_{T} +4\cdot r_{T})
\\&=& (4+\epsilon)\cdot \dist_{G}(u,v)+(16+4\cdot\epsilon)\cdot r_{T}.
\end{eqnarray*}
\end{proof}
As all the steps of the algorithms are now restricted to the $O(\beta)$-ball of each vertex, therefore Lemma \ref{lem:3plusepsspannerloc} follows. 

\paragraph{The \congest\ Model.}
We next consider the implementation details of our spanner construction in the standard \congest\ model \cite{Peleg:2000}. This model is exactly as the \local\ only that in each round, a vertex is limited to send $O(\log n)$ bits on each of its incident edges. 

The implementation in the \congest\ model follow the same line of the meta-algorithm, only that in the last step we build a truncted BFS tree up to depth $5\cdot r_{T'}$ as in the local implementation. 
We have: 
\begin{lemma}\label{lem:3plusepsspannercongest}
For any graph $G=(V,E)$, integer $k$ and any $\epsilon>0,1\ge \rho>0$, one can compute in the \congest\ model a $(4+\epsilon,\beta)$ spanner $H \subseteq G$ for $\beta = O((5+16/\epsilon)^{\log{\rho}+2/\rho}\cdot k^{\log(5+16/\epsilon)})$ in $\widetilde{O}(n^{\rho}\cdot \beta)$ rounds w.h.p.
\end{lemma}
\begin{proof}
There are $T'$ phases. We will fix phase $i$ and show it can implemented in $\widetilde{O}(n^{\rho}+r_i)$ rounds. 
Steps (1) and (3) are based on congestion-free BFS computation up to depth $O(r_i)$. 
Step (5) builds a collection of BFS trees up to depth $O(r_i)$. By Lemma \ref{lem:few-bfs}, each vertex is traversed by $O(n^{\rho}\log n)$ trees. Computing a collection of BFS trees up to depth $r_{i}$ with edge congestion $O(n^{\rho})$ can be done in $\widetilde{O}(n^{\rho}+r_i)$ rounds w.h.p using the random delay approach. Therefore by summing over all $T'$ phases we get $\widetilde{O}(n^{\rho}\cdot \beta)$ rounds. 

Finally, in the last step we have $O(\log n)$ centers in $\mathcal{C}_{T'}$ w.h.p. Computing a depth $O(r_{T'})$-trees from each center can be done in $\widetilde{O}(r_{T'})$ rounds. The time analysis follows. 
\end{proof}

\subsubsection{The Multi-Pass Streaming Setting}
\paragraph{Model.} In the streaming model the input graph is presented to the algorithm edge by edge
as a stream without repetitions and the goal is to solve the problem while minimizing the number of passes and space. For graph algorithms, the usual assumption is that the edges of the input graph are presented to the algorithm in arbitrary order. The next is Corollary 20 of \cite{ElkinN19} while enjoying a better tradeoff in the expense of increasing the multiplicative stretch from $(1+\epsilon)$ to $(4+\epsilon)$:
\begin{lemma}\label{lem:spanner-stream}
For any $n$-vertex unweighted graph $G=(V,E)$, integer $k$ and and $\epsilon>0,1\ge \rho>0$, one can compute in the multi-pass streaming model a $(4+\epsilon,\beta)$ spanner $H\subseteq G$ for $\beta = O((5+16/\epsilon)^{\log{\rho}+2/\rho}\cdot k^{\log(5+16/\epsilon)})$ in $O(\log{n}\cdot n^{1+\rho})$ space w.h.p and $O(\beta)$ passes, or with $O(n^{1+1/k}+(\beta+\log{n}) \cdot n)$ space in expectation and $O(\log{n}\cdot n^{\rho}/\rho \cdot \beta)$ passes w.h.p.
\end{lemma}
\begin{proof}
We will use two alternative implementations in the streaming model, in a very similar way to Theorem 5 in \cite{ElkinN19}. In both implementations, for the BFS traversals the algorithm keeps for each traversed vertex the ID of its parent, the ID and of the root of the BFS tree, and it distance to the the root. The first implementation is very similar to the implementation that we described for the \congest\ model. In this implementation, we compute only truncated BFS trees up to depth $5\cdot r_{T}$, rather then a complete BSF tree. A truncated BFS traversal up to depth $r_i$ can be implemented in $r_{i}$ passes, thus overall the algorithm can be implemented in $O(\beta)$ passes. Since each vertex is visited by $O(\log{n}\cdot n^{\rho})$ BFS trees w.h.p, the total space used is bounded by $O(\log{n}\cdot n^{1+\rho})$ space. 
We now consider the alternative implementation. To reduce the BFS congestion of $n^{\rho}$ in the fifth step, this step is divided into $\tau=c\cdot \log{n}/p_{i}$ sub-steps. In each sub-step, we will sample each of the remaining centers (from which we would like to compute the BFS traversal) independently with probability of $p_i$. We will then compute the truncated BFS traversal only from the centers that got sampled in this sub-step. By the Chernoff bound, w.h.p., each vertex will be visited by $O(\log n)$ traversal, hence a space of $O(n\log n)$ plus the space of the spanner is sufficient for the implementation. After $\tau$ sub-steps, w.h.p., the algorithm has computed the truncated BFS-traversal from each of the cluster centers. The total number of passes is bounded by $O(n^{\rho}/\rho \cdot \log{n}\cdot \beta)$.
\end{proof}

Lemma \ref{lem:apsp-stream} follows immediately by Lemma \ref{lem:spanner-stream}.

\subsection{Efficient Constructions of $(\alpha,\beta)$ Hopsets}\label{sec:eff-hopsets}
In this section we show an efficient construction of $(3+\epsilon,\beta)$ hopsets.
We use the following fact that follows from the proof of Theorem 1.1 in \cite{ThorupZ05}:
\begin{fact}\label{fact:efficientTZ}
Each of the $k$ clustering steps in the distance oracle algorithm by Thorup and Zwick can be implemented in $O(|E| \cdot n^{1/k})$ centralized time.
\end{fact}
\paragraph{The Meta Algorithm.}
The algorithm is similar to the algorithm in the proof of Lemma \ref{lem:3plusepshopset} up to modifying the number of phases and the sampling probability of Eq. (\ref{eq:alpha-p}) in the exact same manner as in Section \ref{sec:spanner-app}. In addition, since we are now working with weighted graphs, we will be using the Dijkstra algorithm to compute shortest path trees, instead of BFS traversals. 
Fix a distance class $[d,2d]$ and define $R''=(5+18/\epsilon)^{\log{\rho}+2/\rho}\cdot k^{\log(5+18/\epsilon)}$. We then slightly change the initial radius of the clustering $\mathcal{C}_0$ to be $r_0=d/(2\cdot R'')$. By similar arguments as in the proof of Lemma \ref{lem:3plusepshopset} and Observation \ref{obs:rad-meta-eff-spanner}, the final radius after $T'=\lceil \log{(k\cdot \rho)} \rceil+\lceil 2/\rho -1 \rceil$ phases is:
\begin{eqnarray}\label{eq:rad-T-eff-hopset}
r_{T'}\le r_0\cdot (5+18/\epsilon)^{T'-1} \le d/2.
\end{eqnarray}
The remaining details are almost identical to the meta-algorithm for spanners so we only state the properties of this construction:
\begin{observation}
The modified algorithm computes a $(3+\epsilon,\beta)$ hopset $H \subseteq G$ with $\beta=O((5+18/\epsilon)^{\log{\rho}+2/\rho}\cdot k^{\log(5+18/\epsilon)})$ and $O((n^{1+1/k}+(\log{k}+1/\rho)\cdot n)\cdot \log{\Lambda})$ edges in expectation.
\end{observation}

\paragraph{Efficient Implementation in the Centralized Setting.}
The tradeoff between the $\beta$ value, the hopset size and the running time of the algorithm is summarized below:
\begin{lemma}
For any graph $G=(V,E,w)$, integer $k\ge 1$ and any $\epsilon>0,1\ge \rho>0$, one can compute a $(3+\epsilon,\beta)$ hopset $H$ with $\beta=O((5+18/\epsilon)^{\log{\rho}+2/\rho}\cdot k^{\log(5+18/\epsilon)})$ with $O( (n^{1+1/k}+(\log{k}+1/\rho)\cdot n)\cdot \log{\Lambda})$ edges and $O(|E| \cdot (n^{1/k}+(\log(k\cdot \rho)+1/\rho)\cdot n^{\rho})\cdot \log\Lambda)$ time in expectation.
\end{lemma}
\begin{proof}
The stretch and size arguments are almost similar to those in Section \ref{sec:spanner-app}, hence we restrict attention to the running time. 
Our algorithm begins with a single clustering step of the Throup and Zwick's algorithm (i.e., computing the first bunch $B_1(u)$ for each vertex $u$). The output of this step is a clustering $\mathcal{C}_{0}$. Since clusters are vertex-disjoint, one can compute them in $\widetilde{O}(|E|)$ time. 
Thus using Fact \ref{fact:efficientTZ} the entire first part of the algorithm can be implemented in $O(|E| \cdot n^{1/k})$ centralized time. From this point onward the analysis is similar to the analysis of the centralized implementation of spanners thus requiring $O(n^{\rho}\cdot |E| \cdot T')$ centralized time.
\end{proof}

\bibliographystyle{alpha}
\bibliography{alpha-beta-hopsets}

\newcommand{\etalchar}[1]{$^{#1}$}
\begin{thebibliography}{ADD{\etalchar{+}}93b}

\bibitem[AB17]{abboud20174}
Amir Abboud and Greg Bodwin.
\newblock The 4/3 additive spanner exponent is tight.
\newblock {\em Journal of the ACM (JACM)}, 64(4):28, 2017.

\bibitem[ABP18a]{abboud2018hierarchy}
Amir Abboud, Greg Bodwin, and Seth Pettie.
\newblock A hierarchy of lower bounds for sublinear additive spanners.
\newblock {\em SIAM Journal on Computing}, 47(6):2203--2236, 2018.

\bibitem[ABP18b]{AbboudBP18}
Amir Abboud, Greg Bodwin, and Seth Pettie.
\newblock A hierarchy of lower bounds for sublinear additive spanners.
\newblock {\em {SIAM} J. Comput.}, 47(6):2203--2236, 2018.

\bibitem[ACIM99]{aingworth1999fast}
Donald Aingworth, Chandra Chekuri, Piotr Indyk, and Rajeev Motwani.
\newblock Fast estimation of diameter and shortest paths (without matrix
  multiplication).
\newblock {\em SIAM Journal on Computing}, 28(4):1167--1181, 1999.

\bibitem[ADD{\etalchar{+}}93a]{althofer1993sparse}
Ingo Alth{\"o}fer, Gautam Das, David Dobkin, Deborah Joseph, and Jos{\'e}
  Soares.
\newblock On sparse spanners of weighted graphs.
\newblock {\em Discrete \& Computational Geometry}, 9(1):81--100, 1993.

\bibitem[ADD{\etalchar{+}}93b]{AlthoferDDJS93}
Ingo Alth{\"{o}}fer, Gautam Das, David~P. Dobkin, Deborah Joseph, and
  Jos{\'{e}} Soares.
\newblock On sparse spanners of weighted graphs.
\newblock {\em Discrete {\&} Computational Geometry}, 9:81--100, 1993.

\bibitem[BKMP05]{baswana2005new}
Surender Baswana, Telikepalli Kavitha, Kurt Mehlhorn, and Seth Pettie.
\newblock New constructions of ($\alpha$, $\beta$)-spanners and purely additive
  spanners.
\newblock In {\em Proceedings of the sixteenth annual ACM-SIAM symposium on
  Discrete algorithms}, pages 672--681. Society for Industrial and Applied
  Mathematics, 2005.

\bibitem[BS07]{BaswanaS07}
Surender Baswana and Sandeep Sen.
\newblock A simple and linear time randomized algorithm for computing sparse
  spanners in weighted graphs.
\newblock {\em Random Struct. Algorithms}, 30(4):532--563, 2007.

\bibitem[Che13]{chechik2013new}
Shiri Chechik.
\newblock New additive spanners.
\newblock In {\em Proceedings of the twenty-fourth annual ACM-SIAM symposium on
  Discrete algorithms}, pages 498--512. Society for Industrial and Applied
  Mathematics, 2013.

\bibitem[Che16]{Chechik16}
Shiri Chechik.
\newblock Additive spanners.
\newblock In {\em Encyclopedia of Algorithms}, pages 22--24. 2016.

\bibitem[Coh00]{cohen2000polylog}
Edith Cohen.
\newblock Polylog-time and near-linear work approximation scheme for undirected
  shortest paths.
\newblock {\em Journal of the ACM (JACM)}, 47(1):132--166, 2000.

\bibitem[Elk17]{elkin2017distributed}
Michael Elkin.
\newblock Distributed exact shortest paths in sublinear time.
\newblock In {\em Proceedings of the 49th Annual ACM SIGACT Symposium on Theory
  of Computing}, pages 757--770. ACM, 2017.

\bibitem[EM70]{erdos1970extremal}
Paul Erd{\"o}s and L~Moser.
\newblock An extremal problem in graph theory.
\newblock {\em Journal of the Australian Mathematical Society}, 11(1):42--47,
  1970.

\bibitem[EN16a]{ElkinN16b}
Michael Elkin and Ofer Neiman.
\newblock Hopsets with constant hopbound, and applications to approximate
  shortest paths.
\newblock In {\em 2016 IEEE 57th Annual Symposium on Foundations of Computer
  Science (FOCS)}, pages 128--137. IEEE, 2016.

\bibitem[EN16b]{elkin2016efficient}
Michael Elkin and Ofer Neiman.
\newblock On efficient distributed construction of near optimal routing
  schemes.
\newblock In {\em Proceedings of the 2016 ACM Symposium on Principles of
  Distributed Computing}, pages 235--244. ACM, 2016.

\bibitem[EN17]{elkin2017linear}
Michael Elkin and Ofer Neiman.
\newblock Linear-size hopsets with small hopbound, and distributed routing with
  low memory.
\newblock {\em arXiv preprint arXiv:1704.08468}, 2017.

\bibitem[EN19a]{ElkinN19}
Michael Elkin and Ofer Neiman.
\newblock Efficient algorithms for constructing very sparse spanners and
  emulators.
\newblock {\em {ACM} Trans. Algorithms}, 15(1):4:1--4:29, 2019.

\bibitem[EN19b]{ElkinN19new}
Michael Elkin and Ofer Neiman.
\newblock Linear-size hopsets with small hopbound, and constant-hopbound
  hopsets in {RNC}.
\newblock In {\em The 31st {ACM} on Symposium on Parallelism in Algorithms and
  Architectures, {SPAA} 2019, Phoenix, AZ, USA, June 22-24, 2019.}, pages
  333--341, 2019.

\bibitem[EP04]{ElkinP04}
Michael Elkin and David Peleg.
\newblock (1+epsilon, beta)-spanner constructions for general graphs.
\newblock {\em {SIAM} J. Comput.}, 33(3):608--631, 2004.

\bibitem[EZ06]{elkin2006efficient}
Michael Elkin and Jian Zhang.
\newblock Efficient algorithms for constructing (1+∊, $\beta$)-spanners in
  the distributed and streaming models.
\newblock {\em Distributed Computing}, 18(5):375--385, 2006.

\bibitem[FL18]{friedrichs2018parallel}
Stephan Friedrichs and Christoph Lenzen.
\newblock Parallel metric tree embedding based on an algebraic view on
  moore-bellman-ford.
\newblock {\em Journal of the ACM (JACM)}, 65(6):43, 2018.

\bibitem[GEN19]{GENPRAMDO19}
Yuval Gitlitz, Michael Elkin, and Ofer Neiman.
\newblock Almost shortest paths in weighted graphs and pram distance oracles.
\newblock {\em Private Communication}, 2019.

\bibitem[HKN16]{henzinger2016deterministic}
Monika Henzinger, Sebastian Krinninger, and Danupon Nanongkai.
\newblock A deterministic almost-tight distributed algorithm for approximating
  single-source shortest paths.
\newblock In {\em Proceedings of the forty-eighth annual ACM symposium on
  Theory of Computing}, pages 489--498. ACM, 2016.

\bibitem[HKN18]{henzinger2018decremental}
Monika Henzinger, Sebastian Krinninger, and Danupon Nanongkai.
\newblock Decremental single-source shortest paths on undirected graphs in
  near-linear total update time.
\newblock {\em Journal of the ACM (JACM)}, 65(6):36, 2018.

\bibitem[HP19]{HuangP19}
Shang{-}En Huang and Seth Pettie.
\newblock Thorup-zwick emulators are universally optimal hopsets.
\newblock {\em Inf. Process. Lett.}, 142:9--13, 2019.

\bibitem[KS97]{klein1997randomized}
Philip~N Klein and Sairam Subramanian.
\newblock A randomized parallel algorithm for single-source shortest paths.
\newblock {\em Journal of Algorithms}, 25(2):205--220, 1997.

\bibitem[MPVX15]{miller2015improved}
Gary~L Miller, Richard Peng, Adrian Vladu, and Shen~Chen Xu.
\newblock Improved parallel algorithms for spanners and hopsets.
\newblock In {\em Proceedings of the 27th ACM symposium on Parallelism in
  Algorithms and Architectures}, pages 192--201. ACM, 2015.

\bibitem[Nan14]{nanongkai2014distributed}
Danupon Nanongkai.
\newblock Distributed approximation algorithms for weighted shortest paths.
\newblock In {\em Proceedings of the forty-sixth annual ACM symposium on Theory
  of computing}, pages 565--573. ACM, 2014.

\bibitem[Par14]{Parter14}
Merav Parter.
\newblock Bypassing erd{\H{o}}s’ girth conjecture: hybrid stretch and
  sourcewise spanners.
\newblock In {\em International Colloquium on Automata, Languages, and
  Programming}, pages 608--619. Springer, 2014.

\bibitem[Pel00]{Peleg:2000}
David Peleg.
\newblock {\em Distributed Computing: A Locality-sensitive Approach}.
\newblock SIAM, 2000.

\bibitem[Pet09]{Pettie09}
Seth Pettie.
\newblock Low distortion spanners.
\newblock {\em {ACM} Trans. Algorithms}, 6(1):7:1--7:22, 2009.

\bibitem[PS89]{peleg1989graph}
David Peleg and Alejandro~A Sch{\"a}ffer.
\newblock Graph spanners.
\newblock {\em Journal of graph theory}, 13(1):99--116, 1989.

\bibitem[PU87]{PelegU87}
David Peleg and Jeffrey~D. Ullman.
\newblock An optimal synchronizer for the hypercube.
\newblock In {\em Proceedings of the Sixth Annual {ACM} Symposium on Principles
  of Distributed Computing, Vancouver, British Columbia, Canada, August 10-12,
  1987}, pages 77--85, 1987.

\bibitem[SS99]{shi1999time}
Hanmao Shi and Thomas~H Spencer.
\newblock Time--work tradeoffs of the single-source shortest paths problem.
\newblock {\em Journal of algorithms}, 30(1):19--32, 1999.

\bibitem[TZ05]{ThorupZ05}
Mikkel Thorup and Uri Zwick.
\newblock Approximate distance oracles.
\newblock {\em J. {ACM}}, 52(1):1--24, 2005.

\bibitem[Woo06]{woodruff2006lower}
David~P Woodruff.
\newblock Lower bounds for additive spanners, emulators, and more.
\newblock In {\em 2006 47th Annual IEEE Symposium on Foundations of Computer
  Science (FOCS'06)}, pages 389--398. IEEE, 2006.

\bibitem[Woo10]{woodruff2010additive}
David~P Woodruff.
\newblock Additive spanners in nearly quadratic time.
\newblock In {\em International Colloquium on Automata, Languages, and
  Programming}, pages 463--474. Springer, 2010.

\end{thebibliography}
\newpage
\appendix 
\section{Complete Proofs of Theorems \ref{thm:secondspanner} and \ref{thm:secondhopset}}\label{app:fraction}
Recall that $t=\lceil k^{\epsilon} \rceil /4$ and $T=\log_{t}(k^{1-2\epsilon})$. We will now handle the case where $T$ is not an integer as assumed in Sections \ref{sec:second-regime-spanner} and \ref{sec:second-regime-hopset}, thus completing the proofs of Theorems \ref{thm:secondspanner} and \ref{thm:secondhopset}. We will focus on the spanner case of Theorem \ref{thm:secondspanner}. The exact same argument holds for the hopsets for Theorem \ref{thm:secondhopset}. By Claim \ref{claim:radiusecondspanner}, when $T$ is an integer we have a final clustering with radius $r_{T}\le 1/30\cdot 64^{1/\epsilon-1}\cdot k^{1-\epsilon}$. 
Let $c=T-\lfloor T \rfloor$ and $T_{0}:=T-c$. We divide the treatment of the fractional case into two possible cases. In the first case, assume that $t^{c}> 1/3\cdot t$. In this case, the middle stage of Alg. $\ImprovedSpannerII$ simply contains $\lceil T \rceil$ phases, each of $t$ steps. In this case we have:
\begin{claim}
\label{cl:bigtc}
For $t^{c}>1/3\cdot t$, it holds that the final clustering $\mathcal{C}_{T_{0}+1}$ has $O(n^{1-1/k^{\epsilon}})$ clusters, in expectation, and the final radius $r_{T_{0}+1}\le 64^{1/\epsilon-1}\cdot k^{1-\epsilon}$.
\end{claim}
\begin{proof}
Since $\lceil T\rceil \ge T$ it follows from Claim \ref{cl:numberofclusters} that the final clustering $\mathcal{C}_{T_{0}+1}$ has $O(n^{1-1/k^{\epsilon}})$ clusters, in expectation. In Claim \ref{claim:radiusecondspanner} we bound the radius of the final clustering in the case where $T$ is an integer by $r^{*} = \lceil k^{\epsilon}\rceil\cdot (2\cdot \lceil k^{\epsilon} \rceil)^{T}\le 1/30\cdot 64^{1/\epsilon-1}\cdot k^{1-\epsilon}$, thus if $t<3 t^{c}$,  by the same claim, we get that the final radius after $T_{0}+1$ phases is at most:
	\begin{eqnarray*}
r_{T_{0}+1}&=& \lceil k^{\epsilon} \rceil \cdot (2\cdot \lceil k^{\epsilon} \rceil)^{T_{0}+1}=\lceil k^{\epsilon} \rceil \cdot (2\cdot \lceil k^{\epsilon} \rceil)^{T-c+1}
\\&=&	(2\cdot \lceil k^{\epsilon} \rceil)^{1-c} \cdot r^{*}=(8\cdot t)^{1-c}\cdot r^*\le 24\cdot r^*\le 64^{1/\epsilon-1}\cdot k^{1-\epsilon}~.
	\end{eqnarray*} 
\end{proof} 
In the complementary case where $t^{c}\le 1/3\cdot t$, the adaptation of Alg. $\ImprovedSpannerII$ is as follows: In the second stage of the algorithm, it applies $T_{0}$ clustering phases as usual (i.e. each with $t=\lceil k^{\epsilon} \rceil /4$ steps), and the last $T_{0}+1$ phase will consist of only $\lfloor t^{c} \rfloor\le t$ steps. We will denote this last $T_{0}+1$ phase as a \emph{fractional} phase. 
Again, we will show that after these $T_{0}+1$ phases, it holds that the number of clusters in the final clustering $\mathcal{C}_{T_{0}+1}$ is $O(n^{1-1/k^{\epsilon}})$ in expectation and that the radius of this clustering is $r_{T_{0}+1}\le 64^{1/\epsilon-1}\cdot k^{1-\epsilon}$. We begin by bounding the number of clusters in the final clustering $\mathcal{C}_{T_{0}+1}$:
\begin{claim}
	\label{cl:fracnumclusters}
After $T_{0}$ phases of $t$ steps and one last phase of $\lfloor t^{c} \rfloor$ steps, the final clustering $\mathcal{C}_{T_{0}+1}$ has $n^{1-1/k^{\epsilon}}$ clusters, in expectation.  
\end{claim}
\begin{proof}
	By Claim \ref{cl:sizesuperclusters} after $T_{0}$ phases, 
	$$|\mathcal{C}_{T_{0}}|=n^{1-\frac{\lceil k^{\epsilon}\rceil\cdot (t+1)^{T_{0}}}{k}}=n^{1-\frac{\lceil k^{\epsilon}\rceil\cdot (t+1)^{T}}{k\cdot (t+1)^{c}}}~,$$ in expectation. By the same claim it follows that after additional $\lfloor t^{c} \rfloor$ steps of Proc. $\ClusterAndAugment$, the expected number of clusters in the final clustering $\mathcal{C}_{T_{0}+1}$ is:
	\begin{eqnarray*}
		n^{1-\frac{\lceil k^{\epsilon}\rceil\cdot (t+1)^{T}\cdot (\lfloor t^{c} \rfloor+1)}{k\cdot (t+1)^{c}}}&\le& n^{1-\frac{\lceil k^{\epsilon}\rceil\cdot (t+1)^{T}\cdot t^{c}}{k\cdot (t+1)^{c}}}\le n^{1-\frac{\lceil k^{\epsilon}\rceil\cdot t^{T}\cdot (\frac{t+1}{t})^{T_{0}}}{k}}\le n^{1-1/k^{\epsilon}}~.
	\end{eqnarray*} 
\end{proof}
We continue to bound the radius of the final clustering $\mathcal{C}_{T_{0}+1}$:
\begin{claim}
	\label{cl:fracrad}
$r_{T_0+1}\le 64^{1/\epsilon-1}\cdot k^{1-\epsilon}$.
\end{claim}
\begin{proof}
By Claim \ref{claim:radiusecondspanner} it holds that after $T_{0}$ phases the radius of the clustering $\mathcal{C}_{T_{0}}$ is bounded by $r_{T_{0}}= \lceil k^{\epsilon}\rceil(2\cdot \lceil k^{\epsilon}\rceil)^{T_{0}} \le \frac{r^*}{(2\cdot \lceil k^{\epsilon}\rceil)^{c}}$. By the same claim after additional $t':=\lfloor t^{c} \rfloor$ steps of Procedure $\ClusterAndAugment$ the radius is bounded by:
\begin{eqnarray*}
r_{T_{0},t'} &\le&\left (1+ \frac{\lceil k^{\epsilon} \rceil+1}{2}\cdot \left((1+\frac{4}{\lceil k^{\epsilon} \rceil-3})^{t'}-1\right)\right)\cdot r_{T_{0},0}-t'\cdot \frac{\lceil k^{\epsilon} \rceil-3}{2}+\frac{k^{2\epsilon}}{8}\cdot \left((1+\frac{4}{\lceil k^{\epsilon} \rceil-3})^{t'}-1\right)
\\&<& \left (1+ \frac{\lceil k^{\epsilon} \rceil+1}{2}\cdot \left(e^{(1+\frac{3}{4\cdot t-3})\cdot t^{c-1}}-1\right)\right)\cdot r_{T_{0},0}-(t^{c}-1)\cdot \frac{\lceil k^{\epsilon} \rceil-3}{2}+\frac{k^{2\epsilon}}{8}\cdot \left(e^{(1+\frac{3}{4\cdot t-3})\cdot t^{c-1}}-1\right)
\\&\le& \left (1+ \frac{4\cdot t+1}{2}\cdot (3.5\cdot t^{c-1})\right)\cdot r_{T_{0},0}+2\cdot t^{2}\cdot 3.5 \cdot t^{c-1}
\\&\le& \left (1+ \frac{4\cdot t+1}{2}\cdot (3.5\cdot t^{c-1})\right)\cdot r_{T_{0},0}+7\cdot t^{1+c}\le 15t^{c}\cdot r_{T_{0},0}\le15\cdot t^{c}\cdot \frac{r^*}{(8\cdot  t)^{c}}\le 15\cdot r^*~,
\end{eqnarray*}
where the third inequality following since for any $x\le 1$:
\begin{eqnarray*}
	e^{x}\le 1+x+\sum_{i=2}x^{i}\le 1+x+\frac{x^{2}}{1-x}~,
\end{eqnarray*}
and in particular since by assumption $4\cdot t\ge 16$, we have:
\begin{eqnarray*}
	e^{(1+\frac{3}{4\cdot t-3})\cdot t^{c-1}}-1&\le& 1.05\cdot t^{c-1}+2.25\cdot \frac{t^{2c-2}}{1-1.05\cdot t^{c-1}}\le 1.05\cdot t^{c-1}+2.25\cdot \frac{t^{2c}}{t(t-1.05\cdot t^{c})}
	\\&\le&3.5\cdot t^{c-1}~,
\end{eqnarray*}
where the last inequality follows since $t\ge 3 t^{c}$. We conclude $r_{T_{0}+1}\le 64^{1/\epsilon-1}\cdot k^{1-\epsilon}$.
\end{proof}
We proceed with providing stretch and size arguments.
\begin{claim}
	For any fixed distance $d$, it holds that Algorithm $\ImprovedSpannerII$ outputs a subgraph $H_{d}\subseteq{G}$ such that for any $u,v\in V$ with $\dist_{G}(u,v)=d$ it holds that $\dist_H(u,v)\le 4\cdot k^{\epsilon}\cdot d+1/15\cdot 64^{1/\epsilon}\cdot k$.
\end{claim}
\begin{proof}
	Fix a distance $d$. By Claim \ref{cl:helper-spanner-second} we have that for any $0$-unclustered vertex $u$ and any $v\in N(u)$ it holds that $\dist_{H}(u,v)\le 2\cdot \lceil k^{\epsilon} \rceil$. 
	
Furthermore, since the definition of the $\alpha_{i,j}$ parameters in unchanged, by Claim \ref{cl:ij-unc-secondspanner}, for any $(i,j)$-unclustered vertex $u$ and any $v\in \partial\Ball_{G}(u,\alpha_{i,j})$  it holds that $\dist_{H}(u,v)\le \lceil k^{\epsilon} \rceil \cdot \dist_{G}(u,v)$. 
In particular, $\dist_{H}(u,v)\le 2\cdot r_{T_{0}+1}$. 
Thus by the exact same argument as in the proof of Lemma \ref{lem:secondspanner} only with $r_{T_{0}+1}$ instead of $r^{*}$, it holds that for any $u,v$ with $\dist_{G}(u,v)=d$, if all the vertices on the shortest path between $u$ and $v$ in $G$ are unclustered, then $\dist_{H}(u,v)\le 2\cdot \lceil k^{\epsilon} \rceil \cdot \dist_{G}(u,v)+2\cdot r_{T_{0}+1}$. 

It remains to prove the analogue of Claim \ref{cl:clustered-secondspanner}, that is, providing the stretch argument for the case where the $u$-$v$ shortest contains at least one clustered vertex $w$. 
By a similar arguments as in Claim \ref{cl:clustered-secondspanner}, we have that in this case $\dist_{H}(u,v)\le 4\cdot k^{\epsilon}\cdot \dist_{G}(u,v)+4\cdot r_{T_{0}+1}$. Thus plugging the value of $r_{T_{0}+1}$ from Claims \ref{cl:bigtc} and  \ref{cl:fracrad}, we get that 
$$\dist_{H}(u,v)\le 4\cdot k^{\epsilon}\cdot d+1/15\cdot 64^{1/\epsilon}\cdot k~.$$
\end{proof}
\begin{claim}
For any fixed distance $d\le 64^{1/\epsilon}\cdot k$ it holds that the algorithm outputs a subgraph $H_{d}$ of expected size $O(k^{\epsilon}\cdot n^{1+1/k}+64^{1/\epsilon}\cdot k^{1+\epsilon}\cdot n)$.
\end{claim}
\begin{proof}
First note that up to the last $T_{0}$ phases of Procedure $\ClusterAndAugment$, the algorithm is similar
to the integral case, thus by the analysis of Lemma \ref{lem:secondspanner}, until this step $O(k^{\epsilon}\cdot n^{1+1/k}+64^{1/\epsilon}\cdot k^{1+\epsilon}\cdot n)$ edges are added to the spanner, in expectation.

We now bound the number of remaining edges added to the spanner in the last phase. By Claim \ref{cl:expectedspinstep} it holds that in each step of the last phase we add $O(n)$ shortest paths in expectation. In step $i$ of the last phase the added shortest paths are of length at most $r_{T_{0},0}+r_{T_{0},i-1}+2\cdot \alpha_{T_{0},i}$, and by Eq. (\ref{eq:rij}) it holds that $r_{T_{0},0}+r_{T_{0},i-1}+2\cdot \alpha_{T_{0},i}\le r_{T_{0},i}$. In the first case, where $t^{c}>1/3\cdot t$, by Claim \ref{cl:bigtc} it holds that we add an extra of $O(n\cdot \sum_{i=1}^{t}r_{T_{0},i+1})=O(n\cdot k^{2\epsilon}\cdot r_{T_{0}+1})=O(n\cdot 64^{1/\epsilon}\cdot k^{1+\epsilon})$ edges in expectation. In the complementary case we add $O(\sum_{i=1}^{\lfloor t^{c} \rfloor+1} n\cdot r_{T_{0},i})=O(n\cdot k^{2\cdot \epsilon}\cdot r_{T_{0},t'})=O(64^{1/\epsilon}\cdot k^{1+\epsilon}\cdot n)$ edges in expectation. Thus over all the middle stage in the fractional case adds an extra of $O(64^{1/\epsilon}\cdot k^{1+\epsilon}\cdot n)$ edges in expectation. 

In the final stage we construct a $(2\cdot \lceil k^{\epsilon} \rceil-3)$ spanner of the cluster graph. By Claims \ref{cl:bigtc} and  \ref{cl:fracnumclusters}, the expected size of the final clustering $\mathcal{C}_{T_{0}+1}$ is $O(n^{1-1/k^{\epsilon}})$, thus as in the proof of Lemma \ref{lem:secondspanner} it holds that the size of the spanner of the cluster graph is $O(n)$. Since each edge in this spanner is translated into a path of length $O(d+r_{T_{0}+1})$ in the final spanner, it holds that this stage adds $O((d+r_{T_{0}+1})\cdot n)$ edges to the spanner. We conclude that the construction in the fractional case is of size $O(k^{\epsilon}\cdot n^{1+1/k}+64^{1/\epsilon}\cdot k^{1+\epsilon}\cdot n+64^{1/\epsilon-1}\cdot k^{1-\epsilon}\cdot n + (d+r_{T_{0}+1}) \cdot n)=O(k^{\epsilon}\cdot n^{1+1/k}+64^{1/\epsilon}\cdot k^{1+\epsilon}\cdot n)$.
\end{proof}
Finally by Observation \ref{claim:isenough}, we conclude that taking the output subgraphs $H_d$ of Algorithm $\ImprovedSpannerII$ for every $1\le d \le 64^{1/\epsilon}\cdot k^{1-\epsilon}$ yields a $(8\cdot k^{\epsilon},64^{1/\epsilon}\cdot k)$-spanner of $G$ of expected size at most $O(64^{1/\epsilon}\cdot k\cdot n^{1+1/k}+64^{2/\epsilon}\cdot k^{2}\cdot n)$.
This completes the proof of the fractional case of Theorem \ref{thm:secondspanner}.
\paragraph{Hopsets.} As in the spanner case we divide the treatment of the fractional case to two cases as explained above. In the first case, where $t^{c}>1/3\cdot t$, the algorithm applies $T_{0}+1$ standard phases of Procedure $\ClusterAndAugmentHop$, each with $t$ steps. In the complementary case, the algorithm applies $T_{0}$ standard phases of the procedure, and a final fractional phase of $t'=\lfloor t^{c} \rfloor$ steps. In Claim \ref{claim:radiushopset} we show that in the case where $T$ is an integer, the final radius is bounded by $r_{T}\le d/648$. By similar claims to Claim \ref{cl:bigtc} and \ref{cl:fracrad}, it holds that the final radius in the fractional case is at most $r_{T_{0}+1}\le 24\cdot r_{T}\le d/27$. Since we have some slackness in the stretch arguments of Claim \ref{cl:clustered-hop-second} and in the proof of Thm. \ref{thm:secondhopset}, the same bound holds when plugging the bound on the final radius $r_{T_{0}+1}$ instead of $r_{T}$. 

For the size analysis, the algorithm is the same of that in Section \ref{sec:second-regime-hopset} until the end of the $T_{0}$th phase. Therefore, by the proof of Thm. \ref{thm:secondhopset} until this last phase, $O(k^{\epsilon}\cdot n^{1+1/k}+k^{\epsilon}/\epsilon\cdot n)$ are added to the hopset. By a similar argument to Claim \ref{cl:expectedspinstep}, the last $T_0+1$ phase with $t'$ steps adds $O(t\cdot n)=O(k^{\epsilon}\cdot n)$ hops to the hopset, in expectation. By the exact same argument as in Claim \ref{cl:fracnumclusters}, in the third stage of the algorithm there are $n^{1-1/k^{\epsilon}}$ clusters, in expectation. Consequently as in the proof of Thm. \ref{thm:secondhopset}, the spanner of the cluster graph is of size $O(n)$, thus this stage contributes $O(n)$ hops to the hopset. We conclude that altogether the hopset in the fractional case is of size $O(k^{\epsilon}\cdot n^{1+1/k}+k^{\epsilon}/\epsilon\cdot n)$. Lemma \ref{lem:secondhopset} follows. 
\section{Improved $(3+\epsilon,\beta)$ Spanners and Hopsets}\label{sec:best-con}
\subsection{Spanners}\label{sec:best-spanner}
In the following subsection we state and sketch the construction of a $(3+\epsilon,\beta)$-spanner with an improved $\beta$ and provide the proof for Lemma \ref{lem:spanner-best}. For example, we show a $(4,\beta)$ spanner for $\beta=k^{\log 11}$. 
The algorithm is similar to the algorithm of Section \ref{sec:3epsspanner}, with the only differences that  we set $\alpha_{1}=1/2$, and in step (3) of the $i^{th}$ phase of the algorithm, the sampled clusters add to $H_{i}$ paths between their centers to centers at distance at most $2\cdot r_{i-1}+2\cdot \alpha_{i}$ (instead of $4\cdot r_{i-1}+4\cdot \alpha_{i}$). This change affects the radii of the clusters in the following way:
\begin{observation}
	\label{obs:rad-spanner-best}
	For every $i \in \{1,\ldots, T\}$, $r_i \leq (3+ 8/\epsilon)^{i-1}$. In particular, the radius of cluster in the final clustering $\mathcal{C}_T$ is $r_T \le (3+ 8/\epsilon)^{\log{k}+1}=(3+ 8/\epsilon)\cdot k^{\log(3+ 8/\epsilon)}$.
\end{observation}
\begin{proof}
	The proof is similar to the proof of Observation \ref{obs:rad3epsspanner}, with the difference that $r_{1}=1$, and at each phase each new cluster is a star of a cluster that connects to other clusters with center-distance at most $2\cdot r_{i-1}+2\cdot \alpha_{i}$, thus the radius of the new cluster is at most $3\cdot r_{i-1}+2\cdot \alpha_{i}\le (3+8/\epsilon)\cdot r_{i-1}$.
\end{proof}
As in the proof of Lemma \ref{lem:3plusepsspanner}, the size of $H$ is bounded by:
\begin{eqnarray*}
	\sum_{i=1}^T |E(H_i)|&\leq & n^{1+1/k}+\sum_{i=2}^{T}r_{i}\cdot n=
	n^{1+1/k}+O(n \cdot \sum_{i=2}^{T}(3+ 8/\epsilon)^{i-1})
	\\&=&O(n^{1+1/k}+(3+ 8/\epsilon)\cdot k^{\log(3+8/\epsilon)}\cdot n)~.
	\end {eqnarray*}
Following the proof of Lemma \ref{lem:3plusepsspanner}, we have that $H$ is a $(3+\epsilon,\beta)$-spanner with $\beta=4\cdot r_{T}=(3+ 8/\epsilon)\cdot k^{\log(3+8/\epsilon)}$, and Lemma \ref{lem:spanner-best} follows. 

\subsection{Hopsets}\label{sec:best-hopset}
In the following subsection we state and sketch the construction of a $(3+\epsilon,\beta)$-spanner with an improved $\beta$ and provide the proof for Lemma \ref{lem:hopset-best}.
For example, we show a $(4,\beta)$ hopset for $\beta=k^{\log 12}$. 
The algorithm is similar to the algorithm of Section \ref{sec:3epshopset}, with the only differences that we set $R'=(3+ \frac{8}{\epsilon})\cdot k^{\log(3+ \frac{8}{\epsilon})}$ and in step (3) of the algorithm, the sampled clusters add to $H_{i}$ hops to cluster centers at distance $2\cdot r_{i-1}+2\cdot \alpha_{i}$ (instead of $4\cdot r_{i-1}+4\cdot \alpha_{i}$). This change affects only the radius of the clusters, and in the following way:
\begin{observation}
\label{obs:rad-hopset-best}
For every $i \in \{1,\ldots, T\}$, $r_i \leq d/2R' \cdot (3+ 8/\epsilon)^{i}$. In particular, the radius of cluster in the final clustering $\mathcal{C}_T$ is $r_T \le d/2$.
\end{observation}
\begin{proof}
We set $r_{0}:=d/2R'$, and the rest follows by the same argument as in Observation \ref{obs:rad-spanner-best}.
\end{proof}
As in the proof of Lemma \ref{lem:3plusepshopset}, the size of a hopset for a fixed distance range is bounded by 
\begin{eqnarray*}
\sum_{i=0}^T |E(H_i)|&= &O(n^{1+1/k}+\log{k}\cdot n)~.
\end{eqnarray*}
Thus the expected size of the hopset is $O((n^{1+1/k}+\log{k}\cdot n)\cdot \log\Lambda)$. Furthermore, as in the proof of Lemma \ref{lem:3plusepshopset}, $\alpha=(3+1.125\cdot \epsilon,\beta)$ and $\beta=16\cdot R'$, by plugging $\epsilon'=8/9\cdot \epsilon$, Lemma \ref{lem:hopset-best} follows.

\end{document}